\documentclass[a4paper,fleqn]{cas-sc}
\usepackage[numbers, sort&compress]{natbib} 
\usepackage{subcaption}
\usepackage{algorithm}
\usepackage{amsthm}
\usepackage{algpseudocode}
\usepackage{tablefootnote}
\newtheorem{theorem}{Theorem}
\newdefinition{rmk}{Remark}
\newtheorem{lemma}[theorem]{Lemma}
\newproof{pf}{Proof}
\newdefinition{df}{Definition}

\begin{document}
\let\WriteBookmarks\relax
\def\floatpagepagefraction{1}
\def\textpagefraction{.001}

\title [mode = title]{A fast improved quasi-physical dynamic algorithm for efficient wireless coverage in convex polygonal regions}  

\author[1]{Zeping Yi}%[<options>]

\fnmark[1]
\ead{yzping@buaa.edu.cn}

\affiliation[1]{organization={School of Mathematical Sciences, Beihang University},
            city={Beijing},
            postcode={100191}, 
            country={China}}

\author[1]{Yongjun Wang}%[]

\cormark[1]
% Footnote of the second author
\fnmark[2]
\ead{wangyj@buaa.edu.cn}

\author[1]{Baoshan Wang}%[]

% Footnote of the second author
\fnmark[3]
\ead{bwang@buaa.edu.cn}

\author[1]{Jian Zhang}%[]
\fnmark[4]
\ead{sy2509220@buaa.edu.cn}
% Footnote of the second author

\author[1]{Songyi Liu}%[]

% Footnote of the second author
\fnmark[5]
\ead{liusongyi@buaa.edu.cn}

\begin{abstract}
Deploying wireless nodes to maximize coverage area within a given region is an important challenge in wireless sensor networks, UAV path planning, base station placement and other industrial fields. This practical problem can be mathematically equivalent to an optimal circle covering problem. Although theoretical optimal configurations exist for simple cases in mathematics, the NP-hard nature of this problem makes it computationally prohibitive for complex polygons with numerous nodes. Existing approaches are usually designed for regular domains, while those applicable to irregular polygons often suffer from poor initialization, excessive coverage overlap and failure to constrain nodes within the boundary, leading to low coverage efficiency and long runtime. To address these issues, we propose an improved quasi-physical dynamic algorithm (IQPD) for wireless node deployment in arbitrary convex polygons. Our contributions are threefold: (1) proposing a structure-preserving initialization that maps a hexagonal close packing pattern into the target polygon via scaling and affine transformation, ensuring near-optimal initial node distribution; (2) constructing a refined virtual force model by incorporating friction and a radius-expansion optimization mechanism to reduce coverage area overlap; (3) developing a boundary encircling strategy leveraging normal and tangential gradients to reposition nodes deployed outside boundaries after initial optimization. Extensive experimental results demonstrate that our method consistently outperforms other new metaheuristic algorithms across diverse convex polygon shapes, including randomly generated data and real-world scenarios. Our method achieves the highest coverage rate and node utilization rate among all compared algorithms, greatly improving wireless coverage efficiency. 
\end{abstract}

\maketitle
% Main text
\section{Introduction} 

Achieving maximum coverage through optimal deployment of wireless nodes is an important challenge in applications such as wireless sensor networks, UAV path planning, base station placement and other industrial fields. \cite{MOHAMMADI2026117912,LIU2021103019,YANG2025117303,YANG2025117108,Kumar2024,CHOWDHURY2021102660,biomimetics10110750,Qi2025,Shaikh2025,Yang2025,ZHAO2026117922,LI2026117477,S01102024, LI2025111431,WANG2018196,Tong2025}. From the mathematical perspective, the problem of optimizing wireless node deployment can be viewed as a circle covering optimization problem. Each wireless node can be regarded as the center of a circle which represents its coverage area, and the target service space is modeled as a given polygonal region. To formulate this problem precisely, we establish the following mathematical model. Given a polygonal region $P$, a set of congruent circles $\left\{ C_i \right\} _{i=1}^{n}$ of radius $r$, and their center coordinates \((x_i, y_i)\ (i = 1, \ldots, n)\), the coverage rate (CR)  is defined as 
\begin{equation}
CR=\text{Area}\left( P\cap \left( \bigcup_{i=1}^n{C_i} \right) \right) /\text{Area}\left( P \right).
\label{eq:CR}
\end{equation}
The mathematical formulation of the  optimal covering problem is as follows:
\begin{equation}
	\underset{\left( x_1,y_1 \right) ,\cdots ,\left( x_n,y_n \right)}{maximize} CR,
\end{equation}
\begin{equation}
	s.t.\quad C_i \cap P \neq \emptyset, \forall i.
\end{equation}

This optimization problem is closely related to the circle packing and covering problem \cite{Kravitz01031967,2011,2016,RYU2020125076,2021,2023,20241}. The covering problem typically aims to cover a given region completely using a set of congruent circles with overlap permitted. In contrast, the packing problem involves placing a set of circles (the radii are not necessarily equal) within a given container without any overlap. As early as the time of Archimedes, researchers began investigating optimal configuration of circles in the plane. In the 17th century, Kepler discussed the close packing of congruent circles on a plane in his work $On\ the\ Six-Cornered\ Snowflake$, and conjectured that hexagonal packing of spheres would be the densest configuration in three-dimensional space \cite{2011}. According to Kepler's conjecture, Fejes Tóth gave a rigorous proof \cite{toth1949dichteste} that the hexagonal arrangement is the optimal solution for circle packing on a plane (with a density of \(\frac{\pi}{2\sqrt{3}} \approx 0.9069\)) in 1949. This important theorem has provided valuable insights for future research on problems of circle packing and covering in irregular polygons \cite{2016,RYU2020125076,2021,2023,20241}.

Due to the strong correlation between optimal wireless node deployment and the circle packing and covering problem, existing methods for  solving the latter, particularly the quasi-physical approach \cite{WENQI2002195,Huang2011,ZHU2016506,HE201826}, can be adapted to the deployment problem. However, this work remains rarely explored in the existing literature. Furthermore, because of the NP-hard nature of the optimal deployment problem for wireless network nodes, metaheuristic algorithms \cite{MOHAMMADI2026117912,YANG2025117303,YANG2025117108,Kumar2024,CHOWDHURY2021102660,biomimetics10110750,Qi2025,Shaikh2025,Yang2025,ZHAO2026117922,LI2026117477,S01102024, LI2025111431,WANG2018196,Tong2025} have become the mainstream approach for solving it, such as improved bee foraging learning particle swarm optimization algorithm \cite{Qi2025}, improved chaotic grey wolf optimizer \cite{Shaikh2025}, quantum genetic algorithm \cite{LI2025111431}, etc. Researchers often draw inspiration from  the collective behavior of animal swarms to enhance the internal search mechanisms of optimization algorithms or to integrate novel optimization strategies, with the aim of improving convergence speed and coverage efficiency. Although these metaheuristic algorithms can achieve high coverage efficiency in regular polygons, they often encounter issues such as excessive overlap among coverage areas, nodes being deployed outside the target region and long run time when dealing with complex irregular polygons. 

To address these challenges, we propose an improved quasi-physical dynamic algorithm (IQPD) differing from current metaheuristic algorithms.
Our contributions are summarized as follows: 
\begin{itemize}
	\item  We propose a structure-preserving initialization strategy 
	that maps a hexagonal close packing pattern into the target polygon via scaling and affine transformation to obtain a high-quality initial node distribution.
	\item We improve the quasi-physical simulation method proposed by Huang et al. \cite{WENQI2002195,Huang2011,ZHU2016506,HE201826} for solving circle packing problems. We mainly refine the virtual force model by incorporating friction and a radius-expansion optimization mechanism to reduce coverage area overlap. 
	\item We develop a boundary encircling strategy based on normal and tangential gradients to reposition nodes deployed outside the target region after initial optimization.
	\item We analyze computational complexity and convergence of our algorithm theoretically and carry out extensive experiments including ablation studies. Experimental results demonstrate that IQPD algorithm greatly improves coverage rate and node utilization rate in complex convex polygons compared to existing metaheuristic algorithms or typical quasi-physical approaches.
\end{itemize}

 The remainder of this paper is structured as follows: Section 2 reviews work related to this paper. Section 3 introduces fundamental concepts in computational geometry that underpin the proposed algorithm. Section 4 elaborates on the improved quasi-physical dynamic algorithm, including theoretical computational complexity and convergence analysis. Section 5 presents comparative experimental results. Section 6 concludes the paper and suggests future work.

\section{Related Work}
\subsection{Quasi-physical approach}

The quasi-physical approach was first proposed for solving packing problems by Huang et al. \cite{WENQI2002195} in 2002. In this approach, the circles to be packed are modeled as smooth elastic balls, and the enclosing boundary is treated as a rigid container. By hypothetically squeezing all the "balls" into the rigid container, elastic forces arise between balls, or between balls and the container wall, driving a series of movements. The result of these movements is expected to be a configuration in which each ball settles into an appropriate position, allowing all balls to be packed into the rigid container without overlap.
Compared with traditional man-made algorithms, its calculation  can
 execute more quickly and efficiently. Huang et al. \cite{Huang2011} refined the previously proposed quasi-physical algorithm by incorporating a series of new strategies in 2011. Refined approach comprises three key steps: (1) Quasi-physical descent procedure, where elastic forces are employed to gradually eliminate circle overlaps and guide all circles to a locally optimal configuration; (2) Quasi-physical basin-hopping procedure, which introduces attractive and non-contact repulsive force to make circles jump out of the local optimum trap and get a better solution in the entire solution space;
 (3) Container radius $R_0$ adjustment procedure, where a binary search method is applied to adjust circular container radius, ensuring the container is as small as possible while accommodating all circles without overlap. Across test instances with the number of circles ranging from 1 to 150, the algorithm improved the best-known packing solutions for 37 cases and reproduced the best-known results in 113 others. By integrating quasi-physical methods with global optimization, the approach effectively avoids local optima and enhances solution quality, highlighting the promise of quasi-physical algorithms in optimization. However, a limitation is its dependence on initial configuration, which can significantly influence convergence speed and optimal results. Therefore, Zhu \cite{ZHU2016506} introduced a quasi-human seniority-order (QS) algorithm to optimize initial configuration for circle packing. Its core idea is adopting a largest-first strategy: placing larger circles first and filling the interstitial gaps with smaller ones to maximize space utilization. Then the QS algorithm is combined with the quasi-physical approach. Compared with algorithms employing random initial configuration, experimental results demonstrate that the QS algorithm significantly reduces the number of iterations and computation time, while also mitigating the risk of premature convergence or infinite loop. Moreover, He et al. \cite{HE201826} proposed an efficient Quasi-Physical Quasi-Human (QPQH) hybrid algorithm for equal circle packing in 2018. Combining a modified Broyden-Fletcher-Goldfarb-Shanno (BFGS) algorithm with a quasi-human basin-hopping strategy, QPQH leverages neighborhood information of each circle to speed up computation. It achieved state-of-the-art results on multiple benchmarks, demonstrating high efficiency and robustness for large-scale problems. 

\subsection{Metaheuristic algorithms}
In recent years, metaheuristic algorithms have become the mainstream approach for solving the optimal coverage problem in wireless networks. A variety of novel metaheuristic algorithms have been proposed for wireless sensor networks, including improved swarm intelligence algorithms \cite{MOHAMMADI2026117912,LIU2021103019,biomimetics10110750,Qi2025,Shaikh2025,Yang2025,ZHAO2026117922,LI2026117477,S01102024, LI2025111431,WANG2018196,Tong2025} and heuristic algorithms that incorporate classical geometric structures \cite{LIU2021103019,CHOWDHURY2021102660}, machine learning \cite{CHOWDHURY2021102660,WANG2018196,XIA2025116791}, quantum computation \cite{LI2025111431}.

In 2021, Chowdhury et al.\cite{CHOWDHURY2021102660} proposed a Voronoi-Glowworm Swarm Optimization-K-means (VGSOK) algorithm, which integrates Glowworm Swarm Optimization, K-means algorithm, and Voronoi diagram. VGSOK algorithm partitions circle centers (sensor nodes) into dense regions via K-means, regarding the cluster centers as generators to construct a Voronoi diagram. It then models each circle center as a “glowworm” with an associated luciferin intensity, driving them to converge toward their corresponding Voronoi generators.  VGSOK effectively eliminates coverage holes caused by random initial deployment, improves coverage rate by 3–10\%, reduces practical energy consumption, and meets multi‑objective optimization requirements. In the same year, Liu et al.\cite{LIU2021103019} also presented a hybrid algorithm combining Particle Swarm Optimization with Voronoi diagram (PSOVD). It decouples the deployment problem into two dimensions. In the vertical dimension, the optimal altitude for maximizing communication coverage is derived mathematically. In the horizontal dimension, each agent employs a PSO-based local deployment algorithm to find a better position within its Voronoi cell, thereby greatly improving coverage rate of agents (sensor nodes). 

As application scenarios demand greater dynamic adaptability and convergence speed, research gradually focuses on improving the internal search mechanisms of metaheuristics. To address issues such as local optima entrapment (leading to "coverage holes") and slow convergence in large-scale conditions, Wang et al. \cite{biomimetics10110750} proposed a Multi-strategy Improved Flamingo Search Algorithm (MIFSA) in 2025. MIFSA incorporates several key strategies: (1) an elite opposition-based learning strategy during initialization to broaden search scope; (2) a  staged step-size control strategy combined with a cosine variation factor, which helps to escape from local optima entrapment; (3) an adaptive Lévy flight mechanism in the final stage to refine solution quality via stochastic perturbation. Extensive experimental results within a 50×50 $m^2$ square region demonstrated that MIFSA outperformed the original flamingo search algorithm (FSA). It achieved coverage rates 7.48\% and 5.68\% higher than those of the original FSA after 100 and 200 iterations, respectively, along with more uniform sensor node distribution and fewer coverage holes. 

To further tackle inefficiencies and coverage holes caused by uneven circle distribution, Yang et al. \cite{Yang2025} proposed an Improved Cuckoo Search algorithm with Multiple Strategies (ICS-MS) in 2025. ICS-MS employs several strategies: (1) a phased dimension-by-dimension update strategy (full-dimensional early, dimension-wise later) to reduce computational complexity; (2) an adaptive discovery probability mechanism, dynamically adjusting discovery probability based on population fitness variance; (3) bidirectional search and global-best guidance to accelerate convergence while maintaining population diversity; (4) opposition-based learning on elite individuals for subsequent deep local search. Compared with other metaheuristic algorithms under various numbers of sensor nodes (20/30) and iterations (200/1000), ICS-MS improved coverage rate by 2.32\% to 22.21\% and obtained more uniform node distribution.

Meanwhile, in order to improve both convergence speed and stability of metaheuristic algorithms, Shaikh et al.\cite{Shaikh2025} introduced an Improved Chaotic Grey Wolf Optimizer (ICGWO) in 2025. ICGWO integrates Gaussian chaotic mapping into the position-updating process and leverage chaotic sequences to diversify search paths. Unlike comparative algorithms (e.g., Particle Swarm Optimization or Ant Colony Optimization) that require intricate parameter tuning, ICGWO maintains a balance between global exploration and local exploitation with minimal parametric overhead. When compared on the same number of sensor nodes, ICGWO surpasses existing counterparts by 2.18–16.41\% in coverage rate, with pronounced advantages in large-scale deployments. The algorithm demonstrates high scalability and robustness, making it suitable for complex real-world scenarios.

Li et al.~\cite{LI2025111431}  also proposed a dynamically adjusted quantum genetic algorithm (DAQGA) in 2025. This approach has three improvements: (1) proposing a niching population initialization strategy based on a normal distribution to balance search scope and local search precision; (2) introducing a dynamic selection strategy for quantum rotation gates and designing a rotation angle adjustment formula based on fitness differences, thereby improving convergence speed; (3) incorporating a quantum catastrophe adjustment strategy to escape local optima. Compared with four traditional heuristic algorithms, DAQGA achieves a coverage improvement of 5.82\% to 18.24\% under the same conditions. Moreover, it exhibits faster convergence, stronger stability, and fewer coverage holes.

In 2026, a new method called Fuzzy Sign-aware Influence Maximization using an adaptive Improved Grey Wolf Optimizer (FSIMI-GWO) \cite{MOHAMMADI2026117912} was proposed by Mohammadi et al., which is more suitable for complex real-world applications. FSIMI-GWO introduces a fitness function that incorporates fuzzy logic to effectively handle diverse user relationships through appropriate weight assignments, where each user is treated as a sensor node. This approach obtains superior  the influence propagation in social networks compared with several new introduced methods.

\section{Preliminaries}
For simplicity, we model the wireless node deployment as an optimal circle coverage problem in this paper. The circle center represents a sensor node, and the circle itself represents its coverage region. The convex polygon corresponds to the target region to be covered. 

\subsection{Geometric computation}

\begin{df}[Convex Set]
	Given a point set \( S \), if for any \( \forall x, y \in S \), \( (1-\lambda)x + \lambda y \in S \) holds for all \( \forall \lambda \in [0,1] \), then \( S \) is a convex set. Geometrically, this intuitively means that a convex set contains the line segment connecting any two points within it.
\end{df}

\begin{df}[Convex Hull]
	The convex hull \(Conv(S)\) of a set S is the smallest convex set containing $S$, that is, the intersection of all convex sets that contain $S$.
\end{df}

A convex polygon in the plane can be defined as the boundary of the convex hull of a finite set of points~$S$. Such a polygon is typically represented by starting at an arbitrary vertex of the convex hull and connecting all vertices sequentially in clockwise (or counterclockwise) order to form a closed region. Consequently, the problem of constructing a convex polygon reduces to computing the convex hull of its given point set~$S$. Well-known algorithms for computing the convex hull are introduced in \cite{GRAHAM1972132,JARVIS197318,ANDREW1979216} and   polygon convexity can be verified as described in \cite{SCHORN19947}.

\begin{lemma}[\cite{berg2008}]
 Given any planar point set containing n points, its convex hull can be constructed in \(O(n\log n)\) time.
\end{lemma}

\begin{lemma}[\cite{SCHORN19947}]
 For each vertex $P_i = (x_i, y_i)$ of the polygon, compute the cross product of vectors formed by three consecutive vertices $P_i \to P_{i+1} \to P_{i+2}$ sequentially. If all cross products have the same sign (either all positive or all negative), then the polygon is convex. Specifically, let $\overrightarrow{P_iP_{i+1}} = (x_{i+1} - x_i, y_{i+1} - y_i)$
and $\overrightarrow{P_{i+1}P_{i+2}} = (x_{i+2} - x_{i+1}, y_{i+2} - y_{i+1})$, then the cross product $\overrightarrow{P_iP_{i+1}} \times \overrightarrow{P_{i+1}P_{i+2}} \quad (i = 1, \dots, n-2)$
must be consistently signed.
\end{lemma}

To determine whether  a circle center lies inside a convex polygon, it is necessary to introduce the right normal vector and right distance, as shown in Fig.~\ref{Fig:1}. 

\begin{df}[Right Normal Vector] 
	Given a directed segment $\overrightarrow{P_1P_2}$, its right normal vector is the unit normal vector that points to the right-hand side when traversing from $P_1$ to $P_2$. Formally, for the direction vector $\vec{v} = (v_x, v_y) = \overrightarrow{P_1P_2}$, the right normal vector is given by
	$\vec{n}_{\mathrm{right}} = (v_y, -v_x).$
\end{df}

\begin{df}[Right Distance] 
	The right distance from a point $A$ to the directed segment $\overrightarrow{P_1P_2}$ is defined to be positive if $A$ lies on the right-hand side of the segment, negative if on the left-hand side, and zero if on the line containing the segment. It is computed by the formula:
	$d_{\mathrm{P_1P_2}}(A) =  \frac{\overrightarrow{P_1A} \cdot \vec{n}_{\mathrm{right}}}{\|\vec{n}_{right}\|}$
	, where $\vec{n}_{\mathrm{right}}$ is the right normal vector as defined above.
\end{df}

\begin{figure}
	\centering
	 	\begin{subfigure}[b]{0.3\textwidth}
		\centering
		\includegraphics[width=\textwidth]{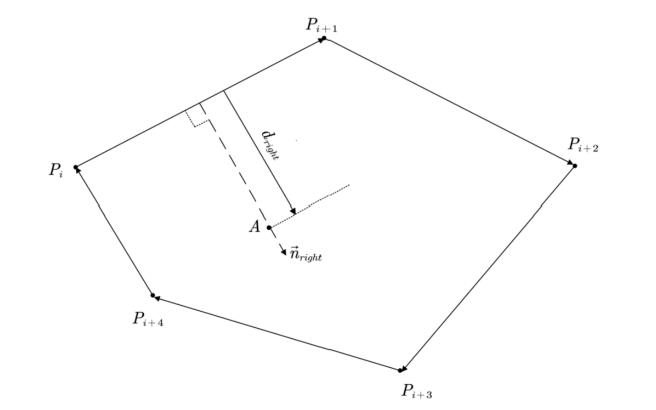}
		\caption{}
	\end{subfigure}
	\begin{subfigure}[b]{0.3\textwidth}
		\includegraphics[width=\textwidth]{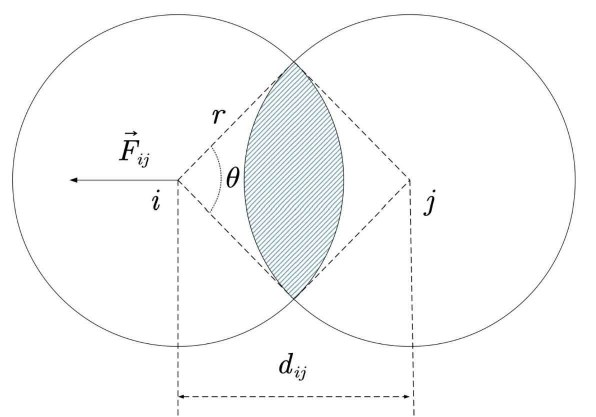}
		\caption{ }
		\label{fig:1b}
	\end{subfigure}
	\begin{subfigure}[b]{0.3\textwidth}
		\includegraphics[width=\textwidth]{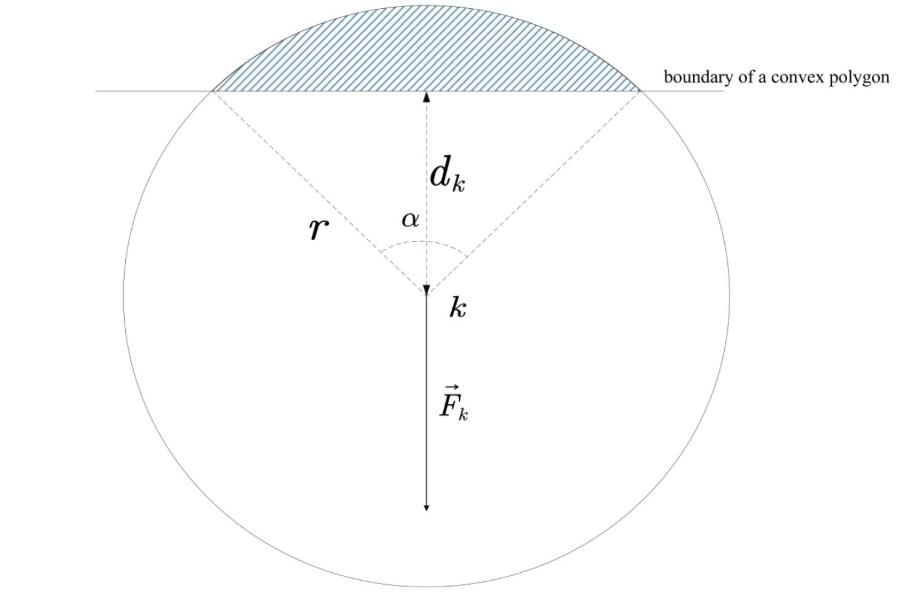}
		\caption{}
		\label{fig:1c}
	\end{subfigure}
	\caption{(a) A convex polygon composed of clockwise-oriented directed line segments; (b) Force between circles; (c) Force between a circle and the boundary}
	\label{Fig:1}
\end{figure}

\subsection{Typical quasi-physical model}

The typical quasi-physical model was initially proposed by Huang et al. \cite{WENQI2002195,Huang2011} for solving the equal-circle packing problem. The core idea is to model the circles as elastic balls confined within the convex polygon boundary, which acts as a container wall. By simulating the collision and squeezing motions among these elastic balls, they gradually spread out within the polygon, leading to an efficient packing configuration. This thought can be adapted to the circle covering optimization problem in convex polygons. 

\begin{df}[Configuration] 
	\label{def:5}
	The vector of circle centers at any time, denoted as $\boldsymbol{X}=\left[x_1,y_1,x_2,y_2,\cdots ,x_n,y_n\right]$, defines the configuration comprising a convex polygon and circles.
\end{df}

The system composed of a convex polygon and circles involves two types of elastic forces: circle–circle and circle–polygon boundary. Let \(\vec{F}_{ij}\) denote the elastic force exerted by circle \(j\) on circle \(i\), directed from \(j\) toward \(i\), and let \(\vec{d}_{ij}\) represent the displacement vector from the center \(j\) to the center \(i\). Similarly, let \(\vec{F}_k\) denote the elastic force exerted by the polygon boundary on circle \(k\), and \(\vec{d}_k\) the displacement vector from the polygon boundary to the center \(k\), as illustrated in Fig.\ref{Fig:1}. The two elastic forces are computed from the overlap depth in the typical model.

\subsection{Voronoi Diagram}

A Voronoi diagram is a partition of a planar region that divides the plane into \( n \) cells. Given a set of points \( S = \{s_1, s_2, \cdots, s_n\} \), these \( n \) distinct points are referred to as generators, with each generator corresponding to a partitioned cell. For example, the cell corresponding to generator \( s_i \) is denoted $\mathcal{V}(s_i)$, and the Voronoi diagram associated with point set \( S \) is denoted \(Vor(S)\).

\begin{lemma}[	\cite{berg2008}]
 A point $p \in \mathcal{V}(s_i)$ if and only if \( \delta(p, s_i) < \delta(p, s_j) \) for all \( s_j \in S (j \neq i) \), where \( \delta(p, s_i) \) represents the Euclidean distance between \( p \) and \( s_i \). 
\end{lemma}

For any two distinct points \( s_i \) and \( s_j \) in the plane, the perpendicular bisector of the line segment \( \overline{s_i s_j} \) partitions the plane into two half-planes. The open half-plane containing \( s_i \) is denoted by \( h(s_i, s_j) \), while the open half-plane containing \( s_j \) is denoted by \( h(s_j, s_i) \). It follows that the Voronoi cell of \( s_i \) is given by the intersection of all such half-planes induced by other generators:
$\mathcal{V}(s_i) = \bigcap_{\substack{1 \le j \le n \\ j \ne i}} h(s_i, s_j).$
The Voronoi diagram of the set \( S \) is then the union of all Voronoi cells:
$Vor(S) = \bigcup_{i=1}^n \mathcal{V}(s_i)$, as shown in Fig.~\ref{Fig:2}.
\begin{figure}
	\centering
	\begin{subfigure}[b]{0.25\textwidth}
		\centering
		\includegraphics[width=\textwidth]{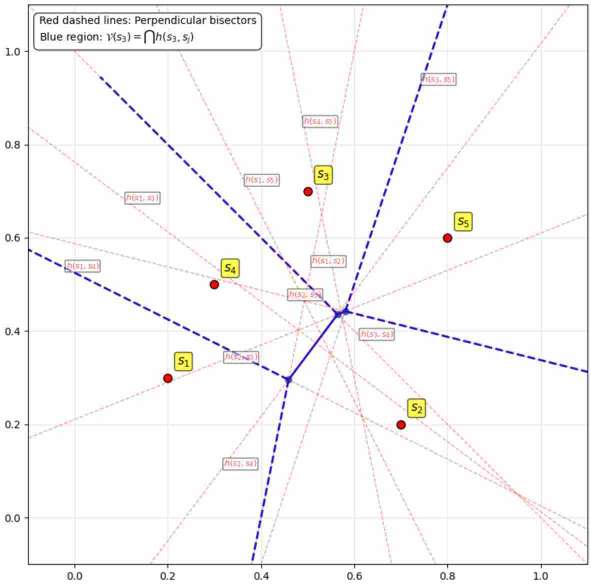}
		\caption{}
	\end{subfigure}
	\begin{subfigure}[b]{0.255\textwidth}
		\centering
		\includegraphics[width=\textwidth]{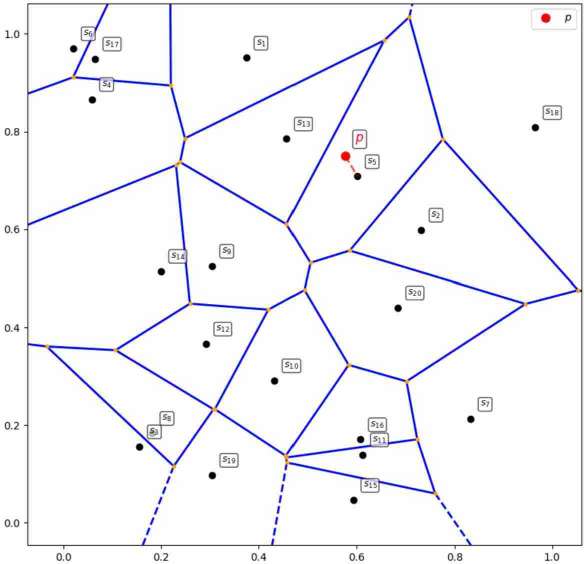}
		\caption{}
	\end{subfigure}
	\caption{(a) Intersection of all half-planes; (b) Voronoi diagram in $\mathbb{R}^2$}
	\label{Fig:2}
\end{figure}

\section{Improved quasi-physical dynamic algorithm }
 \subsection{Structure-preserving initialization algorithm}

The initial configuration of circles significantly influences the quality of the final near-optimal solution and computation time. Therefore, a well-chosen initial configuration is essential for subsequent high-quality optimization iterations. Inspired by the close packing principle of regular hexagons, we propose a structure-preserving initialization strategy via scaling and affine transformation.

First, a series of concentric regular hexagons is generated outward from the center on a base plane. Using appropriate transformation parameters, this configuration is projected into the target convex polygon, thereby transferring the close-packed circle pattern into the polygonal domain as shown in Fig.~\ref{Fig:3}.

This approach ensures that the circles in the interior region of the polygon are already in a near-optimal close-packed arrangement, requiring no further iterative optimization. Consequently, the subsequent optimization process can focus exclusively on adjusting the circles near the boundary. This strategy not only preserves the optimal packing structure across most of the area but also significantly reduces the computational cost of iterative optimization, leading to a substantial improvement in overall algorithmic efficiency.

\begin{figure}
	\centering
	\begin{subfigure}[b]{0.35\textwidth}
		\centering
		\includegraphics[width=\textwidth]{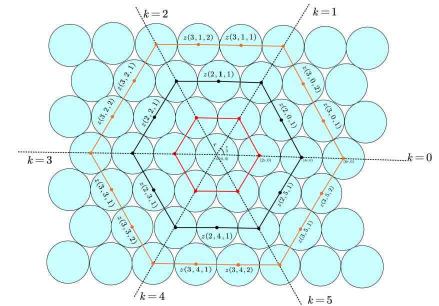}
		\caption{}
		\label{fig:1}
	\end{subfigure}
	\begin{subfigure}[b]{0.6\textwidth}
		\centering
		\includegraphics[width=\textwidth]{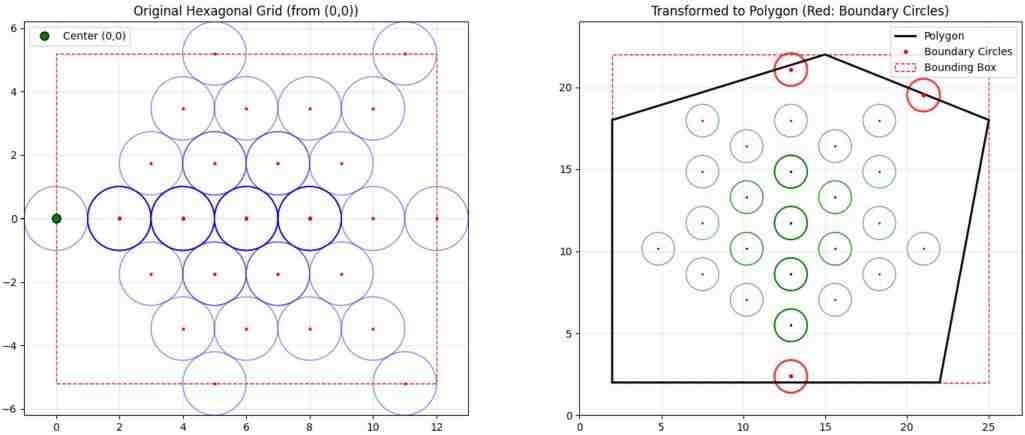}
		\caption{}
		\label{fig:2}
	\end{subfigure}
	\caption{(a) Regular hexagonal close packing of circles generated from the center outward; (b) The projection of densely packed circles into a convex polygonal domain through scaling and affine transformation}
	\label{Fig:3}
\end{figure}

\begin{df} [Circular Sector]
	A circular sector is defined as a planar region bounded by two rays emanating from the circle's center. According to Fig.~\ref{fig:1},We denote the region between $k=i$ and $k=i+1$ as the i-th sector.
\end{df}

Assume that the 0-th layer consists of a single circle with its center at \((0, 0)\). The first layer is composed of six circles generated around the initial circle along six respective directions, as represented in Fig.~\ref{Fig:3}. This pattern continues recursively: for the \(l\)-th layer (\(l \geq 1\)), \(l\) circles are generated along each of the six directions (with an angular interval of \(\frac{\pi}{3}\) radians between adjacent directions). The distance from the center of each circle in this layer to the origin is \(2lr\), and the distance between the centers of any two adjacent circles is \(2r\). To facilitate a concise formulation, complex numbers are employed for the circle center coordinates. Given a rotation angle of \(\theta = k \cdot \frac{\pi}{3}\) (where \(k = 0, 1, \dots, 5\)), the complex coordinate of the \(m\)-th circle in the \(k\)-th sector of the \(l\)-th layer is
\begin{equation}
	z\left( l,k,m \right) =2r\cdot l\cdot e^{i\cdot k\cdot \frac{\pi}{3}}+2r\cdot m\cdot e^{i\cdot \left( k+2 \right) \cdot \frac{\pi}{3}},k=0,\cdots ,5,m=0,\cdots ,l-1.
	\label{eq:11}
\end{equation}

\begin{df}[Minimum Bounding Rectangle, MBR]
	The Minimum Bounding Rectangle of a polygon is defined as the rectangle of the smallest area that can fully contain the polygon. For any convex polygon, the orientation of its MBR is commonly used to determine the polygon's principal direction.
\end{df}

Suppose the vertex coordinates of the minimum bounding rectangle are denoted as \(P_1=(x_1,y_1), P_2=(x_2,y_2), P_3=(x_3,y_3), P_4=(x_4,y_4)\) (arranged in a clockwise order). Without loss of generality, let \(P_1P_2\) be the longest side of the rectangle. Then, the direction vector of the rectangle is given by \(\vec{d} = (x_2 - x_1, y_2 - y_1)\), and the corresponding rotation angle is calculated as \(\theta = \arctan\left(\frac{y_2 - y_1}{x_2 - x_1}\right)\). For instance, consider an arbitrary random convex polygon \(P = \{(0.5, 3), (2, 1), (5, 2), (8, 4), (7, 7), (4, 8), (1, 5)\}\). Its principal direction is illustrated in Fig.\ref{figure:7}. Based on the orientation of the minimum bounding rectangle, we propose a structure-preserving initialization algorithm for hexagonal close-packed circles. The main steps are as follows:

\begin{figure}
	\centering
	\includegraphics[width=0.6\textwidth]{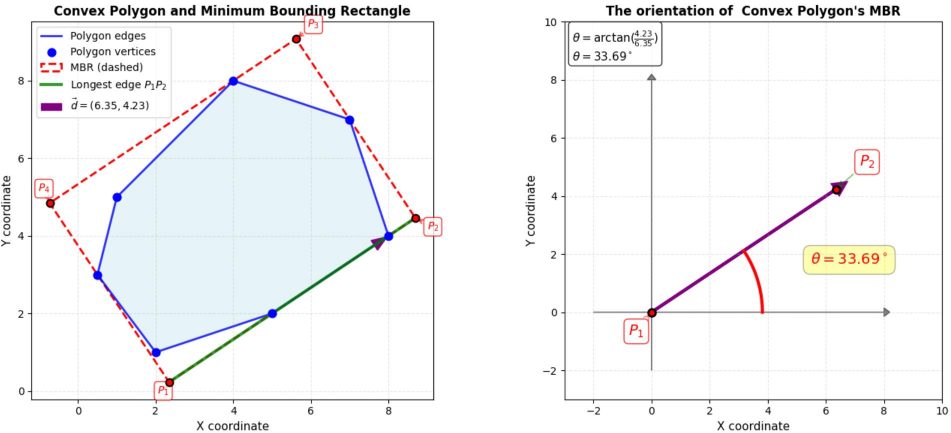} % 图片宽度撑满子图宽度
	\caption{Convex Polygon $P$ and its MBR}
	\label{figure:7}
\end{figure}

\textbf{Step 1:} Compute the MBR of an arbitrary convex polygon using the rotating calipers algorithm \cite{carlos_2020};

\textbf{Step 2:} Calculate the direction vector $\vec{d}$ and the rotation angle $\theta$ of the rectangle;

\textbf{Step 3:} Align the orientation of the hexagonal lattice with the principal direction of the polygon, as follows:
\begin{equation}
	z_{rolated}\left( l,k,m \right) =z\left( l,k,m \right) \cdot e^{i\theta};
\end{equation}

\textbf{Step 4:} Calculate the scaling factor $\beta$ to scale the hexagonal lattice to match the polygon's size:
\begin{equation}
	\left\{ \begin{array}{l}
		\beta =\min \left( \frac{w_p-2r}{w_h},\frac{h_p-2r}{h_h} \right) \cdot \alpha\\
		z_{scaled}\left( l,k,m \right) =\beta \cdot z_{rotated}\left( l,k,m \right)\\
	\end{array} \right. ,
	\label{eq:6}
\end{equation}
Here, $w_p$ and $h_p$ denote the length and width of the polygon's minimum bounding rectangle, respectively, while $w_h$ and $h_h$ represent the length and width of the bounding box of the hexagonal lattice as shown in Fig.~\ref{fig:2}. $r$ is the radius of the circles. To prevent the transformed circles from overflowing the polygon boundary, we define $\alpha \in (0.9, 0.95)$ as a safety coefficient. The terms $\frac{w_p-2r}{w_h}$ and $\frac{h_p-2r}{h_h}$ represent the ratios of the available space within the polygon to the original sizes of the hexagonal lattice along the length and width directions, respectively. The safety coefficient $\alpha$ is designed to reserve adjustment space for subsequent optimization iterations.

\textbf{Step 5:} Place the transformed hexagonal lattice within the polygon by aligning the lattice center with the polygon centroid $\boldsymbol{c}_p$, resulting in a roughly centered initialization as illustrated in Fig.~\ref{fig:2}:
\begin{equation}
	\left\{ \begin{array}{l}
		\boldsymbol{c}_p=\frac{1}{N}\sum_{i=1}^N{\boldsymbol{v}_i}=\left( \frac{1}{N}\sum{x_i},\frac{1}{N}\sum{y_i} \right)\\
		z_{final}\left( l,k,m \right) =z_{scaled}\left( l,k,m \right) +\boldsymbol{c}_p\\
	\end{array} \right. ,
\end{equation}
where $\boldsymbol{v}_i = (x_i, y_i)$ denotes the polygon vertices and $N$ is the number of vertices.

\textbf{Step 6:} Perform boundary detection on the affine-transformed circles and filter out those exceeding the polygon boundary using the boundary-aware filtering algorithm \cite{xianguang_2024}. Subsequently, apply the corner-occupancy principle \cite{WU2002341} to insert circles near the convex polygon interior.

The initialization algorithm enables the circles to maintain hexagonal close-packing characteristics within the polygon. The distance between adjacent circle centers becomes \(2\beta r\), while hexagonal symmetry is preserved. When the number of circles is small, the initial distribution corresponds to the optimal configuration, thus transforming the maximum area coverage problem into a traditional packing problem. Even with a larger number of circles, the algorithm establishes a high-quality starting point for subsequent iterative optimization.

\begin{algorithm}[H]
	\caption{Structure-preserving initialization algorithm}
	\label{alg:1}
	\begin{algorithmic}[1]
		\Require Convex polygon $P$, circles'number $N$, radius $r$, safety coefficient $\alpha \in (0.9, 0.95)$
		\Ensure Initial circle positions $\mathcal{L}_{\text{final}}=\boldsymbol{C}^{(0)} = \{\boldsymbol{c}_1^{(0)}, \dots, \boldsymbol{c}_N^{(0)}\}$
		
		\State \textbf{Step 1: Compute minimum bounding rectangle}
		\State $\text{Rect}, \vec{d}, \theta \gets \text{RotatingCalipers}(P)$ 
		\State $w_p, h_p \gets \text{width}(\text{Rect}), \text{height}(\text{Rect})$
		
		\State \textbf{Step 2: Generate base hexagonal lattice}
		\State $\mathcal{L}_{\text{base}} \gets \text{GenerateHexagonalLattice}(N)$ according to Eq.\eqref{eq:11}
		\State $w_h, h_h \gets \text{BoundingBoxDimensions}(\mathcal{L}_{\text{base}})$
		
		\State \textbf{Step 3: Rotate lattice}
		\State $\mathcal{L}_{\text{rot}} \gets \{ z \cdot e^{i\theta} \mid z \in \mathcal{L}_{\text{base}} \}$
		
		\State \textbf{Step 4: Scale lattice}
		\State $\beta \gets \min\left( \frac{w_p - 2r}{w_h}, \frac{h_p - 2r}{h_h} \right) \cdot \alpha$
		\State $\mathcal{L}_{\text{scaled}} \gets \{ \beta \cdot z \mid z \in \mathcal{L}_{\text{rot}} \}$
		\State \textbf{Step 5: Align lattice center with polygon centroid}
		\State $\boldsymbol{c}_p \gets \frac{1}{|V(P)|} \sum_{\boldsymbol{v} \in V(P)} \boldsymbol{v}$ 
		\State $\mathcal{L}_{\text{trans}} \gets \{ z + \boldsymbol{c}_p \mid z \in \mathcal{L}_{\text{scaled}} \}$
		\State \textbf{Step 6: Filter and optimize}
		\State $\mathcal{L}_{\text{filtered}} \gets \text{BoundaryAwareFilter}(\mathcal{L}_{\text{trans}}, P)$ 
		\State $\mathcal{L}_{\text{final}} \gets \text{CornerOccupancyInsertion}(\mathcal{L}_{\text{filtered}}, P)$ 
		\State \Return $\mathcal{L}_{\text{final}}$
	\end{algorithmic}
\end{algorithm}

\subsection{Virtual-force and radius expansion algorithm}
After initializing the coordinates of all circle centers, we perform dynamic modeling of the system consisting of all circles and the convex polygon.  Assume each circle \( C_i \) is a dynamic object with three attributes: position \( \vec{x}_i = (x_i, y_i) \), velocity \( \vec{v}_i = (v_{ix}, v_{iy}) \), and the resultant external force $\vec{F}_{\text{total},i}=(\vec{F}_{\text{total},ix},\vec{F}_{\text{total},iy})$. 

Different from typical quasi-physical models, the magnitude of the virtual force is represented by the intersection area rather than the overlap depth in order to achieve greater accuracy, as shown in Eq.\eqref{eq:8} and Eq.\eqref{eq:9}:
\begin{equation}
\left\{ \begin{array}{l}
	\theta =2\arccos \left( \frac{d_{ij}/2}{r} \right)\\
	S_a=2\left( \frac{1}{2}r^2\theta -2r\sin \frac{\theta}{2}r\cos \frac{\theta}{2} \right) =r^2\left( \theta -\sin \theta \right)\\
	\vec{F}_{ij}=S_a\cdot \frac{\vec{d}_{ij}}{\left| \vec{d}_{ij} \right|}\\
\end{array} \right. ,
\label{eq:8}
\end{equation}
\begin{equation}
\left\{ \begin{array}{l}
	\alpha =2\arccos \left( \frac{d_k}{r} \right)\\
	S_b=\frac{1}{2}r^2\alpha -2\sqrt{r^2-d_{k}^{2}}d_k\cdot \frac{1}{2}=\frac{1}{2}r^2\alpha -d_k\sqrt{r^2-d_{k}^{2}}\\
	\vec{F}_k=S_b\cdot \frac{\vec{d}_k}{\left| \vec{d}_k \right|}\\
\end{array} \right.. 
\label{eq:9}
\end{equation}
\( S_a \) and \( S_b \) denote the overlapping areas in Fig.~\ref{fig:1b} and Fig.~\ref{fig:1c}. The direction of the elastic force between two circles always points from the center of the exerting circle to the center of the receiving circle, while the direction of the elastic force between a circle and the polygon boundary is given by the rightward normal vector of the directed edge. 

 Additionally, friction \( \vec{F}_{\text{friction}} \)  is introduced for each circle to prevent circles from moving indefinitely caused by inertia. Let the viscous coefficient be \( \mu \) and the friction acts opposite to the velocity direction and is proportional to the magnitude of the velocity.  
Thus, the resultant external force $\vec{F}_{\text{total,i}}$ on the circle $C_i$ is
\begin{equation}
	\left\{ \begin{array}{l}
		\vec{F}_{\text{friction}}=-\mu\cdot \vec{v}\\
		\vec{F}_{\text{total,i}}=\sum_{j=1,j\ne i}^n{\vec{F}_{ij}}+\sum_{k=1}^N{\vec{F}_k}+\vec{F}_{\text{friction}}\\
	\end{array} \right. .
	\label{eq:10}
\end{equation}

To simplify the computation, the mass of each circle can be set to 1 and the time step to $\Delta t$. According to Newton's second law, the update equations for the acceleration, velocity, and position of each circle are then given by:
\begin{equation}
	\left\{ \begin{array}{l}
		\vec{a}_i=\vec{F}_{total,i}\\
		\vec{v}_{i}^{\left( t+\varDelta t \right)}=\vec{v}_{i}^{\left( t \right)}+\vec{a}_i\varDelta t\\
		\vec{x}_{i}^{\left( t+\varDelta t \right)}=\vec{x}_{i}^{\left( t \right)}+\vec{v}_{i}^{\left( t \right)}\varDelta t\\
	\end{array} \right. .
	\label{eq:11}
\end{equation}
In the Python program, updates for velocity and position are performed separately along the $x-$ and $y-$ axes, i.e,
\begin{equation}
	\left\{ \begin{array}{l}
		v_{x}^{\left( t+\varDelta t \right)}=v_{x}^{\left( t \right)}+a_x\varDelta t\\
		v_{y}^{\left( t+\varDelta t \right)}=v_{y}^{\left( t \right)}+a_y\varDelta t\\
	\end{array} \right. ,
	\label{eq:12}
\end{equation}

\begin{equation}
	\left\{ \begin{array}{l}
		x^{\left( t+\varDelta t \right)}=x^{\left( t \right)}+v_x\varDelta t\\
		y^{\left( t+\varDelta t \right)}=y^{\left( t \right)}+v_y\varDelta t\\
	\end{array} \right. .
	\label{eq:13}
\end{equation}

After dynamical modeling, we employ a radius expansion algorithm to optimize the distribution of circles. This algorithm alternates between radius expansion and virtual-force balancing to drive the system from an initial relaxed state toward the target configuration. Specifically, the algorithm starts with a set of radii smaller than the target values. In each iteration, it first increases all radii by a predefined step size $\Delta r_k$, followed by a virtual-force simulation that brings the circles into an approximate mechanical equilibrium under the current radii. This “expansion–equilibration” cycle continues until the radii recover the target values, thereby yielding a uniform and minimally overlapping layout. The iteration is terminated based on the utilization rate (UR).  The system is considered to have converged stably when the change in UR between successive iterations falls below a given threshold.

\begin{df}[Utilization rate, UR]
	$UR$ evaluates overall effective utilization of all circles $\left\{ C_i \right\} _{i=1}^{n}$ , defined as follows:
	\begin{equation}
	UR = \frac{Area\left( \bigcup_{i=1}^{n} (C_i \cap P) \right)}{n \pi r^2}.
	\end{equation}
	\label{df:ur}
\end{df}

\begin{df}[Expansion System State]
	For a given radius \(r\), the set of all circles' positions, velocities, and total external forces is defined as the expansion system state, denoted by
\begin{equation}
	S(r) = \left\{ \boldsymbol{p}_i(r), \boldsymbol{v}_i(r), \boldsymbol{F}_{\text{total},i}(r) \right\}_{i=1}^{n}.
\end{equation}
\(\boldsymbol{p}_i(r)\) represents the position of the \(i\)-th circle, \(\boldsymbol{v}_i(r)\) its velocity, and \(\boldsymbol{F}_{\text{total},i}(r)\) the total external force acting on it.
\end{df}

The iterative process of the expansion system can be formulated as
\begin{equation}
	\left\{ \begin{array}{l}
		S\left( r_{k+1} \right) =\mathcal{O}\left( S\left( r_k \right) ,r_{k+1} \right)\\
		r_{k+1}=r_k+\Delta r_k\\
	\end{array} \right. ,
\end{equation}
where \(\mathcal{O}\) represents the optimization operator that transforms the system from the state at radius \(r_k\) to the state at radius \(r_{k+1}\). We employ an adaptive step size to update the circle radii, using the utilization rate as the criterion for radius adjustment, as follows:
\begin{equation}
	\Delta r_{k+1}=\begin{cases}
		2\Delta r_k&		\text{if\,\,}U\left( r_k+\Delta r_k \right) >U_{\text{th}}\,\,\text{and\,\,}C\geq C_{\text{th}}\\
		\frac{1}{2}\Delta r_k&		\text{if\,\,}U\left( r_k+\Delta r_k \right) \leq U_{\text{th}}\\
		\Delta r_k&		\text{otherwise}\\
	\end{cases}
\end{equation}
Here, \( U(r) \) denotes the utilization rate of the circles at radius \( r \), and \( U_{\text{th}} \) is the corresponding utilization threshold. The counter \( C \) records the number of consecutive successful radius expansions, while \( C_{\text{th}} \) defines the acceleration threshold for step‑size doubling. This design enables the Algorithm ~\ref{alg:2} to accelerate expansion when the distribution exhibits a high utilization rate, and to reduce the step size for fine‑grained search when utilization rate is low.

\begin{algorithm}[H]
	\caption{Virtual-force and radius expansion algorithm}
	\label{alg:2}
	\begin{algorithmic}[1]
		\Require Convex polygon $P$, circles' number $N$, target radius $r_{\text{target}}$
		\Ensure Optimized expansion system $\mathcal{S}$
		\State $r_{\text{current}} \gets 0.1 \times r_{\text{target}}$ \Comment{Initial small radius}
		\State $\Delta r_k \gets r_{\text{target}} \times \alpha_{\text{inflate}}$ \Comment{Initial inflation step size}
		\State $\mathcal{S} \gets \text{initialize\_hexagonal\_grid}(P, N, r_{\text{current}})$
		\State $C \gets 0$
		\While{$\Delta r_k > \epsilon_{\text{inflate}}$} \Comment{Main inflation loop}
		\State $r_{\text{test}} \gets r_{\text{current}} + \Delta r_k$ \Comment{Try larger radius}
		\State $\mathcal{S}_{\text{test}} \gets \mathcal{S}.\text{copy}()$
		\Comment{$\mathcal{S}.\text{copy}()$ denotes the system state after initialization.}
		\State $\mathcal{S}_{\text{test}}.\text{radius} \gets r_{\text{test}}$
		\State $\mathcal{S}_{\text{test}} \gets \text{optimize\_positions}(\mathcal{S}_{\text{test}})$ \Comment{Optimize under new radius}
		\If{$\mathcal{S}_{\text{test}}.\text{utilization\_rate} > U_\text{th}$} \Comment{Evaluate system state}
		\State $\mathcal{S} \gets \mathcal{S}_{\text{test}}$ \Comment{Accept new radius}
		\State $r_{\text{current}} \gets r_{\text{test}}$
		\State $C \gets C + 1$
		\If{$C \geq C_{\text{th}}$}
		\State $\Delta r_k \gets \Delta r_k \times 2$ \Comment{Accelerate inflation}
		\State $C \gets 0$
		\EndIf
		\Else
		\State $\Delta r_k \gets \Delta r_k / 2$ \Comment{Shrink step size}
		\State $C \gets 0$
		\EndIf
		\EndWhile
		\State \Return $\mathcal{S}$
	\end{algorithmic}
\end{algorithm}

\subsection{Boundary encircling algorithm}
Based on the structure-preserving initialization algorithm and the constraints of virtual forces, the majority of  circles are concentrated within the interior region of the convex polygon. However, some circles inevitably overflow beyond the polygon boundary, necessitating the development of a mechanism to handle boundary overflow. To address this problem, we propose a boundary encircling algorithm. The core idea is to relocate the overflowing circles along the boundary gradient direction to the nearest feasible positions.

We reformulate the original problem as minimizing the negative coverage rate to make it compatible with gradient descent, as follows:
\begin{equation}
J_{\text{cover}}\left( \boldsymbol{C} \right) =-Coverage\ ratio.
\end{equation}

\begin{df}[Boundary Distance Term]
	To discourage circles from exceeding the boundary, we define the boundary distance term
\begin{equation}
	J_{\text{boundary}}(\boldsymbol{C}) = \frac{1}{n} \sum_{i=1}^n \left( \min_{j} d_j(\boldsymbol{c}_i) - \delta  r \right)^2,
	\label{eq:19}
\end{equation}
	where \(\delta  \in [1,\infty )\) controls the desired distance between a circle and the boundary, and \(d_j(\boldsymbol{c}_i)\) denotes the right distance from the circle center \(\boldsymbol{c}_i\) to the \(j\)-th edge of the polygon. The condition \(\min_j d_j(\boldsymbol{c}_i) = \delta  r\) indicates that the center is positioned exactly \(\delta  r\) away from its closest polygon edge.
\end{df}

To retrieve overflow circles and implement boundary-encircling optimization, we sequentially introduce two constraints: a boundary constraint and an inter-center distance constraint. For all $i,j$, the following inequality constraints must be satisfied:
\begin{equation}
	\left\{ \begin{array}{l}
g_j\left( \boldsymbol{c}_i \right) =-\left( d_j\left( \boldsymbol{c}_i \right) -r \right) \le 0\\
		h_{ij}\left( \boldsymbol{C} \right) =2r-||\boldsymbol{c}_i-\boldsymbol{c}_j||\le 0\\
	\end{array} \right. .
	\label{eq:20}
\end{equation}
The augmented Lagrangian function is then formulated as follows:
\begin{equation}
	\begin{split}
		\mathcal{L}(\boldsymbol{C}, \lambda, \mu, \rho) 
		&= J_{\text{cover}}(\boldsymbol{C}) + \kappa  J_{\text{boundary}}(\boldsymbol{C}) \\
		&\quad + \sum_{i,j} \Big[ \mu_{ij} \max\big(0, g_j(\boldsymbol{c}_i)\big) 
		+ \frac{\rho_b}{2} \max\big(0, g_j(\boldsymbol{c}_i)\big)^2 \Big] \\
		&\quad + \sum_{i<j} \Big[ \lambda_{ij} \max\big(0, h_{ij}(\boldsymbol{C})\big) 
		+ \frac{\rho_c}{2} \max\big(0, h_{ij}(\boldsymbol{C})\big)^2 \Big].
	\end{split}
	\label{Eq:21}
\end{equation}

\begin{itemize}
	\item \textbf{Normal gradient}
\end{itemize}

Differentiating the boundary constraint leads to the normal gradient:
\begin{equation}
	\nabla _{\bot}\mathcal{L} = -\nabla_{\boldsymbol{c}_i} \left[ \frac{\rho_b}{2} \max\big(0, g_j(\boldsymbol{c}_i)\big)^2 \right] 
	= \begin{cases}
		\rho _b\bigl( r-d_j\left( \boldsymbol{c}_i \right) \bigr) \cdot \frac{\boldsymbol{n}_j}{\|\boldsymbol{n}_j\|}, & \text{if } d_j\left( \boldsymbol{c}_i \right) < r\\
		0, & \text{otherwise}
	\end{cases}.
	\label{Eq:22}
\end{equation}
This gradient provides an inward corrective force when the circle overflows($d_j(\boldsymbol{c}_i) < r$), and vanishes otherwise.

\begin{itemize}
	\item \textbf{Tangential gradient}
\end{itemize}

To drive the overflow circles along the boundary for optimal placement, we incorporate a tangential gradient that promotes boundary encircling. Let \( k \) be the index of the edge closest to the circle center \(\boldsymbol{c}_i\), with unit normal vector \(\boldsymbol{n} = \boldsymbol{n}_k / \|\boldsymbol{n}_k\|\). The tangent vector \(\boldsymbol{t}\) satisfies \(\boldsymbol{t} \cdot \boldsymbol{n} = 0\).
Differentiating the boundary distance term yields its gradient:
\begin{equation}
	\nabla_{\boldsymbol{c}_i} J_{\text{boundary}} = \frac{2}{n} \bigl( d_k(\boldsymbol{c}_i) - \delta r \bigr) \boldsymbol{n}.
\end{equation}
Since $\nabla_{\boldsymbol{c}_i} J_{\text{boundary}}$ is collinear with the normal vector $\boldsymbol{n}$, its projection onto the tangential space is zero:
\begin{equation}
	\nabla_{\parallel} J_{\text{boundary}} = \boldsymbol{P}_{\parallel} \nabla_{\boldsymbol{c}_i} J_{\text{boundary}} = \boldsymbol{0},
\end{equation}
where $\boldsymbol{P}_{\parallel} = \boldsymbol{I} - \boldsymbol{n}\boldsymbol{n}^\top$ is tangential projection operator. Therefore, to generate a driving force along the boundary, we further introduce the concept of boundary encircling points.

\begin{df}[Boundary Encircling Points]
	Uniformly spaced boundary encircling points are generated along each polygon edge:
	\begin{equation}
		\boldsymbol{b}_{j,l}=\boldsymbol{v}_j+\frac{l-0.5}{n_j}\left( \boldsymbol{v}_{j+1}-\boldsymbol{v}_j \right) +\vartheta  r\cdot \boldsymbol{n}_j\left( l=1,\cdots ,n_j \right) ,
		\label{eq:27}
	\end{equation}
	where the number of points on the \(j\)-th edge is defined as  
\begin{equation}
	n_j = \max\left(1, \Bigl\lfloor n \cdot \frac{\|\boldsymbol{v}_{j+1} - \boldsymbol{v}_j\|}{L}\Bigr\rfloor\right).
\end{equation}
	Here, $n$ is the total number of circles, and \(L\) denotes the perimeter of the polygon. The symbols are defined in Table ~\ref{tab:1}. $\frac{l-0.5}{n_j}$ ensures points are located strictly inside the edge. $\alpha$ controlls the degree of inward offset along the normal direction. $\mathcal{B}$ serves as candidate target positions for overflow circles.
\end{df}

\begin{table}[H]
	\caption{Description of Boundary Encircling Point Parameters}\label{tab:1}
	\begin{tabular*}{\tblwidth}{@{}LL@{}}
		\toprule
		\textbf{Symbol} & \textbf{Meaning} \\
	\midrule
$\boldsymbol{b}_{j,l}$ & Coordinates of the $l$-th boundary encircling point on the $j$-th edge. \\
$\boldsymbol{v}_j$, $\boldsymbol{v}_{j+1}$ & Coordinates of the $j$-th and $(j+1)$-th vertices of the polygon. \\
$n_j$ & Number of boundary encircling points allocated to the $j$-th edge. \\
$\frac{l-0.5}{n_j}$ & Normalized positional parameter with values in $(0,1)$. \\
$\vartheta $ & Boundary offset coefficient ($\vartheta  \in [0,1]$)  \\
$r$ & Radius of the circles. \\
$\boldsymbol{n}_j$ & Unit inward normal vector of the $j$-th edge (pointing inside the polygon). \\
$\mathcal{B}$ & Set of all boundary encircling points: $\mathcal{B} = \{\boldsymbol{b}_{j,l} \mid \forall j,l\}$, \\
\bottomrule
	\end{tabular*}
\end{table}

 For each circle $C_i$, we should find its nearest boundary encircling point:
\begin{equation}
	\boldsymbol{b}_i^* = \arg\min_{\boldsymbol{b} \in \mathcal{B}} \|\boldsymbol{c}_i - \boldsymbol{b}\|.
\end{equation}
The tangential attractive force generated by boundary encircling as follows:
\begin{equation}
	\boldsymbol{F}_{\text{border}}(\boldsymbol{c}_i) = \gamma (\boldsymbol{b}_i^* - \boldsymbol{c}_i),
\end{equation}
where $\gamma > 0$ denotes the attractive coefficient. Then, we obtain the tangential gradient as the projection of this attractive force onto the tangential space:
\begin{equation}
	\nabla_{\parallel} \mathcal{L}_{\text{border}} = \boldsymbol{P}_{\parallel} \boldsymbol{F}_{\text{border}}(\boldsymbol{c}_i) 
	= \gamma \boldsymbol{P}_{\parallel} (\boldsymbol{b}_i^* - \boldsymbol{c}_i).
	\label{eq:28}
\end{equation}

\begin{itemize}
	\item \textbf{Repulsive gradient between circles}
\end{itemize}

The repulsive gradient induced by inter-circle constraints is given by:
\begin{equation}
	\begin{split}
		\nabla_{\text{circle}}\mathcal{L} &= 
		\nabla_{\boldsymbol{c}_i}\left[ \frac{\rho_c}{2} \max\left(0, h_{ij}(\boldsymbol{C})\right)^2 \right] \\
		&= \begin{cases}
			\rho_c \bigl( 2r - \|\boldsymbol{c}_i - \boldsymbol{c}_j\| \bigr) \cdot 
			\dfrac{\boldsymbol{c}_i - \boldsymbol{c}_j}{\|\boldsymbol{c}_i - \boldsymbol{c}_j\|} & 
			\text{if } \|\boldsymbol{c}_i - \boldsymbol{c}_j\| < 2r \\
			0 & \text{otherwise}.
		\end{cases}.
	\end{split}
	\label{eq:29}
\end{equation}
By integrating the above gradients, the update formula for circle center optimization when overflow occurs is:

\begin{equation}
	\left\{ \begin{array}{l}
		\boldsymbol{c}_{i}^{\left( t+1 \right)}=\boldsymbol{c}_{i}^{\left( t \right)}-\eta \left[ \nabla _{\bot}\mathcal{L}+\nabla _{\parallel}\mathcal{L}_{\text{border}}+\sum_{j\ne i}{\nabla _{\text{circle}}}\mathcal{L} \right]\\
		\max_i|\boldsymbol{c}_{i}^{\left( t+1 \right)}-\boldsymbol{c}_{i}^{\left( t \right)}|<\epsilon\\
	\end{array} \right. .
	\label{eq:31}
\end{equation}
where $\eta > 0$ denotes learning rate and $\epsilon > 0$ is a predefined convergence threshold. For example, the boundary encircling process is illustrated in Fig.~\ref{fig:5} and Algorithm \ref{alg:3}.

\begin{figure}
	\centering
	\begin{subfigure}[b]{0.4\textwidth}
		\centering
		\includegraphics[width=\textwidth]{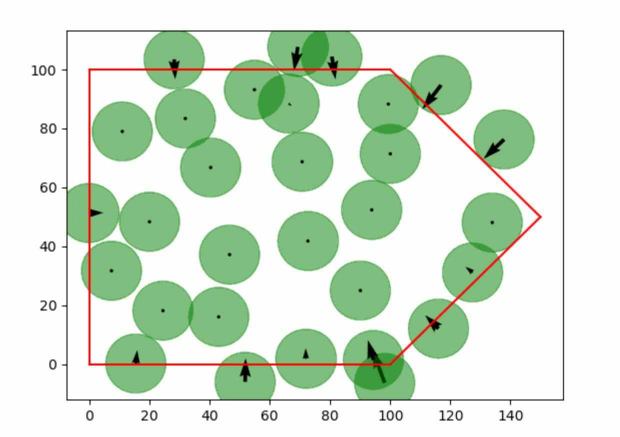}
		\caption{Normal Gradient}
	\end{subfigure}
	\begin{subfigure}[b]{0.4\textwidth}
		\centering
		\includegraphics[width=\textwidth]{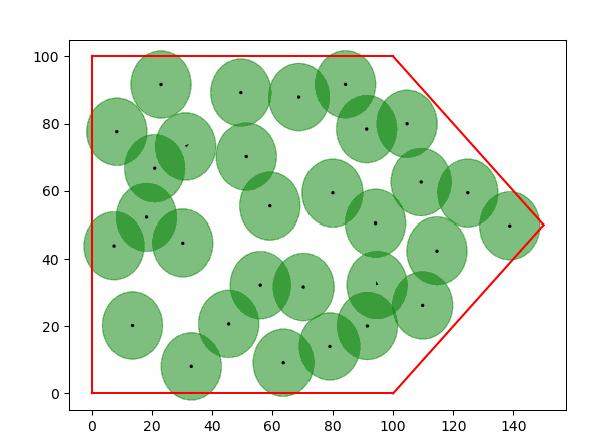}
		\caption{Tangential and repulsive gradient}
	\end{subfigure}
	\caption{Boundary encircling process}
	\label{fig:5}
\end{figure}

\begin{algorithm}[H]
	\caption{Boundary encircling algorithm}
	\label{alg:3}
	\begin{algorithmic}[1]
		\Require Optimized circle centers $\left\{ \boldsymbol{p}_i\left( r \right) \right\} _{i=1}^{n}$ after Algorithm \ref{alg:2}, parameters $\delta,\kappa ,\vartheta , \gamma, \rho_b, \rho_c, \eta, \epsilon$
		\Ensure Final optimized circle centers $\boldsymbol{C}^{*}$
		\State Generate the set of boundary encircling points $\mathcal{B}$ according to Eq.\eqref{eq:27}
		\For{$t = 0, 1, 2, \dots$ until $\max_i \|\boldsymbol{c}_i^{(t+1)} - \boldsymbol{c}_i^{(t)}\| < \epsilon$}
		\For{$i = 1$ \textbf{to} $n$}
		\State Compute the minimum boundary distance: $d_{\min}(\boldsymbol{c}_i^{(t)}) = \min_j d_j(\boldsymbol{c}_i^{(t)})$
		\State Find the closest boundary encircling point: $\boldsymbol{b}_i^* = \arg\min_{\boldsymbol{b} \in \mathcal{B}} \|\boldsymbol{c}_i^{(t)} - \boldsymbol{b}\|$
		\State Compute the normal gradient $\nabla_{\bot}\mathcal{L}$ 
		\State Compute the tangential gradient $\nabla_{\parallel}\mathcal{L}_{\text{border}}$ 
		\State Compute the inter‑circle repulsive gradient $\sum_{j \neq i} \nabla_{\text{circle}}\mathcal{L}$ 
		\State Update the center position:
		$\boldsymbol{c}_i^{(t+1)} = \boldsymbol{c}_i^{(t)} - \eta \bigl[ \nabla_{\bot}\mathcal{L} + \nabla_{\parallel}\mathcal{L}_{\text{border}} + \sum_{j \neq i} \nabla_{\text{circle}}\mathcal{L} \bigr]$
		\EndFor
		\EndFor
		\State \Return $\boldsymbol{C}^{*}$
	\end{algorithmic}
\end{algorithm}

\subsection{Computational complexity}
We analyze the computational complexity of the  proposed algorithm (IQPD) in this subsection. We quantify the computational overhead of each module within the optimization framework, including structure-preserving initialization, virtual force and radius expansion strategy, and boundary encircling algorithm, with respect to the number of circles $n$ and polygon vertices $m$.

\begin{lemma}[Complexity of structure-preserving initialization]
	\label{lemma4}
The  initialization algorithm computes an initial configuration of \(n\) circles within a convex polygon of \(m\) vertices in $O\left( \max \left\{ m,n \right\} \log n \right)$  time.
\end{lemma}

\begin{proof}
	Algorithm \ref{alg:1}  for initialization proceeds in six steps. Step 1 computes the minimum bounding rectangle (MBR) using the rotating calipers algorithm, which requires \(O(m)\) time \cite{carlos_2020}. Steps 2 through 5 generate the hexagonal lattice and apply rotation, scaling, and translation transformations; each of these steps processes all \(n\) circles in \(O(n)\) time. Step 6 performs boundary-aware filtering \cite{11515007} and corner-occupancy insertion \cite{LU20081742}, incurring \(O(n \log n + m \log n)\) time. Therefore, the  computational complexity of the structure-preserving initialization algorithm is
	\begin{equation}
\begin{aligned}
	T_{\text{init}} &=O\left( m \right) +4\cdot O\left( n \right) +O\left( n\log n+m\log n \right)\\ 
	&=O\left( m \right) +O\left( n \right) +O\left( n\log n \right) +O\left( m\log n \right) \\
	&=O\left( \max \left\{ m,n \right\} \log n \right) 
\end{aligned}.
	\end{equation}
\end{proof}

\begin{lemma}[Complexity of virtual-force and radius expansion]
	\label{lem:5}
	Let \(K_{\text{outer}}\) denote the number of radius expansion steps and \(K_{\text{inner}}\) the number of virtual force  iterations per expansion step. The virtual-force and radius expansion algorithm has time complexity $O\left( K_{\text{outer}}\cdot K_{\text{inner}}\cdot \left( n^2+mn \right) \right)$.
\end{lemma}

\begin{proof}
	In each virtual force iteration, computing  repulsive forces between circles requires \(O(n^2)\) operations, as all \(\binom{n}{2}\) circle pairs must be examined. Computing boundary forces for all circles against \(m\) polygon edges costs \(O(n \cdot m)\).  The velocity and position updates for $n$ circles cost \(O(n)\) per iteration. Thus, each  virtual force iteration has complexity
	\begin{equation}
	 T_{\text{force}}=O\left( n^2+mn \right).
	\end{equation}
	The radius expansion phase performs \(K_{\text{inner}}\) such iterations for each of the \(K_{\text{outer}}\) expansion steps. The total time complexity is therefore 
	\begin{equation}
		T_{\text{inflate}}=\sum_{i=1}^{K_{\text{outer}}}{K_{\text{inner}}\cdot T_{\text{force}}}=O\left( K_{\text{outer}}\cdot K_{\text{inner}}\cdot \left( n^2+mn \right) \right) .
	\end{equation}
	
Then we present theoretical estimates of the upper bounds for $K_{\text{outer}}$ and $K_{\text{inner}}$. $K_{\text{outer}}$ is determined by the adaptive step-size strategy and the stopping threshold in Algorithm \ref{alg:2}, which is formulated as:
\begin{equation}
	K_{\text{outer}} = \min \left\{ k \mid \Delta r_k < \epsilon_{\text{inflate}} \right\},
\end{equation}
where $\Delta r_k$ represents the radius increment at the $k$-th iteration, and $\epsilon_{\text{inflate}}$ is the predefined stopping threshold. We consider the worst-case scenario where the step size is never doubled and is always halved, the number of iterations required for the step size to decay from the initial step size $\Delta r_0 = \alpha_{\text{inflate}} \cdot r_{\text{target}}$ to the stopping threshold $\epsilon_{\text{inflate}}$ satisfies
\begin{equation}
K_{\text{outer}} \leq \left\lceil \log \left( \frac{\Delta r_0}{\epsilon_{\text{inflate}}} \right) \right\rceil + O(1) \leq O\left(\log\left(\frac{r_{\text{target}}}{\epsilon_{\text{inflate}}}\right) \right)
\end{equation}
according to Theorem \ref{thm:10}. The system converges to an equilibrium state via virtual force iterations, and $K_{\text{inner}}$ is determined by the convergence condition:
\begin{equation}
	K_{\text{inner}} = \min\left\{ t \,\bigg|\, |U^{(t)} - U^{(t-1)}| <1- U_{\text{th}}\right\}.
\end{equation}
The virtual force iteration terminates when the utilization rate changes by less than \(\epsilon_U = 1 - U_{\text{th}}\) between successive iterations. The virtual-force dynamics exhibit linear convergence to the equilibrium state according to Theorem \ref{thm:9} , implying that the utilization rate converges linearly as well:
\begin{equation}
|U^{(t)} - U^*| \leq C \cdot \rho^t, \quad \rho < 1.
\end{equation}
To achieve \(|U^{(t)} - U^*| < \epsilon_U\), the required number of iterations is
\begin{equation}
	K_{\text{inner}} \leq \frac{\log(C / \epsilon_U)}{-\log(\rho)} = O\left( \log\left( \frac{1}{1 - U_{\text{th}}} \right) \right).
\end{equation}
\end{proof}

\begin{lemma}[Complexity of boundary encircling algorithm]
	\label{lem:boundary_complexity}
	Let \(K_{\text{boundary}}\) denote the number of gradient descent iterations. The boundary encircling algorithm has time complexity \(O(K_{\text{boundary}} \cdot (n^2+m \cdot n))\).
\end{lemma}

\begin{proof}
	The algorithm \ref{alg:3} performs four steps per iteration. 
	
	\textbf{Step 1: Generation of boundary encircling points.} According to Eq.\ref{eq:27}, the number of boundary encircling points on the $j$-th edge is $	n_j = \max\left(1, \left\lfloor n \cdot \frac{\|\boldsymbol{v}_{j+1} - \boldsymbol{v}_j\|}{L} \right\rfloor\right),$
	where $L$ is the perimeter of the polygon. Since $\sum_{j=1}^m \|\boldsymbol{v}_{j+1} - \boldsymbol{v}_j\| = L$, the total number of boundary points satisfies
\begin{equation}
	n_b = \sum_{j=1}^m n_j = O(n).
\end{equation}
	Generating these points requires iterating over all $m$ edges and computing coordinates for each point, which costs
\begin{equation}
	T_{\text{gen}} = O(m + n_b) = O(m + n).
\end{equation}
	
\textbf{Step 2: Nearest boundary point search.} 	For each circle, we find its closest boundary encircling point
$\boldsymbol{b}_i^* = \arg\min_{\boldsymbol{b} \in \mathcal{B}} \|\boldsymbol{c}_i - \boldsymbol{b}\|$ by KD-tree \cite{LIAN2026123347}, which costs time complexity 
\begin{equation}
	T_{\text{find}} = O(n \log n).
\end{equation}
	
\textbf{Step 3: Gradient computation.}
Three gradient components are computed in each iteration as shown in Table \ref{tab:2}.
\begin{table}[H]
	\caption{Complexity of gradient components in the boundary encircling algorithm}
	\label{tab:2}
	\begin{tabular*}{\tblwidth}{@{}llll@{}}
		\toprule
		\textbf{Gradient Term} & \textbf{Formula} & \textbf{Complexity} & \textbf{Explanation} \\
		\midrule
		Normal gradient \( \nabla_{\perp} \mathcal{L} \) & Eq.\ref{Eq:22} & \( O(n \cdot m) \) & Compute distance from each circle  to all \( m \) polygon edges. \\
		Tangential gradient \( \nabla_{\parallel} \mathcal{L}_{\text{border}} \) & Eq.\ref{eq:28} & \( O(n) \) & Use precomputed nearest boundary point.\\
		Inter-circle repulsive gradient \( \nabla_{\text{circle}} \mathcal{L} \) & Eq.\ref{eq:29} & \( O(n^2) \) & Iterate over all \( \binom{n}{2} \) circle pairs for overlap detection.\\
		\bottomrule
	\end{tabular*}
\end{table}
Therefore, the total gradient computation cost per iteration is
\begin{equation}
T_{\text{grad}} = O(n^2 + n \cdot m + n) = O(n^2 + n \cdot m).
\end{equation}

\textbf{Step 4: Position update.}
The update rule for each circle center is given by Eq.\ref{eq:31}. Updating all $n$ circles requires
\begin{equation}
	T_{\text{update}} = O(n).
\end{equation}

Summing the costs of all steps, the total complexity per iteration is
\begin{equation}
\begin{aligned}
	T_{\text{iter}} &= T_{\text{gen}} + T_{\text{find}} + T_{\text{grad}} + T_{\text{update}} \\
	&= O(m + n) + O(n \log n) + O(n^2+m \cdot n) + O(n) \\
	&= O(n^2 + m \cdot n)
\end{aligned}.
\end{equation}
Then we obtain the total complexity of the boundary encircling algorithm 
\begin{equation}
T_{\text{boundary}} = O(K_{\text{boundary}} \cdot (n^2 + m \cdot n)).
\end{equation}

The iteration count \( K_{\text{boundary}} \) is determined by the stopping criterion (Eq.\ref{eq:31}). Due to the linear convergence rate of gradient descent on strongly convex functions, \( K_{\text{boundary}} \) satisfies
\begin{equation}
K_{\text{boundary}} = O\left( \log \left( \frac{1}{\epsilon} \right) \right),
\end{equation}
where \( \epsilon \) is the predefined convergence tolerance.
\end{proof}

\begin{theorem}[Total complexity]
	\label{lem:total_complexity}
	The overall algorithm achieves a time complexity of
	\begin{equation} 
O\!\left( \max\{m, n\} \log n + K_{\text{outer}} K_{\text{inner}} (n^2 + m n) + K_{\text{boundary}} (n^2 + m n) \right),
\end{equation}
 where the constants \(K_{\text{outer}}\) , \(K_{\text{inner}}\) , and \(K_{\text{boundary}}\) are bounded by problem-independent thresholds.
\end{theorem}

\subsection{Convergence analysis}

We analyze the convergence of the proposed algorithm in this subsection. Let \( \boldsymbol{X} = (\boldsymbol{x}_1, \ldots, \boldsymbol{x}_n) \in \mathbb{R}^{2n} \) denote the configuration of all circles in the convex polygon as defined in Definition \ref{def:5}, and let \( E(\boldsymbol{X}) \) denote the total potential energy of the system, comprising overlap energy between circles and boundary penetration energy.

\begin{theorem}[Energy dissipation]
	\label{thm:energy_dissipation}
	Consider the virtual-force dynamical system defined by Eqs.~\eqref{eq:10}--\eqref{eq:13}, where the conservative forces derive from a potential energy function $\Phi(\boldsymbol{X})$ according to $\boldsymbol{F}_{\text{cons},i} = -\nabla_{\boldsymbol{x}_i} \Phi$. The total energy of
	the system is
\begin{equation}
	\label{Eq:48}
	E(\boldsymbol{X}, \boldsymbol{V}) = \frac{1}{2} \sum_{i=1}^n \|\boldsymbol{v}_i\|^2 + \Phi(\boldsymbol{X}).
\end{equation}
	Then the energy dissipation rate satisfies
\begin{equation}
	\frac{dE}{dt} = -\mu \sum_{i=1}^n \|\boldsymbol{v}_i\|^2 \leq 0,
\end{equation}
	where $\mu > 0$ is the friction coefficient. Consequently, the virtual-force dynamical system is energy-dissipative, and the energy decreases monotonically until all velocities vanish.
\end{theorem}

\begin{proof}
	In classical mechanics, a particle of mass \( m \) moving with velocity \( v \) possesses kinetic energy $\frac{1}{2} m v^2$. In our algorithm, to simplify computation, the mass of each circle is set to $1$. Therefore, the kinetic energy of a single circle $C_i$ is $\frac{1}{2} \|\boldsymbol{v}_i\|^2$. Summing over all \( n \) circles, the total kinetic energy is
\begin{equation}
	T = \frac{1}{2} \sum_{i=1}^{n} \|\boldsymbol{v}_i\|^2.
\end{equation}

Let \( \Phi(\boldsymbol{X}) \) denote the potential energy function of the system, which depends only on the positions of the circle centers \( \boldsymbol{X} = (\boldsymbol{x}_1, \ldots, \boldsymbol{x}_n) \). In our virtual force model, conservative forces (inter-circle repulsion and boundary repulsion) are derived from the potential energy $\boldsymbol{F}_{\text{cons}, i} = -\nabla_{\boldsymbol{x}_i} \Phi.$ This implies that:
\begin{itemize}
	\item The potential energy  increases with the overlapping area between circles;
	\item The force points in the direction of decreasing potential energy (i.e., repulsive forces push circles away from overlapping regions).
\end{itemize}

The total energy $E(\boldsymbol{X}, \boldsymbol{V})$ of the system is defined as the sum of kinetic and potential energies, as given in Eq.\eqref{Eq:48}, which is the standard definition in classical mechanics. Then $E = \frac{1}{2} \sum_i \|\boldsymbol{v}_i\|^2 + \Phi(\boldsymbol{X})$ satisfies
\begin{equation}
\frac{dE}{dt} = \frac{d}{dt}\left( \frac{1}{2} \sum_{i=1}^n \|\boldsymbol{v}_i\|^2 \right) + \frac{d\Phi}{dt}.
\end{equation}
The time derivative of the kinetic energy of the $i$-th circle is
\begin{equation}
\frac{d}{dt}\left( \frac{1}{2} \|\boldsymbol{v}_i\|^2 \right) = \boldsymbol{v}_i \cdot \dot{\boldsymbol{v}}_i = \boldsymbol{v}_i \cdot (\boldsymbol{F}_{\text{cons},i} - \mu \boldsymbol{v}_i) = \boldsymbol{v}_i \cdot \boldsymbol{F}_{\text{cons},i} - \mu \|\boldsymbol{v}_i\|^2.
\end{equation}
Summing over all circles, we obtain
\begin{equation}
	\frac{d}{dt}\left( \frac{1}{2} \sum_{i=1}^n \|\boldsymbol{v}_i\|^2 \right) = \sum_{i=1}^n \boldsymbol{v}_i \cdot \boldsymbol{F}_{\text{cons},i} - \mu \sum_{i=1}^n \|\boldsymbol{v}_i\|^2.
	\label{Eq:53}
\end{equation}
Since the conservative forces derive from the potential energy function $\Phi$ via $\boldsymbol{F}_{\text{cons},i} = -\nabla_{\boldsymbol{x}_i} \Phi$, the chain rule gives
\begin{equation}
	\frac{d\Phi}{dt} = \sum_{i=1}^n \left( \nabla_{\boldsymbol{x}_i} \Phi \right) \cdot \boldsymbol{v}_i = -\sum_{i=1}^n \boldsymbol{v}_i \cdot \boldsymbol{F}_{\text{cons},i}.
	\label{Eq:54}
\end{equation}
Substituting Eq.~\eqref{Eq:54} into the Eq.~\eqref{Eq:53} yields
\begin{equation}
	\frac{d}{dt}\left( \frac{1}{2} \sum_{i=1}^n \|\boldsymbol{v}_i\|^2 \right) = -\frac{d\Phi}{dt} - \mu \sum_{i=1}^n \|\boldsymbol{v}_i\|^2.
\end{equation}
Finally, the total energy $E = \frac{1}{2} \sum_i \|\boldsymbol{v}_i\|^2 + \Phi$ satisfies
$\frac{dE}{dt} = \frac{d}{dt}\left( \frac{1}{2} \sum_{i=1}^n \|\boldsymbol{v}_i\|^2 \right) + \frac{d\Phi}{dt} = -\mu \sum_{i=1}^n \|\boldsymbol{v}_i\|^2 \leq 0.$
\end{proof}

Equality holds if and only if all velocities are zero, i.e., $\boldsymbol{v}_i = \boldsymbol{0}$ for all $i$. This indicates that the system's energy decreases monotonically until the system reaches a static equilibrium. 

\begin{theorem}[Local linear convergence]
	\label{thm:9}
Assume that the potential energy function $\Phi(\boldsymbol{X})$ is twice continuously differentiable in a neighborhood of a local minimizer \( \boldsymbol{X}^* \), and that the Hessian $\nabla^2 \Phi(\boldsymbol{X}^*)$ is positive definite. Let \( \lambda_{\min} \) and \( \lambda_{\max} \) be the smallest and largest eigenvalues of \( \nabla^2 \Phi(\boldsymbol{X}^*) \), respectively. If the time step \( \Delta t \) satisfies $\Delta t < \frac{2}{\sqrt{\lambda _{\max}}}$, then the discrete-time iteration (Eq.~\eqref{eq:12}--\eqref{eq:13}) converges linearly to \( \boldsymbol{X}^* \) with convergence rate
\begin{equation}
	\rho = \max\left\{ |1 - \lambda_{\min}\Delta t|,\; |1 - \lambda_{\max}\Delta t| \right\} < 1.
\end{equation}
Specifically, there exists a constant \( C > 0 \) such that
\begin{equation}
	\|\boldsymbol{X}^{(k)} - \boldsymbol{X}^*\| \leq C \cdot \rho^k.
\end{equation}
\end{theorem}

\begin{proof}
The entire virtual force dynamical system is determined by the following equations
\begin{equation}
	\begin{cases} 
		\dot{\boldsymbol{v}}(t) = -\nabla \Phi(\boldsymbol{x}(t)) - \mu \boldsymbol{v}(t) \\ 
		\dot{\boldsymbol{x}}(t) = \boldsymbol{v}(t) 
	\end{cases},
	\label{eq:58}
\end{equation}
where \(\mu > 0\) is the friction coefficient. Discretization of Eq.\eqref{eq:58} yields Eqs.\eqref{eq:12}--\eqref{eq:13}, which can also be written as
\begin{equation}
\left\{ \begin{array}{l}
	\boldsymbol{v}^{\left( k+1 \right)}=\boldsymbol{v}^{\left( k \right)}-\Delta t\,\nabla \Phi \left( \boldsymbol{x}^{\left( k \right)} \right) -\mu \Delta t\,\boldsymbol{v}^{\left( k \right)}\\
	\boldsymbol{x}^{\left( k+1 \right)}=\boldsymbol{x}^{\left( k \right)}+\Delta t\boldsymbol{v}^{\left( k \right)}\\
\end{array} \right. .
\end{equation}
 Then we eliminate the velocity variable. Substituting \( \boldsymbol{v}^{(k)} = \frac{\boldsymbol{x}^{(k+1)} - \boldsymbol{x}^{(k)}}{\Delta t} \) into the velocity update equation and neglecting higher-order terms yields the approximate second-order difference equation
\begin{equation}
\boldsymbol{x}^{(k+1)} - 2\boldsymbol{x}^{(k)} + \boldsymbol{x}^{(k-1)} = -\Delta t^2 \, \nabla \Phi(\boldsymbol{x}^{(k)}) - \mu \Delta t \, (\boldsymbol{x}^{(k)} - \boldsymbol{x}^{(k-1)}).
\end{equation}
In a neighborhood of a local minimum \( \boldsymbol{x}^* \), the gradient can be approximated as \( \nabla \Phi(\boldsymbol{x}) \approx \nabla^2 \Phi(\boldsymbol{x}^*) (\boldsymbol{x} - \boldsymbol{x}^*) \). Defining the error \( \boldsymbol{e}^{(k)} = \boldsymbol{x}^{(k)} - \boldsymbol{x}^* \) and neglecting higher-order terms, we obtain the linearized system
\begin{equation}
\boldsymbol{e}^{(k+1)} - 2\boldsymbol{e}^{(k)} + \boldsymbol{e}^{(k-1)} = -\Delta t^2 \boldsymbol{H} \boldsymbol{e}^{(k)} - \mu \Delta t (\boldsymbol{e}^{(k)} - \boldsymbol{e}^{(k-1)}),
\end{equation}
where \( \boldsymbol{H} = \nabla^2 \Phi(\boldsymbol{x}^*) \) is the Hessian matrix. We can rewrite the system in state-space form
\begin{equation}
\left\{ \begin{array}{l}
	\boldsymbol{y}^{\left( k \right)}=\left[ \boldsymbol{e}^{\left( k \right)};\,\boldsymbol{e}^{\left( k-1 \right)} \right]\\
	\boldsymbol{y}^{\left( k+1 \right)}=\boldsymbol{Ay}^{\left( k \right)}\\
\end{array} \right. ,
 \end{equation}
where \( \boldsymbol{A} \) is a \( 2n \times 2n \) matrix. For each eigenvalue \( \lambda \) of the Hessian \( \boldsymbol{H} \), the corresponding eigenmode satisfies 
\begin{equation}
\boldsymbol{e}^{\left( k+1 \right)}-\left( 2-\Delta t^2\lambda -\mu \Delta t \right) \boldsymbol{e}^{\left( k \right)}+\left( 1-\mu \Delta t \right) \boldsymbol{e}^{\left( k-1 \right)}=0.
\end{equation}
Let $ e^{(k)} = \xi^k $, we obtain the characteristic equation
\begin{equation}
\xi^2 - (2 - \Delta t^2 \lambda - \mu \Delta t) \xi + (1 - \mu \Delta t) = 0.
\label{Eq:65}
\end{equation}
Solving Eq.~\eqref{Eq:65} yields  two characteristic roots \( \xi_1, \xi_2 \). According to von Neumann stability analysis \cite{VonNeumann},
the convergence rate of the system is determined by
$\rho(\lambda) = \max(|\xi_1|, |\xi_2|).$ For convergence, both characteristic roots need  satisfying $|\xi| < 1$. Under the condition $\Delta t < 2/\sqrt{\lambda _{\max}}$ and with appropriate damping $\mu > 0$, the spectral radius of the system is given by
\begin{equation}
	\rho = \max\left\{ |1 - \lambda_{\min}\Delta t|,\ |1 - \lambda_{\max}\Delta t| \right\}.
	\label{Eq:66}
\end{equation}
Eq.~\eqref{Eq:66} can be verified by analyzing the characteristic equation. When $\rho < 1$, the error decays as $\|\boldsymbol{e}^{(k)}\| \leq C \rho^k$ for some constant $C > 0$, establishing linear convergence. 
\end{proof}

\begin{theorem}[Convergence of radius expansion]
	\label{thm:10}
	Let $r^* = r_{\text{target}}$ be the target radius. Starting from \( r_0 = \alpha r^* \) with \( \alpha \in (0,1) \), the adaptive radius expansion algorithm (Algorithm~\ref{alg:2}) produces a sequence \( \{r_k\} \) that converges monotonically to \( r^* \). The number of expansion steps satisfies
\begin{equation}
	K_{\text{outer}} \leq \left\lceil \log \left( \frac{r^*}{\epsilon_{\text{inflate}}} \right) \right\rceil + \left\lceil \log \left( \frac{1}{\alpha} \right) \right\rceil + O(1) \leq O \left( \log \left( \frac{r_{\text{target}}}{\epsilon_{\text{inflate}}} \right) \right)
\end{equation}
	where \( \epsilon_{\text{inflate}} \) is the stopping threshold for the step size.
\end{theorem}

\begin{proof}
	The step size $\Delta r_k$ doubles when the utilization rate $U(r_k + \Delta r_k) > U_{\text{th}}$ for $C_{\text{th}}$ consecutive successes, and halves otherwise.  For a well-initialized system, the algorithm performs a binary-like search on the radius. The initial step size is $\Delta r_0 = r^* \cdot \alpha_{\text{inflate}}$. The number of doublings needed to approach $r^*$ is at most $ \lceil \log(1/\alpha) \rceil$. After overshoot, the step size halves geometrically until it falls below $\epsilon_{\text{inflate}}$, requiring at most $\lceil \log(r^*/\epsilon_{\text{inflate}}) \rceil$ halvings. The total number of steps is bounded by the sum of these terms plus a constant. 
\end{proof}

\begin{theorem}[Convergence of augmented Lagrangian method]
	\label{thm:boundary_convergence}
Consider the constrained optimization problem defined by Eqs.~\eqref{eq:19} and \eqref{eq:20}, where the objective function \(J_{\text{cover}}(\boldsymbol{C})\) and constraint functions \(g_j(\boldsymbol{c}_i)\), \(h_{ij}(\boldsymbol{C})\) are continuously differentiable. Let \(\{\boldsymbol{C}^{(k)}\}\) be the sequence generated by Algorithm~\ref{alg:3} with the  update rules given in Eq.~\eqref{eq:31}. We can assume that:
\begin{enumerate}
	\item The penalty parameters \(\rho_b, \rho_c > 0\) are chosen sufficiently large;
	\item The gradient descent step size \(\eta > 0\) is sufficiently small to ensure convergence of the inner loop;
	\item A constraint qualification (e.g., Mangasarian-Fromovitz Constraint Qualification \cite{Li}) holds at the limit point.
\end{enumerate}
Then every accumulation point of \(\{\boldsymbol{C}^{(k)}\}\) is an approximate KKT point \cite{Feng2019} of the original problem. Moreover, the constraint violations satisfy
\begin{equation}
\max\left\{ \max_{i,j} \max(0, g_j(\boldsymbol{c}_i^{(k)})),\; \max_{i<j} \max(0, h_{ij}(\boldsymbol{C}^{(k)})) \right\} = O\left(\frac{1}{\min(\rho_b, \rho_c)}\right).
\end{equation}
\end{theorem}

\begin{proof}
	Algorithm \ref{alg:3} is a variant of the augmented Lagrangian method combined with gradient descent. For fixed penalty parameters \(\rho_b, \rho_c > 0\), the augmented Lagrangian function \(\mathcal{L}(\boldsymbol{C}, \lambda, \mu, \rho)\) is continuously differentiable with respect to \(\boldsymbol{C}\). Since the domain of \(\boldsymbol{C}\) is bounded (all circles are confined within the convex polygon), the gradient \(\nabla_{\boldsymbol{C}} \mathcal{L}\) is Lipschitz continuous on this domain. Let \(L\) denote the Lipschitz constant. With a step size satisfying \(\eta < 1/L\), the gradient descent iterates are guaranteed to converge to a stationary point of \(\mathcal{L}\) for fixed multipliers and penalties \cite{Nocedal2006}.
	
    For sufficiently large penalty parameters \(\rho_b, \rho_c\), the quadratic penalty terms dominate, making the augmented Lagrangian locally convex near the feasible region. Consequently, the stationary point obtained from the inner loop satisfies the approximate KKT conditions, where the violation of each constraint is bounded by \(O(1/\rho_b)\) and \(O(1/\rho_c)\), respectively. Under the assumption that a constraint qualification holds at the accumulation point, standard results in augmented Lagrangian theory  \cite{Nocedal2006} guarantee that any accumulation point of \(\{\boldsymbol{C}^{(k)}\}\) is an approximate KKT point. The constraint violation bound follows directly from the update formulas, with residual error inversely proportional to the penalty parameter. Hence,
	\begin{equation}
	\max(0, g_j(\boldsymbol{c}_i^{(k)})) = O(1/\rho_b), \quad \max(0, h_{ij}(\boldsymbol{C}^{(k)})) = O(1/\rho_c),
	\end{equation}
	which yields the stated bound. 
\end{proof}

\begin{theorem}[Overall Convergence]
	\label{thm:overall_convergence}
	Let the target number of circles \( n \) and target radius \( r^* \) be given. The complete algorithm (Algorithm \ref{alg:1} + Algorithm \ref{alg:2} + Algorithm \ref{alg:3}) produces a sequence of configurations \( \{\boldsymbol{X}^{(t)}\} \) such that:
	\begin{enumerate}
		\item The coverage rate is non-decreasing after the first few iterations.
		\item The sequence of utilization rates \( U^{(t)} \) converges to a local maximum.
		\item The final configuration satisfies all constraints  up to a tolerance determined by the stopping criteria.
	\end{enumerate}
Furthermore, the overall iteration count is bounded by a constant that depends logarithmically on the desired precision.
\end{theorem}

\begin{proof}
Theorem \ref{thm:overall_convergence} follows directly from Theorems~\ref{thm:energy_dissipation}–\ref{thm:boundary_convergence}. The virtual-force phase ensures that for each fixed radius, the system approaches an equilibrium state with minimal overlap. The radius expansion monotonically increases the radii toward the target while maintaining feasibility. The boundary encircling then refines the positions to satisfy all constraints exactly. The monotonicity of coverage follows from the fact that increasing radii (without causing overflow) expands the covered area, and the subsequent boundary adjustment does not decrease the feasible coverage. 
\end{proof}

\begin{rmk}
The above theorems guarantee local convergence under mild regularity conditions. Global convergence to the global optimum is not guaranteed due to the non-convex nature of the problem; however, the proposed initialization strategy (structure-preserving hexagonal lattice) provides a high-quality starting point that significantly improves the likelihood of reaching a near-global optimum, which greatly reduces computational complexity.
\end{rmk}

 \section{Experiments and Discussion}
To evaluate the performance of IQPD algorithm in optimal deployment, we compare it with eight state-of-the-art metaheuristic algorithms from 2021 to 2026: ICGWO \cite{Shaikh2025}, ICS-MS \cite{Yang2025}, MIFSA \cite{biomimetics10110750}, VGSOK \cite{CHOWDHURY2021102660}, DAQGA \cite{LI2025111431}, FSIMI-GWO \cite{MOHAMMADI2026117912}, PSO-VD \cite{LIU2021103019} and EHPSO \cite{Tong2025}. To enhance the credibility of our experimental results, we conduct extensive experiments on both synthetic datasets and real-world scenarios. The experimental environment and default experimental parameters are summarized in Table~\ref{tab:3} and Table~\ref{tab:4}.

\begin{table}
	\centering
	\caption{Experimental environment}
	\label{tab:3}
	\begin{tabular*}{\tblwidth}{@{}LL@{}}
		\toprule
		\textbf{Component} & \textbf{Specification} \\
		\midrule
		Processor & Intel Core Ultra 9 275HX @ 2.70 GHz (24 cores, 24 threads) \\
		Memory & 32 GB DDR5 @ 6400 MT/s (2×16 GB SODIMM) \\
		GPU & NVIDIA GeForce RTX 5060 Laptop GPU (8 GB GDDR6 VRAM) \\
		Operating System & Windows 11 Pro (64-bit) \\
		Development Environment & Visual Studio Code 1.106.3 \\
		Programming Language & Python 3.12.5 \\
		\multirow{1}{*}{Key Python Libraries} & NumPy 2.3.5, SciPy 1.16.3, Shapely 2.1.2, \\
		& Matplotlib 3.10.7, Seaborn 0.13.2, Pandas 2.3.3, \\
		& Openpyxl 3.1.5, Scikit-learn 1.7.2 \\
		\bottomrule
	\end{tabular*}
\end{table}

\begin{table}
	\centering
	\caption{Default experimental parameters}
	\label{tab:4}
	\begin{tabular*}{\tblwidth}{@{}LLL@{}}
		\toprule
		Parameter & Symbol & Value  \\
		\midrule
		Circle (Node) number &  $n$  &   $\left\{ 5,10,15,20,25,30 \right\} \cup \left\{ 100,200,300 \right\} 	$\\
		Circle (Coverage) radius &  $r$  &   $2$ (In practical scenarios, this can represent values such as 20, 200, etc.)	\\
		Maximum iteration &  $T_{max}$  &   $500$\\
		Safety coefficient & $\alpha$ &  $0.92$\\
		Viscous coefficient & $\mu $ &  $0.6$\\
		Utilization threshold & $U_{\text{th}} $ &  $0.9$\\
		Acceleration threshold & $C_{\text{th}} $ &  $3$\\
		Distance controlling coefficient  &$\delta$ & $1.2$ \\
		Boundary offset coefficient  &$\vartheta$ & $0.5$\\
		Attractive coefficient &$\gamma$ & $0.8$\\
		Learning rate  &$\eta$ &  0.01\\
		Convergence threshold &$\epsilon$ & $10^{-4}$\\
		\bottomrule
	\end{tabular*}
\end{table}

Moreover, we employ five metrics: coverage rate (CR)  defined in Eq.\ref{eq:CR}, utilization rate (UR) as definition \ref{df:ur} , uniformity index (UI) in \cite{LIU2021103019}, Minimum Gap (MG) \cite{Liu2011}, and Computational Time (CT). UI is employed to quantify the spatial distribution uniformity of the circle centers. It is calculated based on the Voronoi diagram generated by the circle centers within the polygon. A higher UI value indicates a more uniform distribution, which helps in avoiding coverage blind holes and ensures that the coverage potential of each circle is fully used. The index is formally defined as follows:
\begin{equation}
	UI = 1 - \frac{\sigma_{\text{areas}}}{\mu_{\text{areas}}}.
	\label{Eq:ui}
\end{equation}
In the Eq.\ref{Eq:ui}, $\mu_{\text{areas}}$ and $\sigma_{\text{areas}}$ represent the mean and standard deviation of the areas of the Voronoi cells generated by the circle centers respectively. MG is a constraint metric for measuring the feasibility of a coverage solution, ensuring that any two circles do not overlap excessively or maintain a reasonable distance. The metric is calculated as follows:
\begin{equation}
	MG=\min_{i\ne j}\left( ||\boldsymbol{c}_i-\boldsymbol{c}_j||-2r \right).
\end{equation}

\begin{table}[htbp]
	\centering
	\caption{Statistical summary of the three domain types (synthetic datasets)}
	\label{tab:5}
	\begin{tabular*}{\tblwidth}{@{\extracolsep{\fill}}lcccl@{}}
		\toprule
		\textbf{Domain Type} & \textbf{Count} & \textbf{Vertex Range} & \textbf{Area (Mean ± Std)} & \textbf{Area Range} \\
		\midrule
		Random Convex Polygons & 300 & 3--12 & \(1629.75 \pm 1202.73\) & \([67.60,\ 5452.08]\) \\
		Regular Polygons & 175 & 3--12 & \(3019.91 \pm 1969.22\) & \([225.09,\ 7024.34]\) \\
		Rectangles & 25 & 4 & \(863.38 \pm 476.18\) & \([68.47,\ 1885.76]\) \\
		\bottomrule
	\end{tabular*}
\end{table}

\subsection{Synthetic datasets}

To strengthen generality of conclusions, extensive experiments were carried out across three types of regions to simulate real-world scenarios: rectangles, regular convex polygons, and arbitrary convex polygons. For fairness, we randomly generate a set of 500 convex polygons using Andrew's algorithm in \cite{ANDREW1979216}. The specific details of the dataset are illustrated in Fig.~\ref{fig:6}. The synthetic  dataset includes 300 randomly generated convex polygons, 175 regular convex polygons, and 25 rectangles with different aspect ratios. To account for the irregularity of real-world scenarios, all polygon vertices are represented as floating-point numbers. For example, one of triangles has vertices [(-20.3454,95.8096),(43.766,90.6525),(-9.3902,115.9362)]. The complete experimental datasets, results and Python implementation source code are openly accessible on GitHub: \url{https://github.com/YZP-BUAA/MY-Project}.

\begin{itemize}
	\item  \textbf{Random Convex Polygons (300 samples)}: Vertex counts vary from 3 to 12, providing a comprehensive test for evaluating algorithmic performance under irregular geometric conditions.
	\item  \textbf{Regular Polygons (175 samples)}: Include equilateral triangles, squares, pentagons, and polygons up to 12 sides. These serve as canonical test cases for assessing algorithm behavior under ideal symmetric conditions.
	\item   \textbf{Rectangles (25 samples)}: Represent the simplest convex quadrilateral case, providing baseline performance and enabling analysis of algorithm efficiency under both trivial and extreme geometric conditions.
\end{itemize}

\begin{figure}
	\centering
	\begin{subfigure}[b]{0.36\textwidth}
		\centering
		\includegraphics[width=\textwidth]{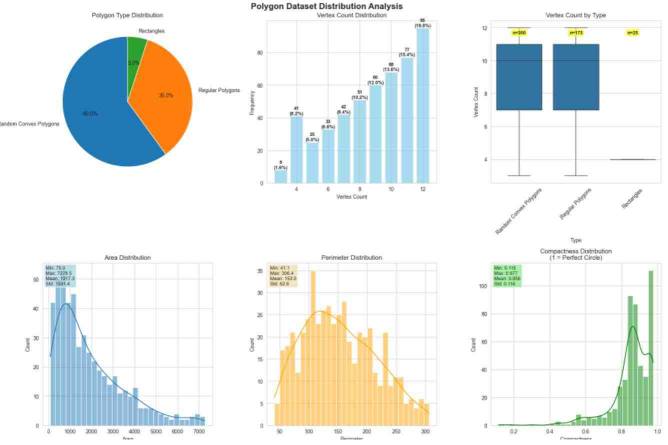}
		\caption{}
	\end{subfigure}
	\begin{subfigure}[b]{0.4\textwidth}
		\centering
		\includegraphics[width=\textwidth]{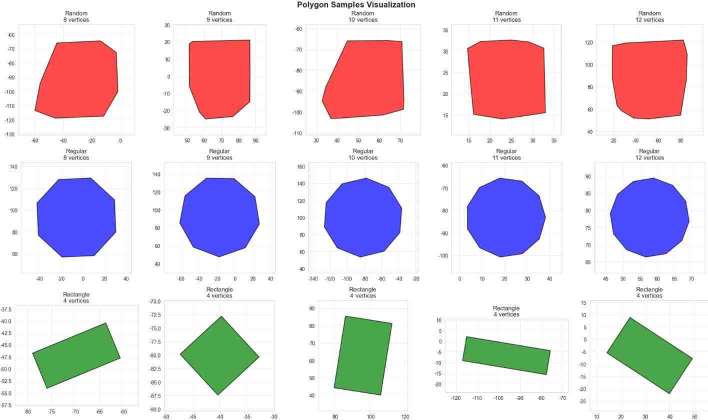}
		\caption{}
	\end{subfigure}
	\caption{(a) Visualization of the benchmark dataset: (1) Proportion of each polygon category; (2) Statistical distribution of polygon properties (area, perimeter, vertex count); (3) Spatial distribution of polygon centroids; (b) Visualization of polygon dataset samples:
		\textbf{Row 1:} Random convex polygons with 8 to 12 vertices (left to right); 
		\textbf{Row 2:} Regular polygons with 8 to 12 vertices (left to right);
		\textbf{Row 3:} Representative rectangle samples including slender rectangles. }
	\label{fig:6}
\end{figure}

According to Table~\ref{tab:5}, a fixed circle radius of 2 is suitable for comparative experiments. This value represents a moderate scale. It allows 100\% coverage to be achieved on some polygons, while on others, full coverage is infeasible and only approximate optimal solutions can be sought. Thus, this setting enables the simulation of both sufficient-resource and constrained-resource coverage scenarios. The number of circles takes values from the set \( N= \{5, 10, 15, 20, 25, 30\}\). Each of the five algorithms was tested on 500 convex polygons per circle count, yielding a total of \( 5 \times 6 \times 500 = 15000 \) experimental trials. The results of the comparative experiments are shown in Figs.~\ref{Fig:7}--\ref{Fig:10}.

Figs.~\ref{Fig:7}--\ref{Fig:8} illustrate the variations of Coverage Rate (CR), Utilization Rate (UR), and Computational Time (CT) with respect to the number of circles. Statistical analysis of the results on the three types of domains reveals that the IQPD algorithm achieves the highest coverage rate among all metaheuristic algorithms, and its coverage rate increases almost linearly with the number of circles. Due to the presence of narrow regions in the rectangular domain, the ICS-MS algorithm exhibits noticeable fluctuations. Moreover, as the number of circles increases, the coverage advantage of IQPD becomes more pronounced. However, on regular polygons, the superiority of IQPD is not significant. As the number of circles grows, overlap among circles becomes inevitable. Consequently, all algorithms except ICGWO show a general downward trend in utilization rate. Nevertheless, IQPD still maintains the highest utilization rate. In terms of computational time, IQPD ranks second only to VGSOK. Furthermore, as the number of circles increases, the computational time of the other algorithms increases substantially, whereas that of IQPD increases only marginally.

\begin{figure}
	\centering
	\begin{subfigure}[b]{0.16\textwidth}
		\centering
		\includegraphics[width=\textwidth]{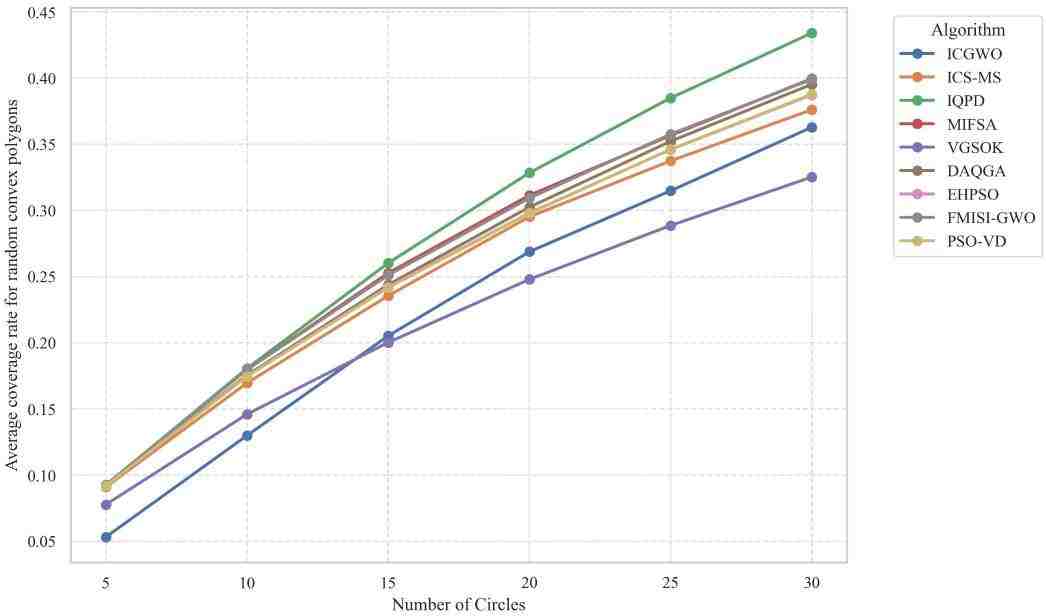}
		\caption{}
	\end{subfigure}
	\begin{subfigure}[b]{0.16\textwidth}
		\centering
		\includegraphics[width=\textwidth]{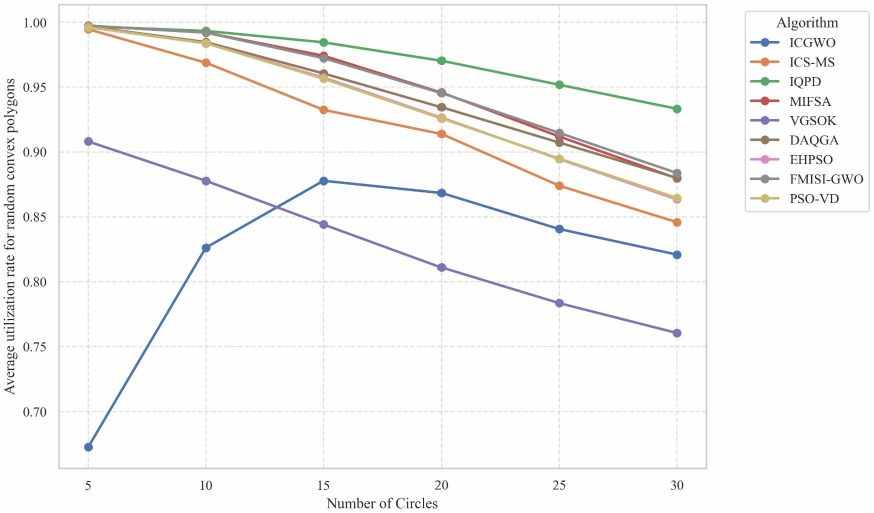}
		\caption{}
	\end{subfigure}
	\begin{subfigure}[b]{0.16\textwidth}
		\centering
		\includegraphics[width=\textwidth]{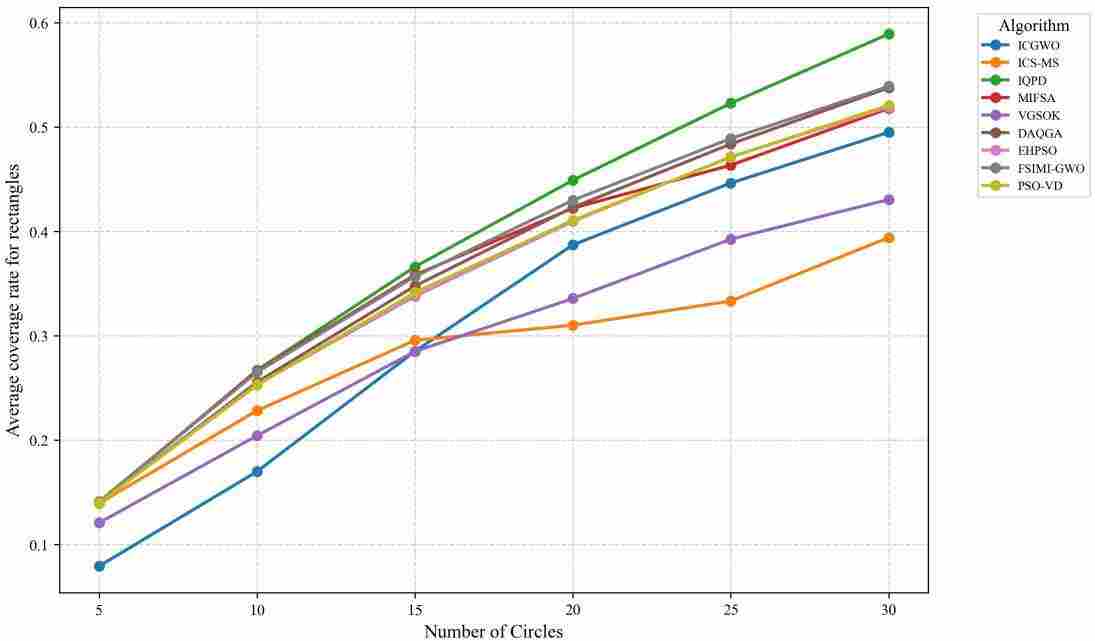}
		\caption{}
	\end{subfigure}
	\begin{subfigure}[b]{0.16\textwidth}
		\centering
		\includegraphics[width=\textwidth]{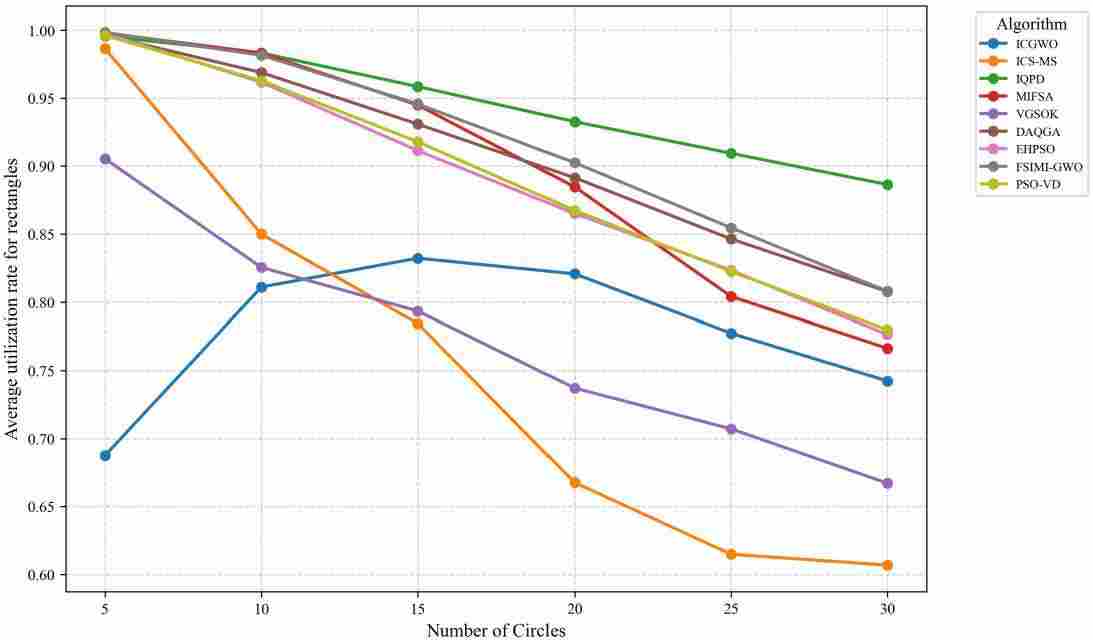}
		\caption{}
	\end{subfigure}       
	\begin{subfigure}[b]{0.16\textwidth}
		\centering
		\includegraphics[width=\textwidth]{ 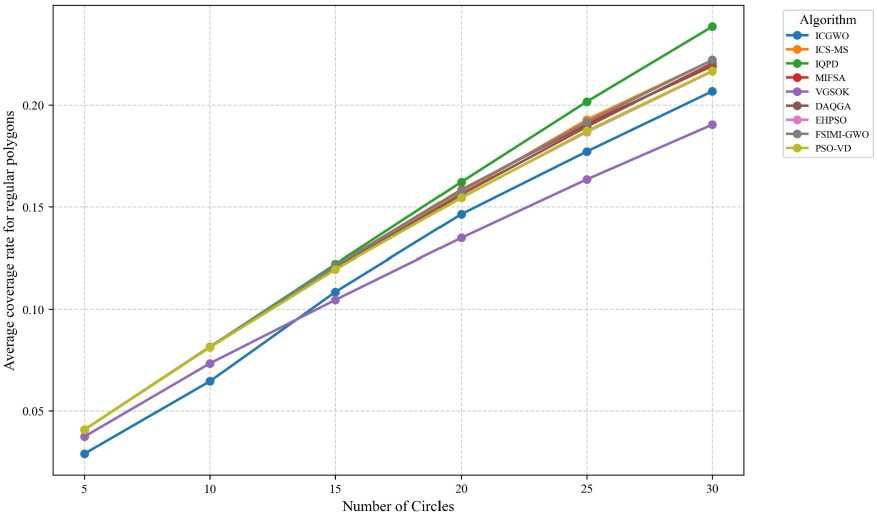}
		\caption{}
	\end{subfigure}
		\begin{subfigure}[b]{0.16\textwidth}
		\centering
		\includegraphics[width=\textwidth]{ 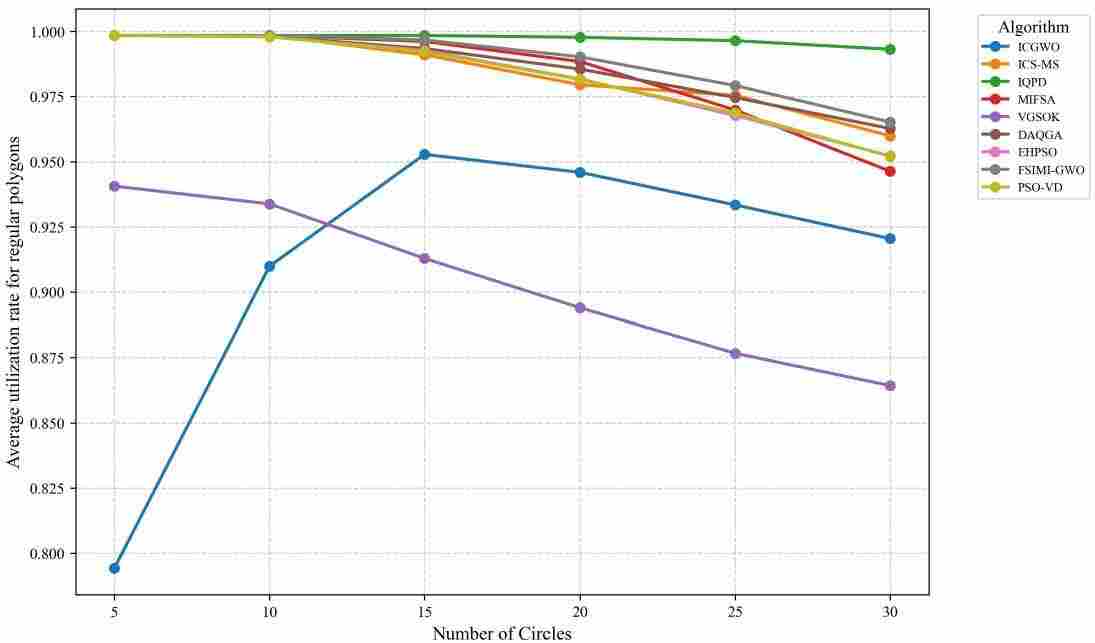}
		\caption{}
	\end{subfigure}
	\caption{Comparative experimental performance of nine algorithms on three types of domains as the number of circles varies, in terms of Coverage Rate (CR) and Utilization Rate (UR).}
	\label{Fig:7}
\end{figure}

\begin{figure}
	\centering
	\begin{subfigure}[b]{0.23\textwidth}
		\centering
		\includegraphics[width=\textwidth]{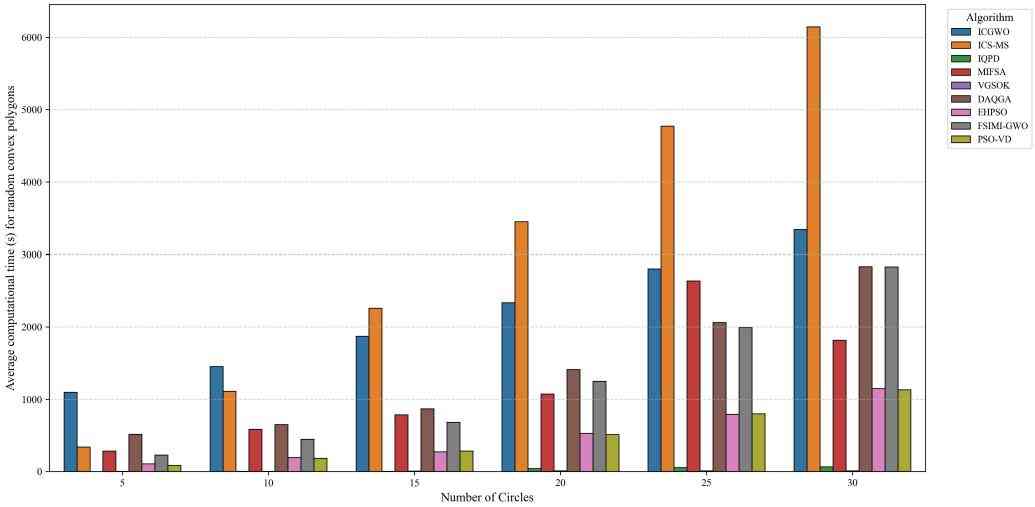}
		\caption{}
	\end{subfigure}
	\begin{subfigure}[b]{0.25\textwidth}
		\centering
		\includegraphics[width=\textwidth]{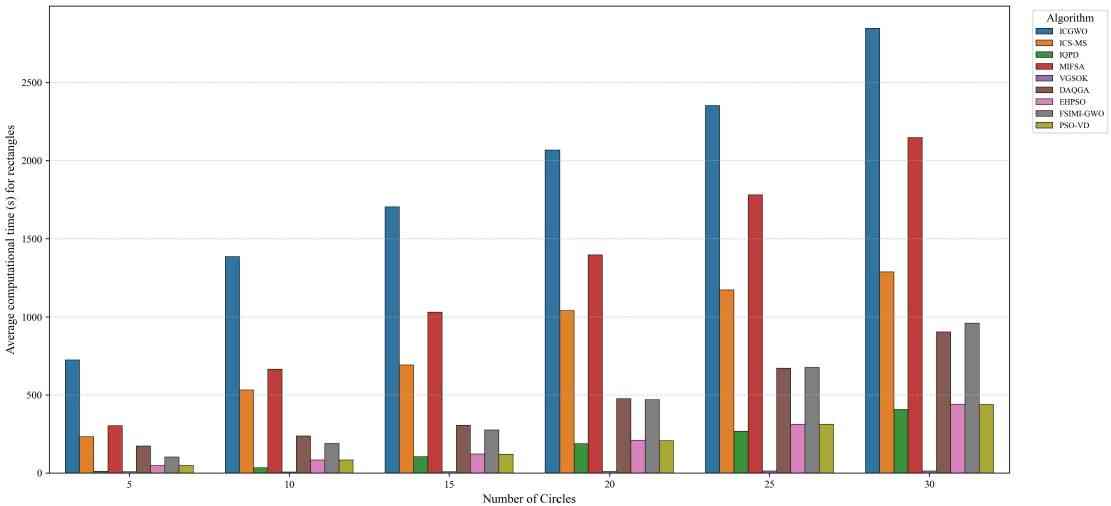}
		\caption{}
	\end{subfigure}
	\begin{subfigure}[b]{0.27\textwidth}
		\centering
		\includegraphics[width=\textwidth]{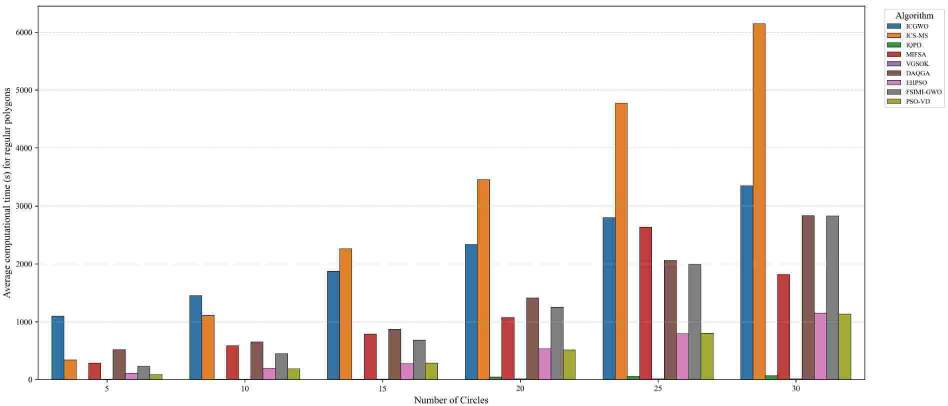}
		\caption{}
	\end{subfigure}
	\caption{Comparative experimental performance of nine algorithms on three types of domains as the number of circles varies, in terms of Computation Time (CT).}
	\label{Fig:8}
\end{figure}

\begin{figure}
	\centering
	\begin{subfigure}[b]{0.39\textwidth}
		\centering
		\includegraphics[width=\textwidth]{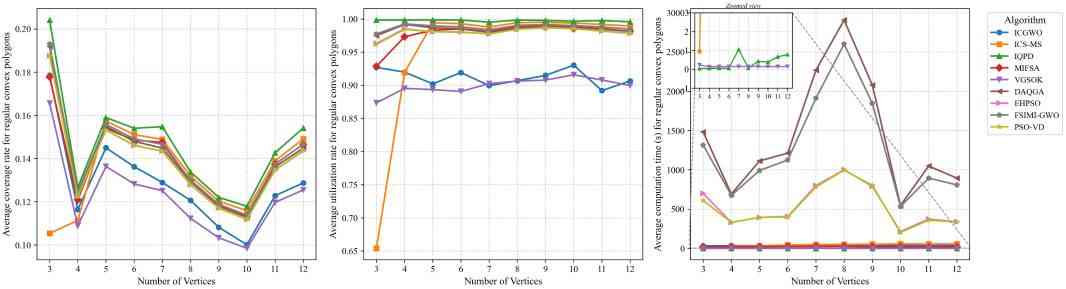}
		\caption{}
	\end{subfigure}
	\begin{subfigure}[b]{0.39\textwidth}
		\centering
		\includegraphics[width=\textwidth]{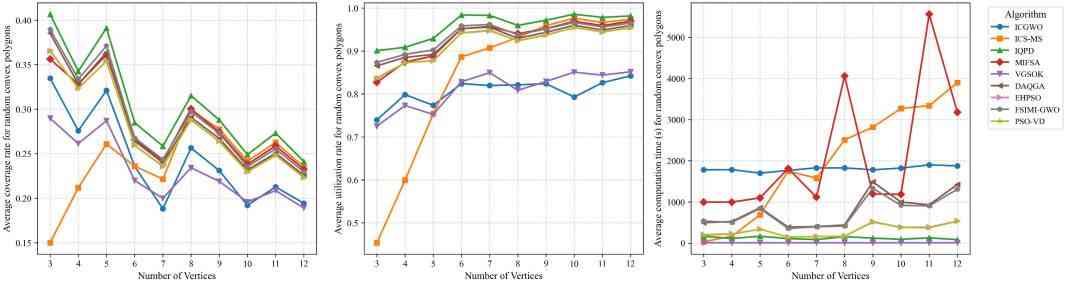}
		\caption{}
	\end{subfigure}
	\caption{Comparative experimental performance of nine algorithms on two types of domains as the number of polygonal vertices increases, in terms of three metrics (CR, UR, CT).}
	\label{Fig:9}
\end{figure}

Fig.~\ref{Fig:9} illustrates the overall variation of the three metrics (CR, UR, CT) with respect to the number of vertices on arbitrary convex polygons and regular polygons. It can be observed that the IQPD algorithm still achieves the best performance in terms of Coverage Rate (CR) and Utilization Rate (UR) across polygons with different vertex counts, while its computational time (CT) is second only to VGSOK. The coverage rate is significantly affected by the number of vertices, whereas the utilization rate  shows only a minor dependence. Moreover, the computation time of several comparative algorithms fluctuates significantly with the number of vertices, exhibiting instability and making them less adaptable to complex real-world scenarios.

\begin{table}
	\caption{Overall performance comparison of nine algorithms on random convex polygonal domains}
	\label{tab:6}
	\begin{tabular*}{\tblwidth}{@{\extracolsep{\fill}}lccccc@{}}
		\toprule
		\textbf{Algorithm} & \textbf{Coverage Rate} & \textbf{Utilization Rate} & \textbf{Uniformity Index} & \textbf{Min Gap} & \textbf{Computation Time (s)} \\
		\midrule
		ICGWO    & 0.2224 & 0.8177 & 0.7957 & -1.0384 & 1828.728 \\
		ICM-MS   & 0.2508 & 0.9216 & 0.5371 & 0.0008  & 2711.502 \\
		MIFSA    & 0.2654 & 0.9501 & 0.5816 & 0.1512  & 2621.94  \\
		VGSOK    & 0.2143 & 0.8308 & 0.5530  & -1.8341 & 8.19     \\
		DAQGA    & 0.2602 & 0.9439 & 0.7156 & -0.2321 & 944.9391 \\
		EHPSO    & 0.2567 & 0.9370  & 0.7047 & -0.2888 & 360.6901 \\
		FSIMI-GWO & 0.2649 & 0.9509 & 0.7368 & 0.0571  & 880.2540 \\
		PSO-VD   & 0.2566 & 0.9369 & 0.7064 & -0.2971 & 355.5338 \\
		\textbf{IQPD}     & \textbf{0.2801} & \textbf{0.9717} & 0.5871 & 0.1983  & \textbf{113.376}  \\
		\bottomrule
	\end{tabular*}
\end{table}

\begin{table}
	\caption{Overall performance comparison of nine algorithms on rectangular domains}
	\label{tab:7}
	\begin{tabular*}{\tblwidth}{@{\extracolsep{\fill}}lccccc@{}}
		\toprule
		\textbf{Algorithm} & \textbf{Coverage Rate} & \textbf{Utilization Rate} & \textbf{Uniformity Index} & \textbf{Min Gap} & \textbf{Computation Time (s)} \\
		\midrule
		ICGWO     & 0.3104 & 0.7785 & 0.4006 & -1.9302 & 1847.082 \\
		ICM-MS    & 0.2834 & 0.7517 & 0.4747 & -0.2652 & 827.274  \\
		MIFSA     & 0.3615 & 0.8968 & 0.5950 & -0.5993 & 1220.568 \\
		VGSOK     & 0.2948 & 0.7726 & 0.5569 & -2.3017 & 9.546    \\
		DAQGA      & 0.3646 & 0.9069 & 0.7470 & -0.8566 & 461.765  \\
		EHPSO     & 0.3553 & 0.8892 & 0.7370 & -1.0773 & 203.285  \\
		FSIMI-GWO & 0.3701 & 0.9151 & 0.7639 & -0.7434 & 445.625  \\
		PSO-VD    & 0.3561 & 0.8911 & 0.7262 & -1.0396 & 201.879  \\
	\textbf{IQPD}  & \textbf{0.3890} & \textbf{0.9443} & 0.6106 & -0.1172 & \textbf{169.098}  \\
		\bottomrule
	\end{tabular*}
\end{table}

\begin{table}
	\caption{Overall performance comparison of nine algorithms on regular domains}
	\label{tab:8}
	\begin{tabular*}{\tblwidth}{@{\extracolsep{\fill}}lccccc@{}}
		\toprule
		\textbf{Algorithm} & \textbf{Coverage Rate} & \textbf{Utilization Rate} & \textbf{Uniformity Index} & \textbf{Min Gap} & \textbf{Computation Time (s)} \\
		\midrule
		ICGWO     & 0.1220  & 0.9095  & 0.4278 & 6.1489  & 2151.27  \\
		ICM-MS    & 0.1357  & 0.9838  & 0.4921 & 0.2323  & 3013.356 \\
		MIFSA     & 0.1353  & 0.9829  & 0.5216 & 0.9346  & 1197.384 \\
		VGSOK     & 0.1173  & 0.9037  & 0.5569 & -0.9409 & 8.622    \\
		DAQGA     & 0.1346  & 0.9854  & 0.6794 & 0.7754  & 1391.188 \\
		EHPSO     & 0.1332  & 0.9817  & 0.6674 & 0.9200  & 507.7377 \\
		FSIMI-GWO & 0.1358  & 0.9880  & 0.6914 & 1.0387  & 1238.325 \\
		PSO-VD    & 0.1333  & 0.9818  & 0.6615 & 0.8171  & 500.2203 \\
		\textbf{IQPD}      & \textbf{0.1411}  & \textbf{0.9971}  & 0.5264 & 0.8767  & \textbf{28.83}    \\
		\bottomrule
	\end{tabular*}
\end{table}

Fig.~\ref{Fig:10} presents the overall distribution uniformity (Uniformity Index) and the minimum coverage gap (Minimum Gap) for all algorithms. While IQPD does not yet achieve the highest uniformity, it attains the second-highest Minimum Gap density, following only ICS-MS. This observation suggests that IQPD is capable of mitigating excessive coverage overlap and minimizing uncovered areas, consequently enhancing the overall coverage efficiency.

The overall performance of the algorithms across all metrics is presented in the Tables~\ref{tab:6},~\ref{tab:7} and~\ref{tab:8}. Based on multiple repeated experiments, the proposed IQPD algorithm achieves the best coverage rate and utilization rate, while significantly reducing computational time and minimizing coverage blind spots. However, its performance in terms of coverage uniformity still requires further improvement. Corresponding to real-world applications such as wireless network coverage or agricultural irrigation, our algorithm is capable of adapting to complex convex polygonal regions, expanding coverage range, improving node utilization efficiency, and reducing optimization time greatly.  To present the optimization results of different algorithms  better, we provide optimized coverage figures for several regions with complex shapes, along with the corresponding algorithmic iteration convergence curves. 

\begin{figure}
	\centering
	\begin{subfigure}[b]{0.4\textwidth}
		\centering
		\includegraphics[width=\textwidth]{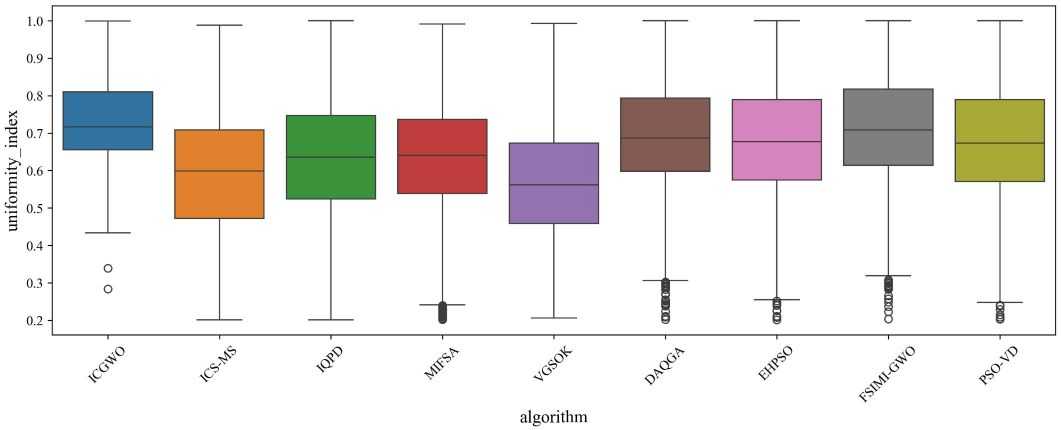}
		\caption{}
	\end{subfigure}
	\begin{subfigure}[b]{0.35\textwidth}
		\centering
		\includegraphics[width=\textwidth]{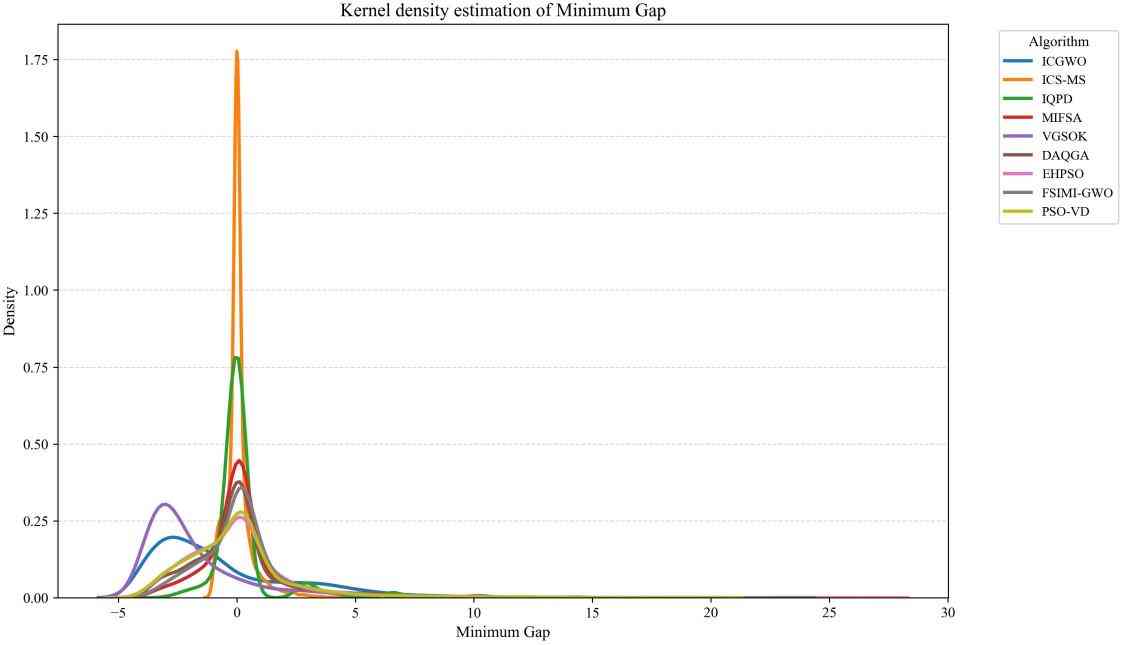}
		\caption{}
	\end{subfigure}
	\caption{Overall performance of nine algorithms in terms of Uniformity Index (UI) and Minimum Gap (MG).}
	\label{Fig:10}
\end{figure}

To clearly demonstrate the optimized deployment results, we present several convex polygons with extreme shapes (with vertex counts ranging from 4 to 8 as examples), as shown in Figs.~\ref{Fig:14}–\ref{Fig:18}. The convergence curves demonstrate that, even on these highly irregular convex polygonal regions, our algorithm consistently maintains the fastest convergence speed, second only to VGSOK, while simultaneously achieving the highest coverage rate. The remaining results are publicly accessible on GitHub: \url{https://github.com/YZP-BUAA/MY-Project}.

\begin{figure}
	\centering
	\begin{subfigure}[b]{0.18\textwidth}
		\centering
		\includegraphics[width=\textwidth]{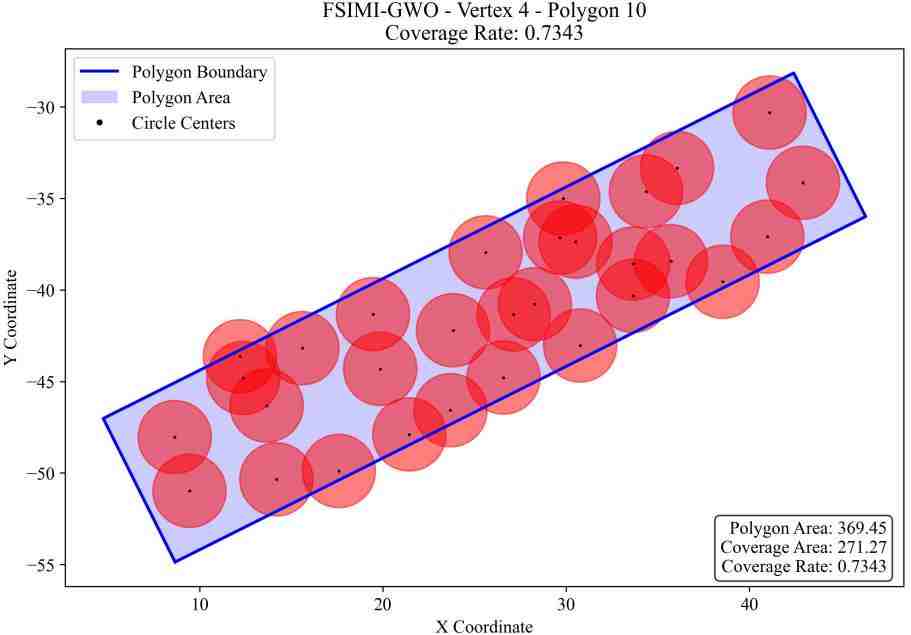}
	\end{subfigure}
	\begin{subfigure}[b]{0.18\textwidth}
		\centering
		\includegraphics[width=\textwidth]{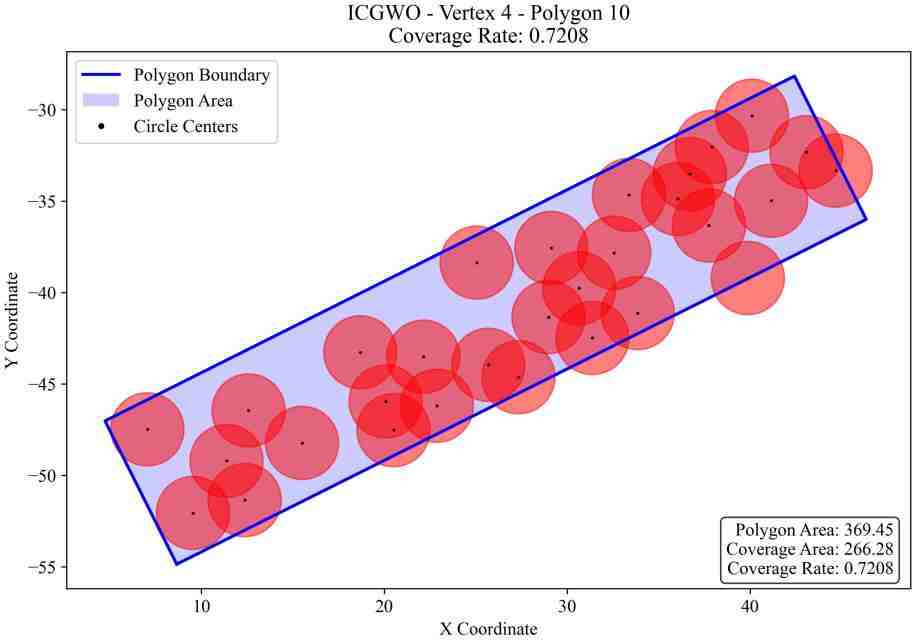}
	\end{subfigure}
		\begin{subfigure}[b]{0.18\textwidth}
		\centering
		\includegraphics[width=\textwidth]{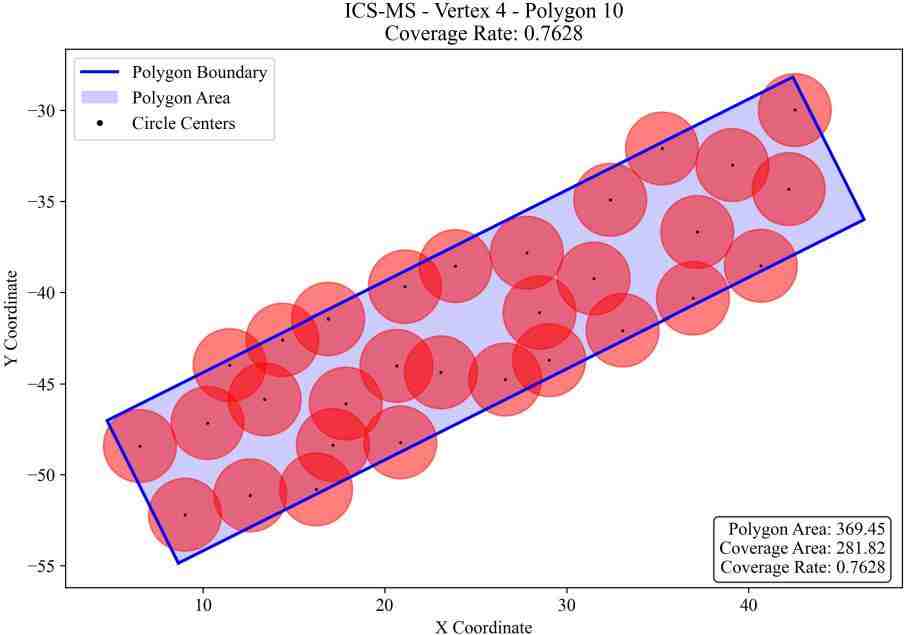}
	\end{subfigure}
		\begin{subfigure}[b]{0.18\textwidth}
		\centering
		\includegraphics[width=\textwidth]{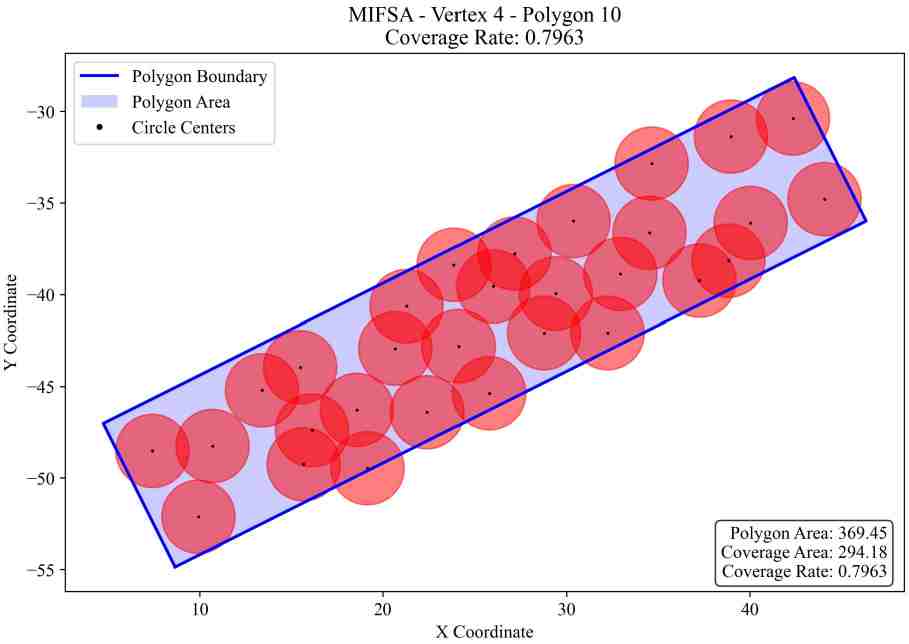}
	\end{subfigure}
		\begin{subfigure}[b]{0.18\textwidth}
		\centering
		\includegraphics[width=\textwidth]{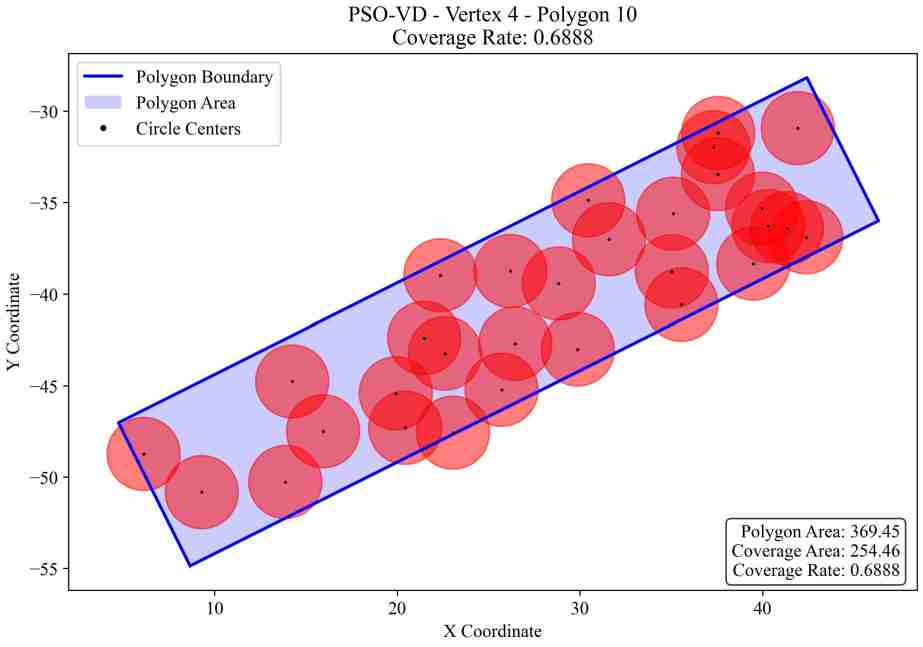}
	\end{subfigure}
		\begin{subfigure}[b]{0.18\textwidth}
		\centering
		\includegraphics[width=\textwidth]{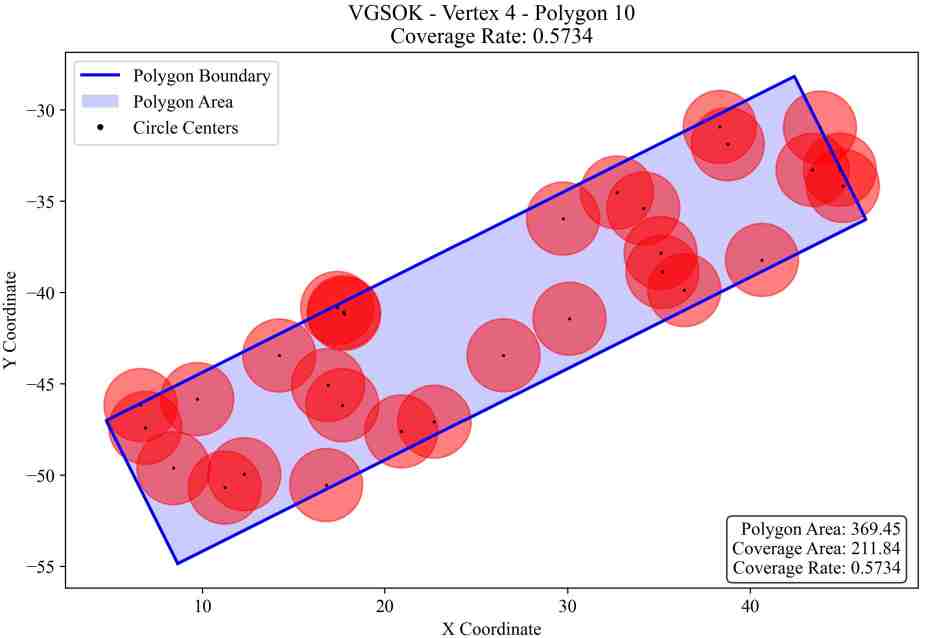}
	\end{subfigure}
		\begin{subfigure}[b]{0.18\textwidth}
		\centering
		\includegraphics[width=\textwidth]{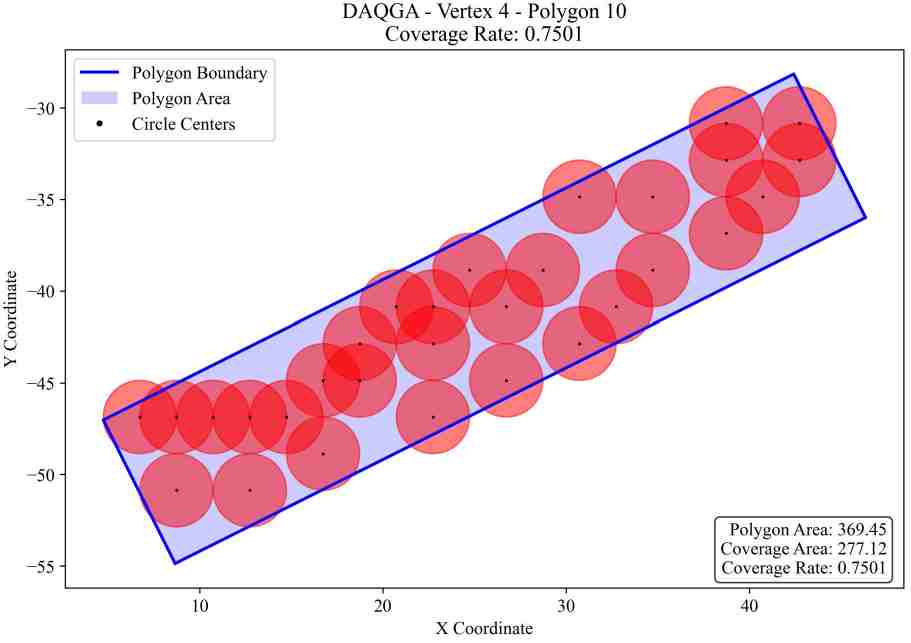}
	\end{subfigure}
		\begin{subfigure}[b]{0.18\textwidth}
		\centering
		\includegraphics[width=\textwidth]{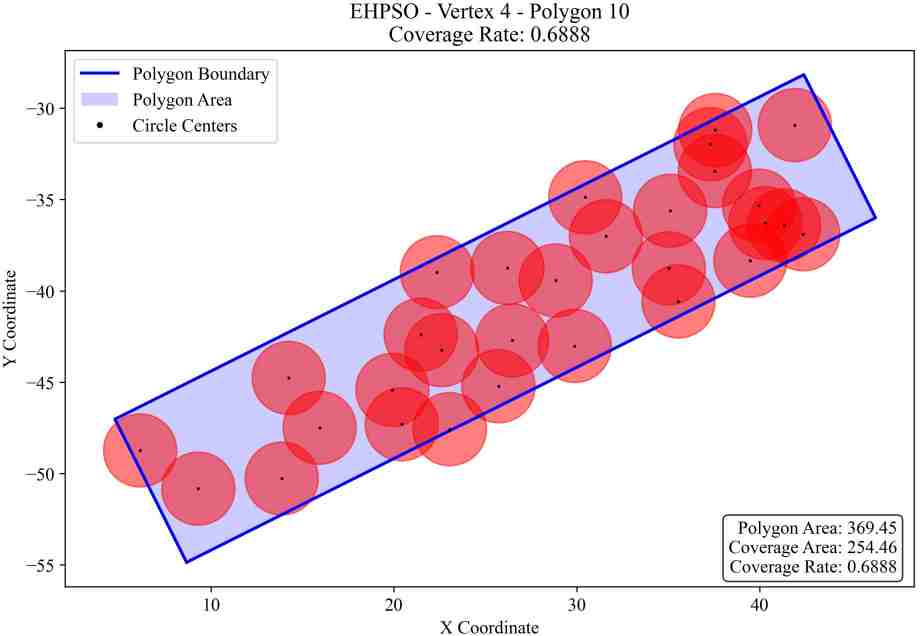}
	\end{subfigure}
		\begin{subfigure}[b]{0.18\textwidth}
		\centering
		\includegraphics[width=\textwidth]{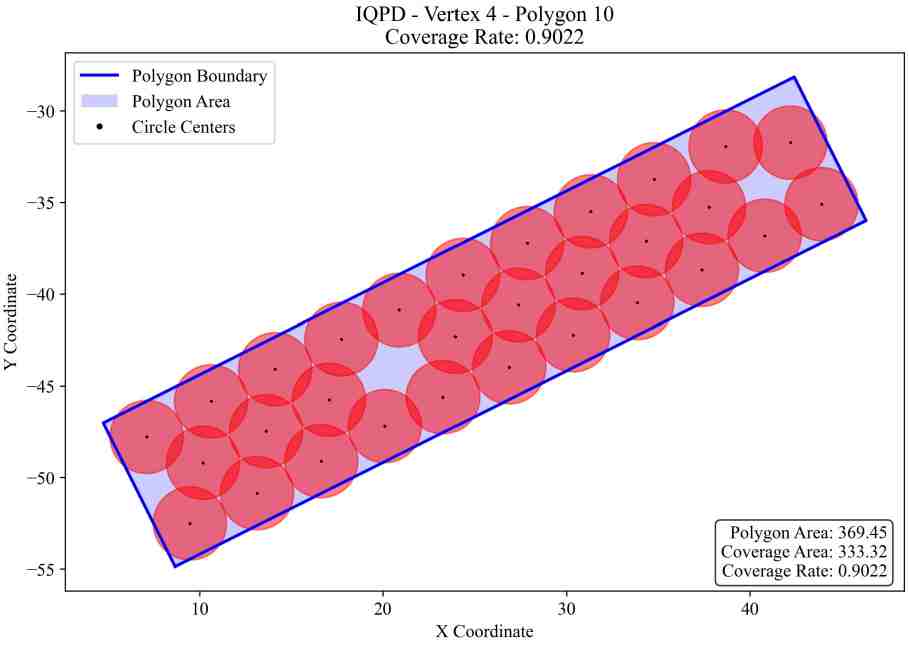}
	\end{subfigure}
		\begin{subfigure}[b]{0.18\textwidth}
		\centering
		\includegraphics[width=\textwidth]{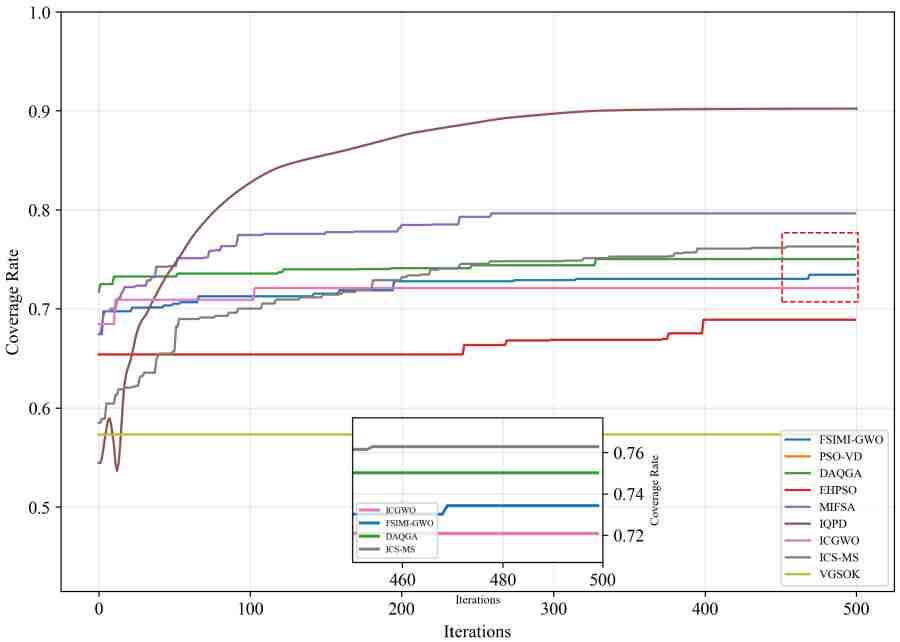}
	\end{subfigure}
	\caption{Optimized deployment layout and convergence curves for the case with 4 vertices and 30 circles.}
	\label{Fig:14}
\end{figure}

\begin{figure}
	\centering
	\begin{subfigure}[b]{0.13\textwidth}
		\centering
		\includegraphics[width=\textwidth]{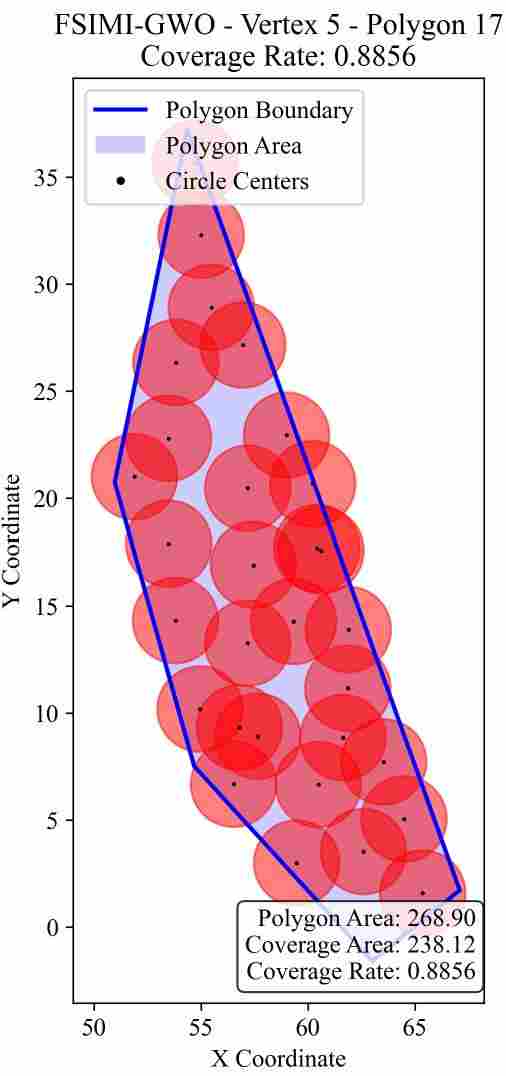}
	\end{subfigure}
	\begin{subfigure}[b]{0.125\textwidth}
		\centering
		\includegraphics[width=\textwidth]{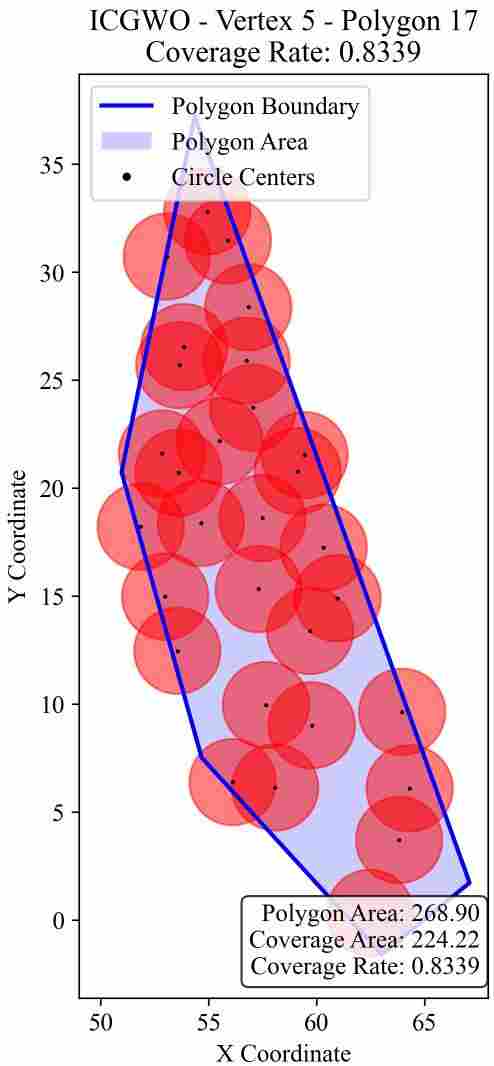}
	\end{subfigure}
	\begin{subfigure}[b]{0.125\textwidth}
		\centering
		\includegraphics[width=\textwidth]{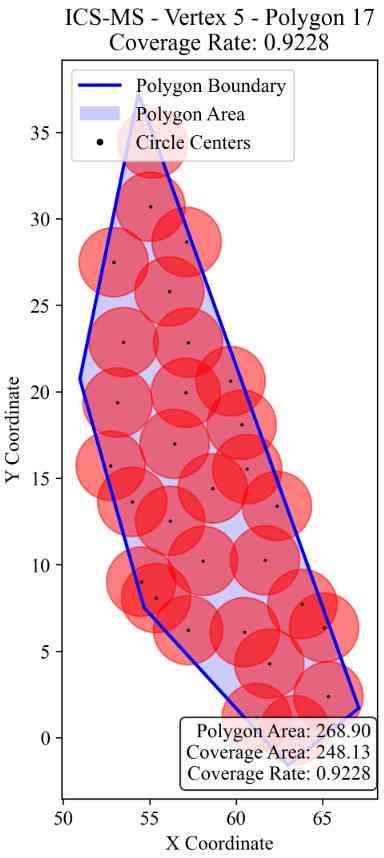}
	\end{subfigure}
	\begin{subfigure}[b]{0.13\textwidth}
		\centering
		\includegraphics[width=\textwidth]{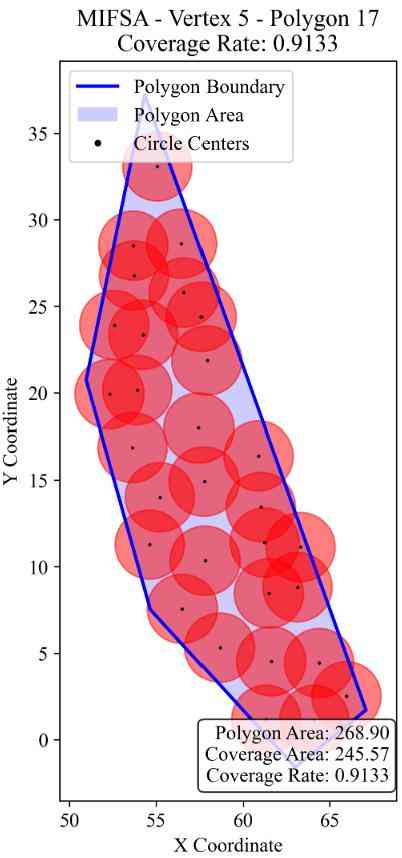}
	\end{subfigure}
	\begin{subfigure}[b]{0.13\textwidth}
		\centering
		\includegraphics[width=\textwidth]{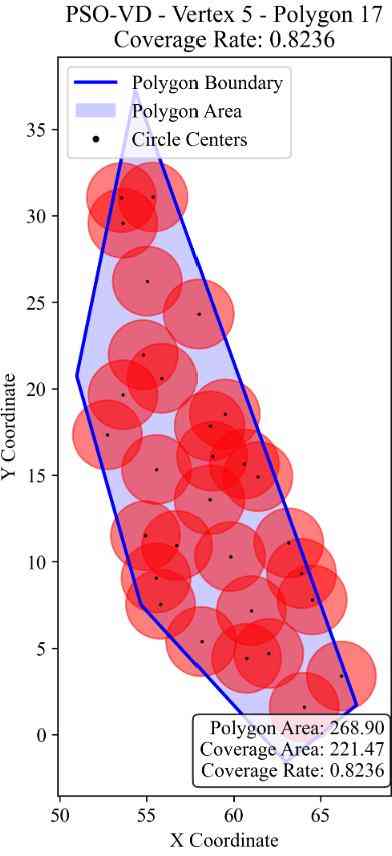}
	\end{subfigure}
	\begin{subfigure}[b]{0.13\textwidth}
		\centering
		\includegraphics[width=\textwidth]{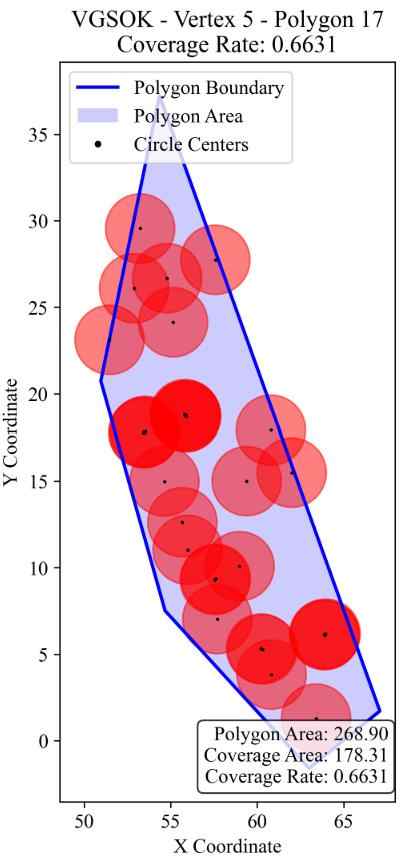}
	\end{subfigure}
	\begin{subfigure}[b]{0.125\textwidth}
		\centering
		\includegraphics[width=\textwidth]{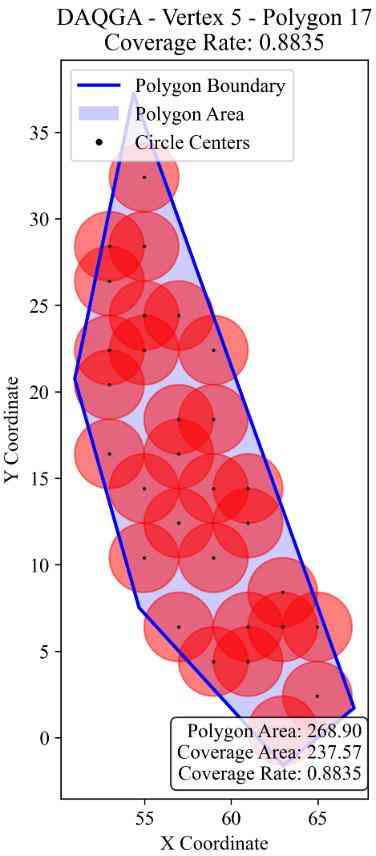}
	\end{subfigure}
	\begin{subfigure}[b]{0.125\textwidth}
		\centering
		\includegraphics[width=\textwidth]{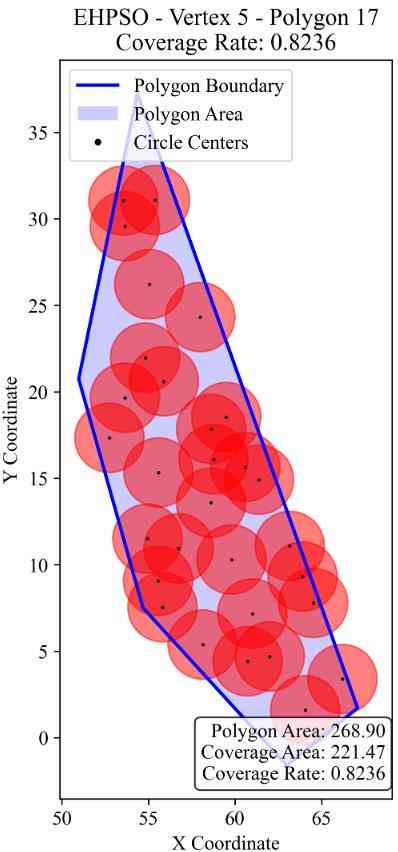}
	\end{subfigure}
	\begin{subfigure}[b]{0.12\textwidth}
		\centering
		\includegraphics[width=\textwidth]{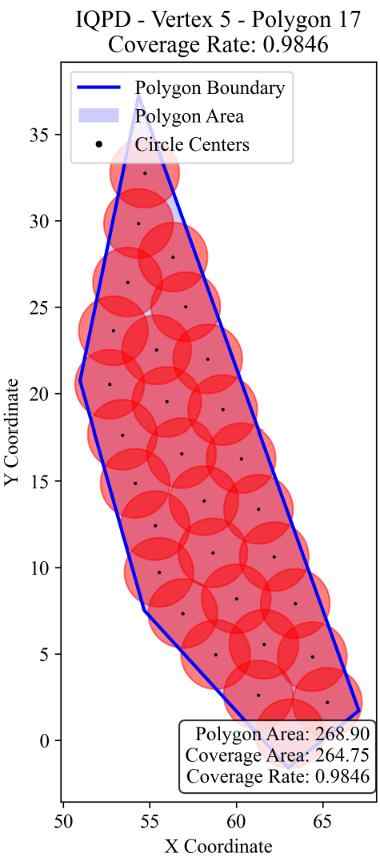}
	\end{subfigure}
	\begin{subfigure}[b]{0.25\textwidth}
		\centering
		\includegraphics[width=\textwidth]{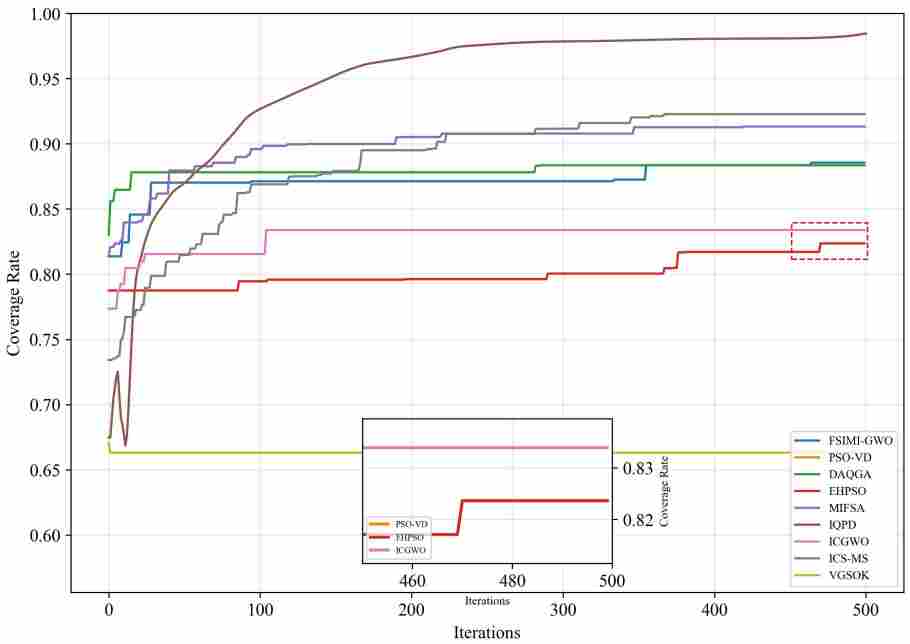}
	\end{subfigure}
	\caption{Optimized deployment layout and convergence curves for the case with 5 vertices and 30 circles.}
	\label{Fig:15}
\end{figure}

\begin{figure}
	\centering
	\begin{subfigure}[b]{0.15\textwidth}
		\centering
		\includegraphics[width=\textwidth]{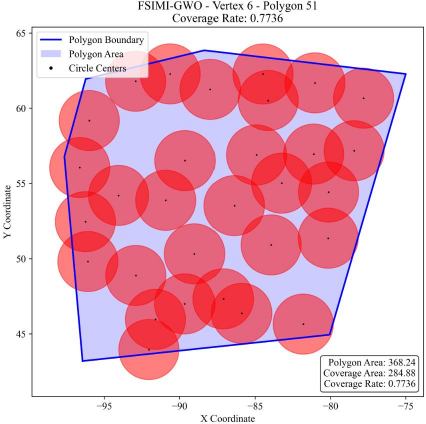}
	\end{subfigure}
	\begin{subfigure}[b]{0.15\textwidth}
		\centering
		\includegraphics[width=\textwidth]{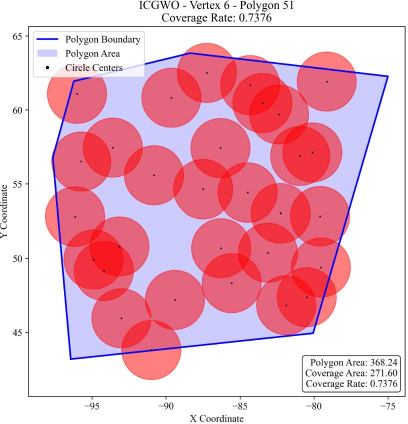}
	\end{subfigure}
	\begin{subfigure}[b]{0.15\textwidth}
		\centering
		\includegraphics[width=\textwidth]{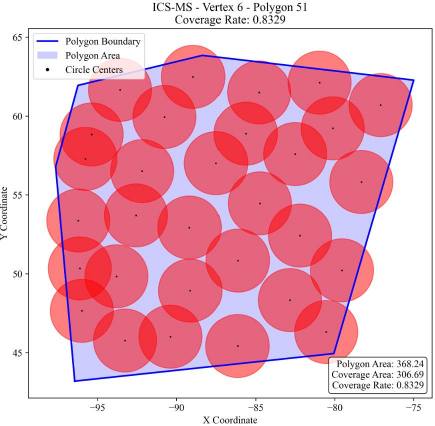}
	\end{subfigure}
	\begin{subfigure}[b]{0.15\textwidth}
		\centering
		\includegraphics[width=\textwidth]{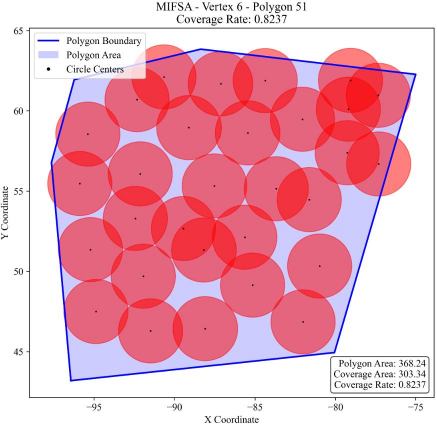}
	\end{subfigure}
	\begin{subfigure}[b]{0.15\textwidth}
		\centering
		\includegraphics[width=\textwidth]{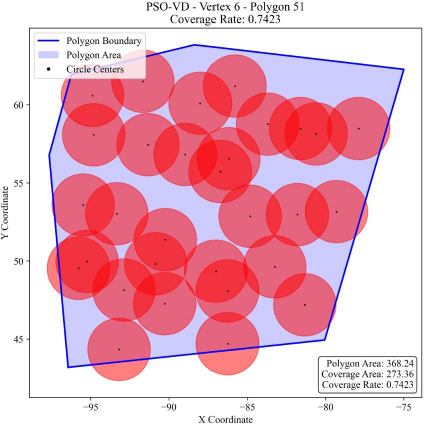}
	\end{subfigure}
	\begin{subfigure}[b]{0.15\textwidth}
		\centering
		\includegraphics[width=\textwidth]{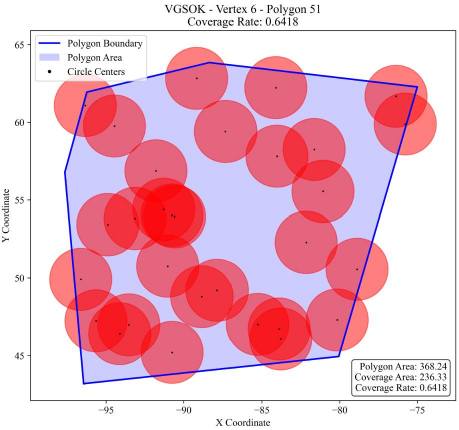}
	\end{subfigure}
	\begin{subfigure}[b]{0.15\textwidth}
		\centering
		\includegraphics[width=\textwidth]{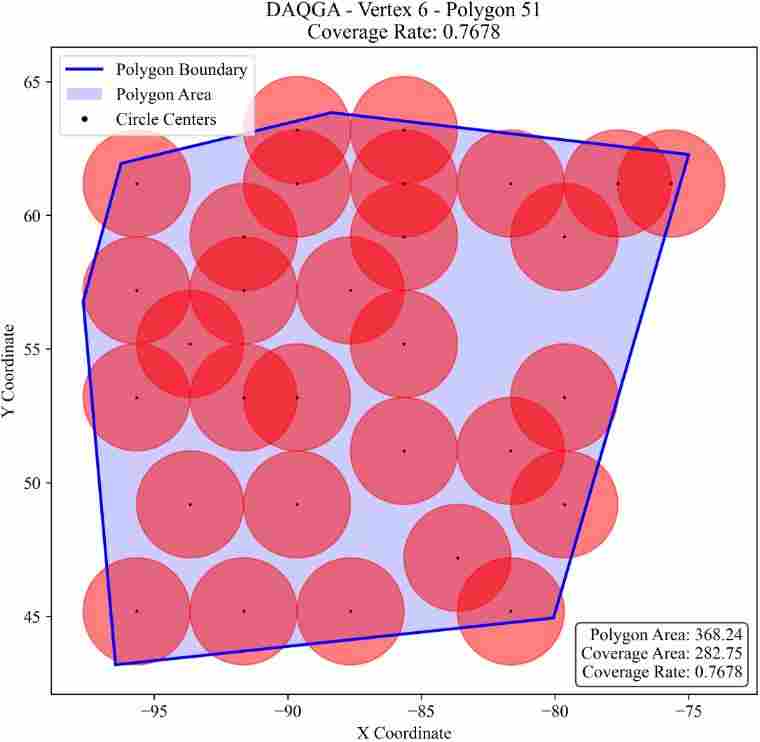}
	\end{subfigure}
	\begin{subfigure}[b]{0.15\textwidth}
		\centering
		\includegraphics[width=\textwidth]{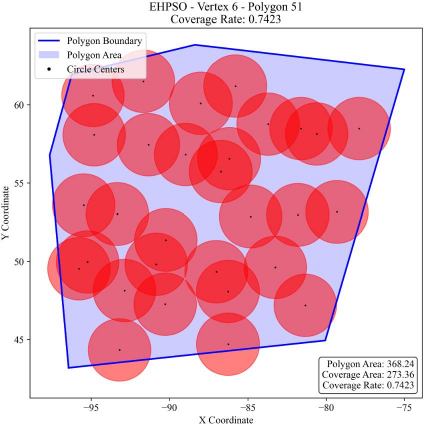}
	\end{subfigure}
	\begin{subfigure}[b]{0.15\textwidth}
		\centering
		\includegraphics[width=\textwidth]{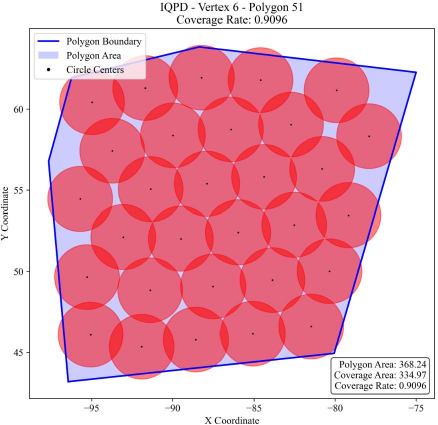}
	\end{subfigure}
	\begin{subfigure}[b]{0.2\textwidth}
		\centering
		\includegraphics[width=\textwidth]{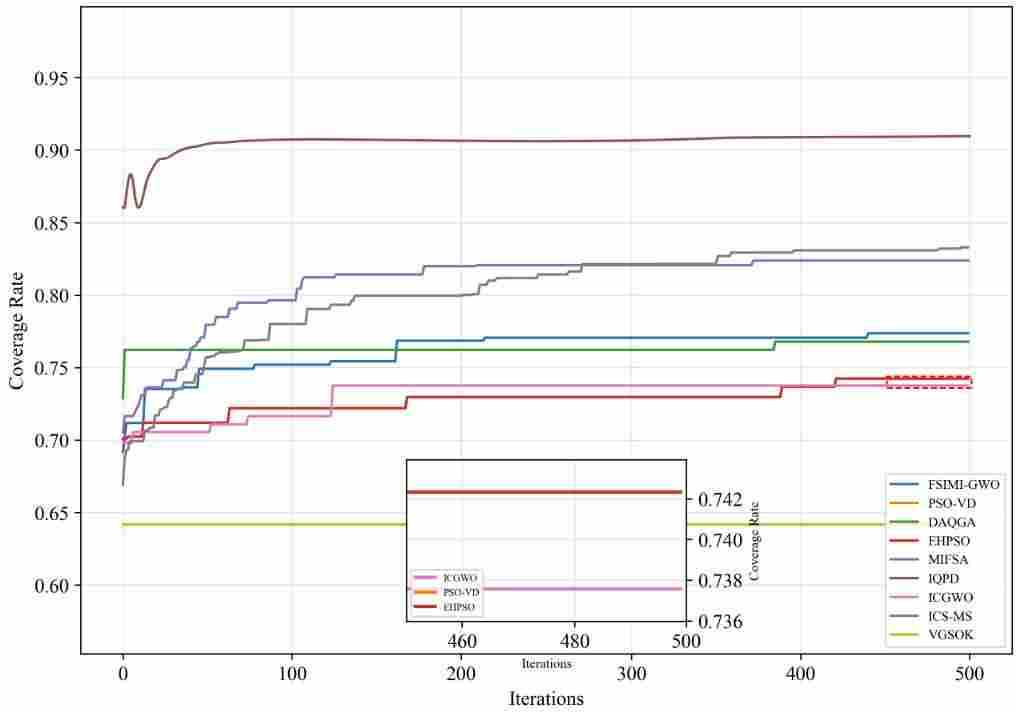}
	\end{subfigure}
	\caption{Optimized deployment layout and convergence curves for the case with 6 vertices and 30 circles.}
	\label{Fig:16}
\end{figure}

\begin{figure}
	\centering
	\begin{subfigure}[b]{0.15\textwidth}
		\centering
		\includegraphics[width=\textwidth]{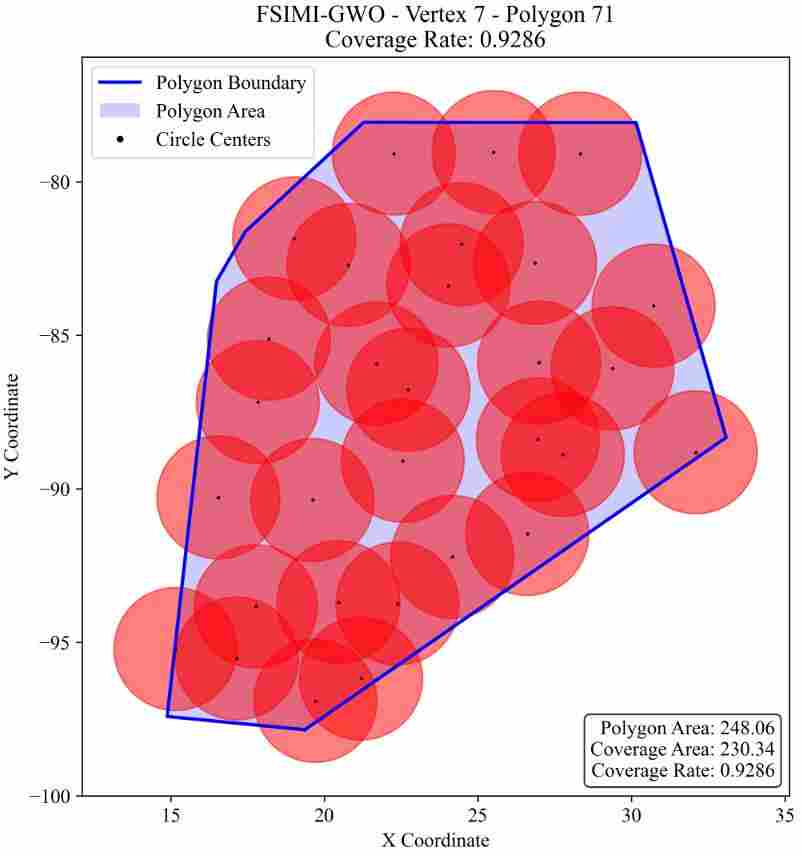}
	\end{subfigure}
	\begin{subfigure}[b]{0.15\textwidth}
		\centering
		\includegraphics[width=\textwidth]{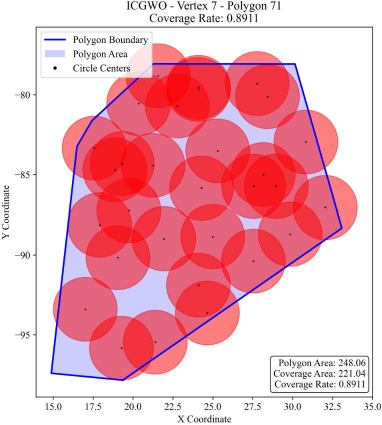}
	\end{subfigure}
	\begin{subfigure}[b]{0.15\textwidth}
		\centering
		\includegraphics[width=\textwidth]{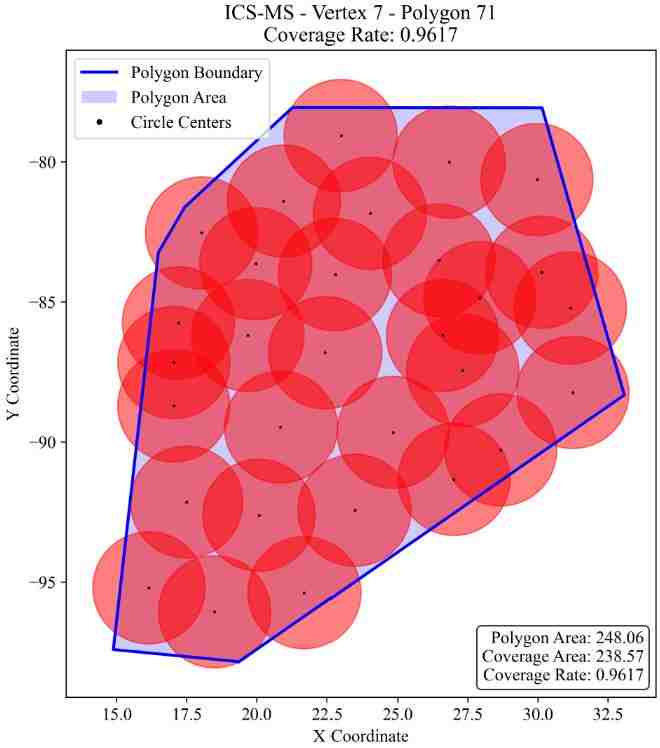}
	\end{subfigure}
	\begin{subfigure}[b]{0.16\textwidth}
		\centering
		\includegraphics[width=\textwidth]{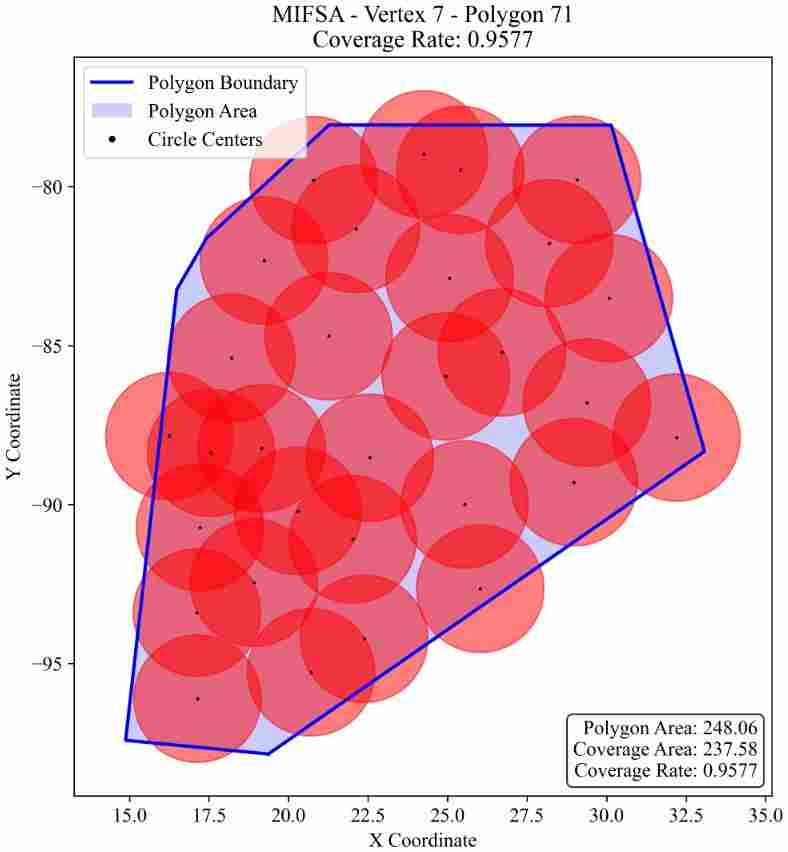}
	\end{subfigure}
	\begin{subfigure}[b]{0.15\textwidth}
		\centering
		\includegraphics[width=\textwidth]{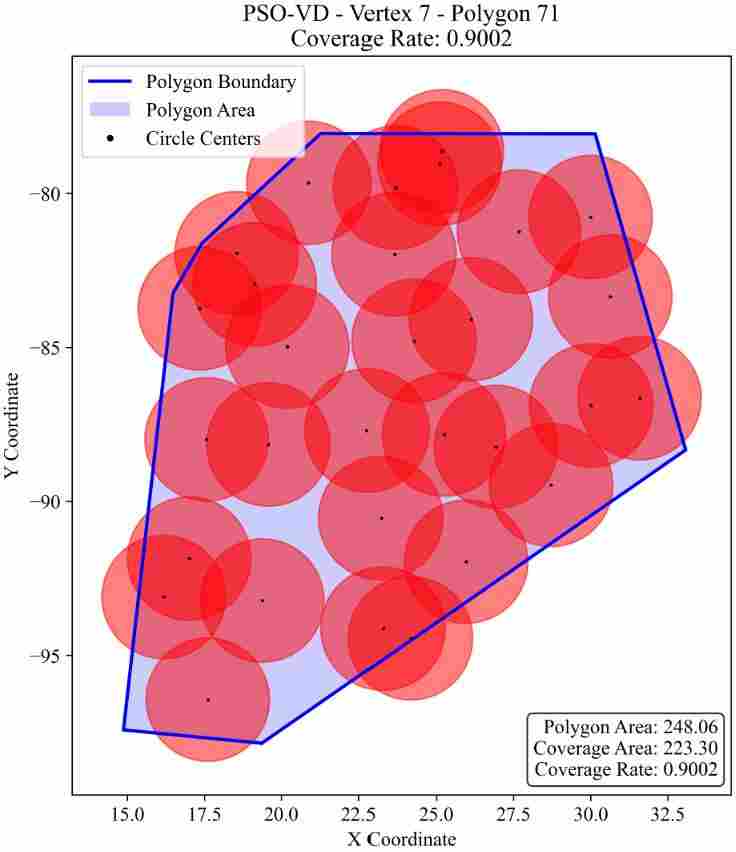}
	\end{subfigure}
	\begin{subfigure}[b]{0.15\textwidth}
		\centering
		\includegraphics[width=\textwidth]{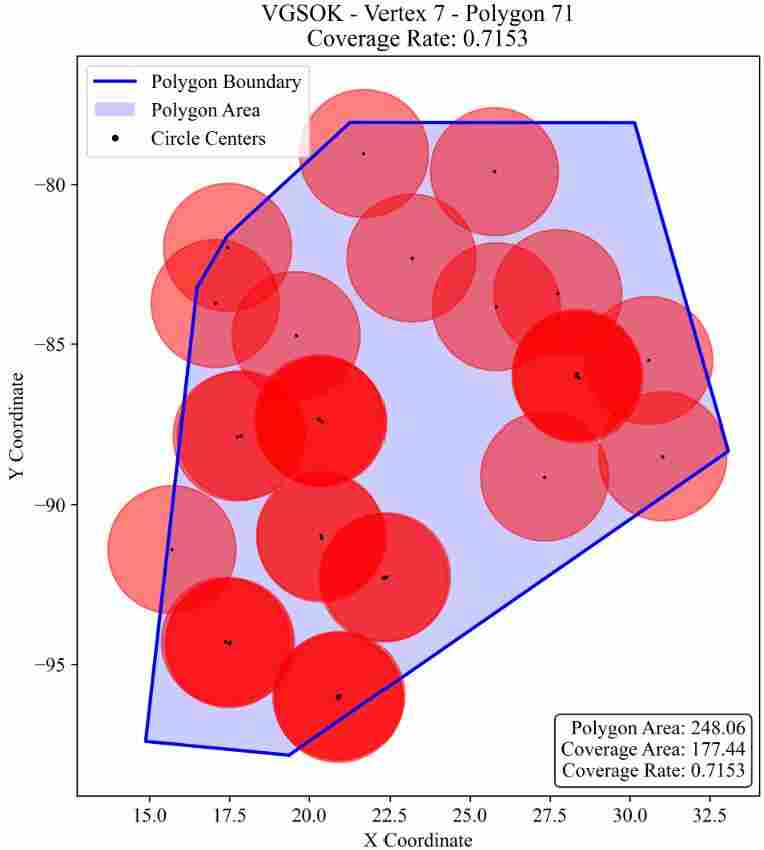}
	\end{subfigure}
	\begin{subfigure}[b]{0.15\textwidth}
		\centering
		\includegraphics[width=\textwidth]{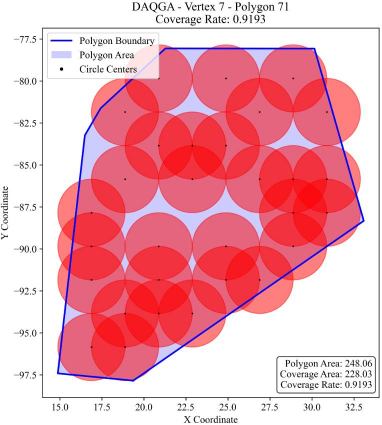}
	\end{subfigure}
	\begin{subfigure}[b]{0.15\textwidth}
		\centering
		\includegraphics[width=\textwidth]{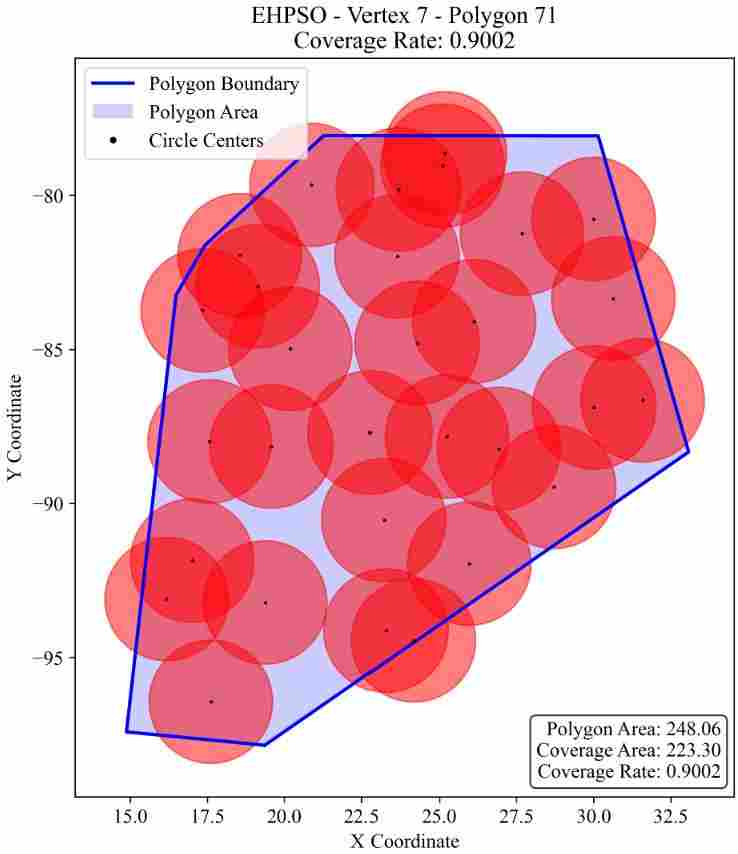}
	\end{subfigure}
	\begin{subfigure}[b]{0.15\textwidth}
		\centering
		\includegraphics[width=\textwidth]{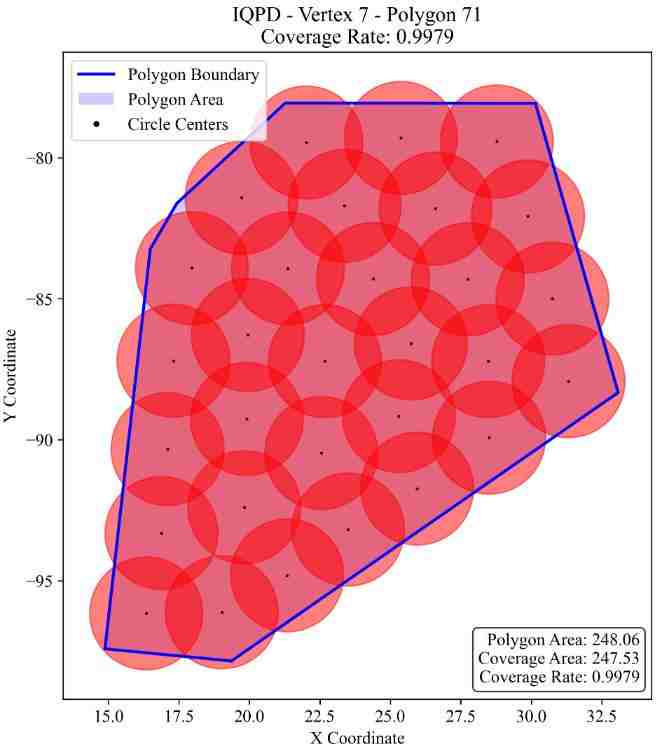}
	\end{subfigure}
	\begin{subfigure}[b]{0.2\textwidth}
		\centering
		\includegraphics[width=\textwidth]{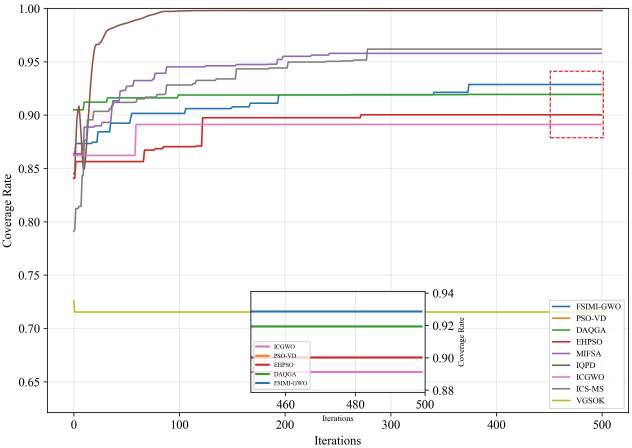}
	\end{subfigure}
	\caption{Optimized deployment layout and convergence curves for the case with 7 vertices and 30 circles.}
	\label{Fig:17}
\end{figure}

\begin{figure}
	\centering
	\begin{subfigure}[b]{0.15\textwidth}
		\centering
		\includegraphics[width=\textwidth]{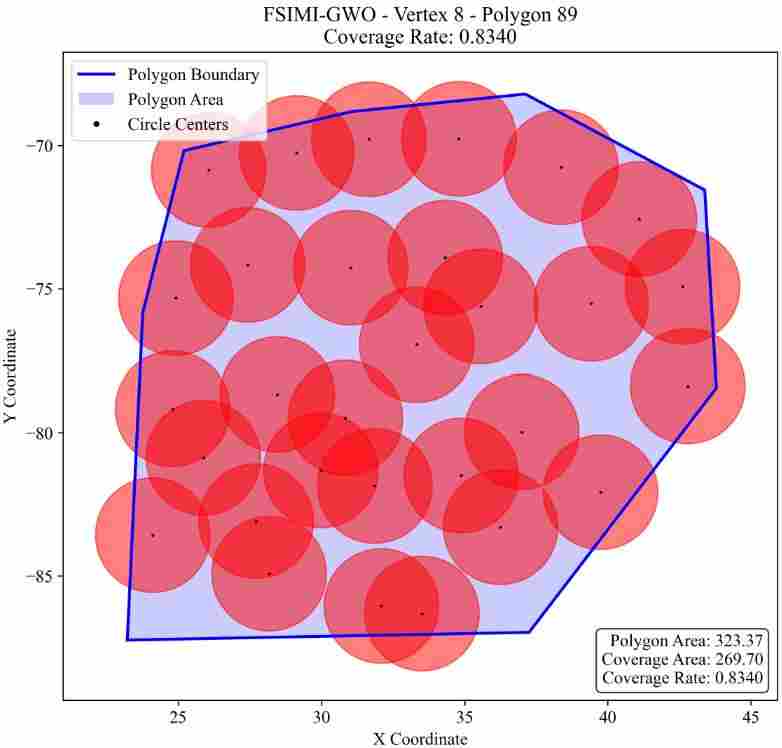}
	\end{subfigure}
	\begin{subfigure}[b]{0.15\textwidth}
		\centering
		\includegraphics[width=\textwidth]{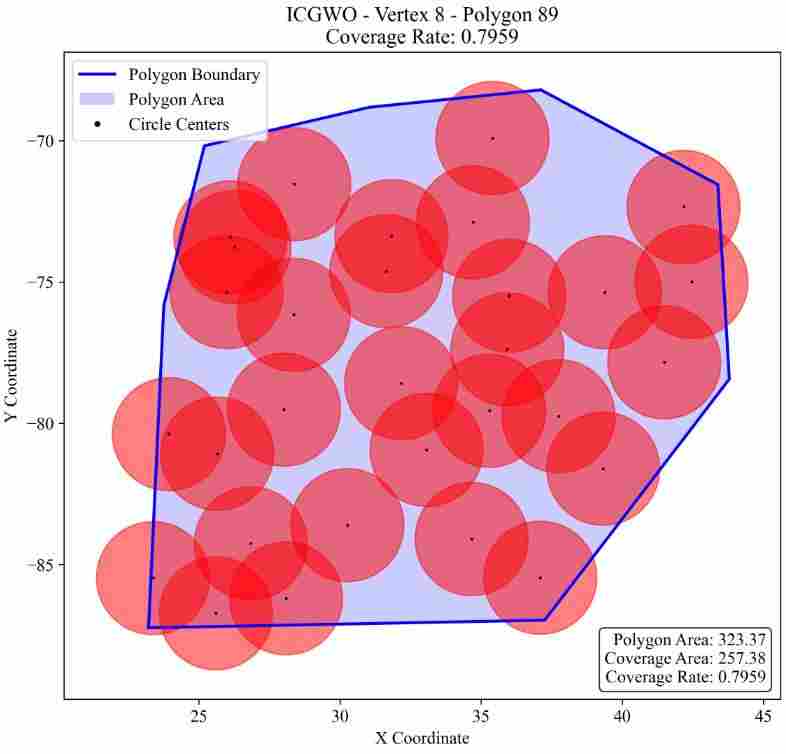}
	\end{subfigure}
	\begin{subfigure}[b]{0.15\textwidth}
		\centering
		\includegraphics[width=\textwidth]{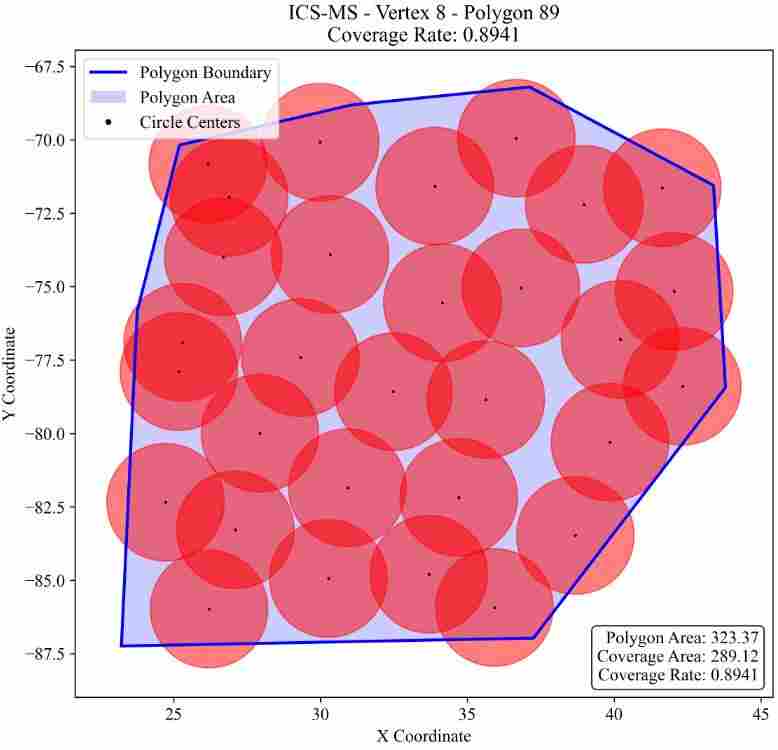}
	\end{subfigure}
	\begin{subfigure}[b]{0.15\textwidth}
		\centering
		\includegraphics[width=\textwidth]{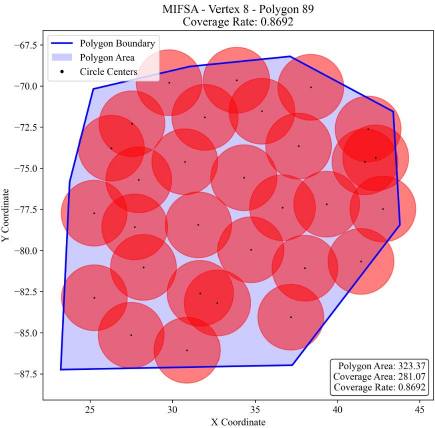}
	\end{subfigure}
	\begin{subfigure}[b]{0.15\textwidth}
		\centering
		\includegraphics[width=\textwidth]{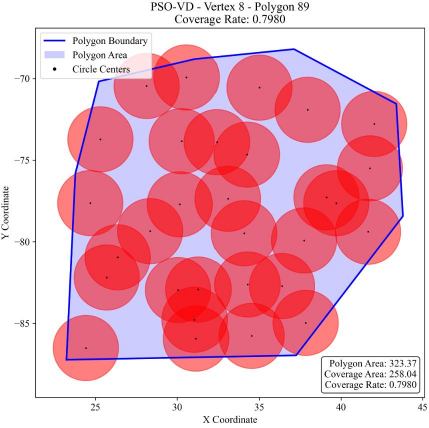}
	\end{subfigure}
	\begin{subfigure}[b]{0.15\textwidth}
		\centering
		\includegraphics[width=\textwidth]{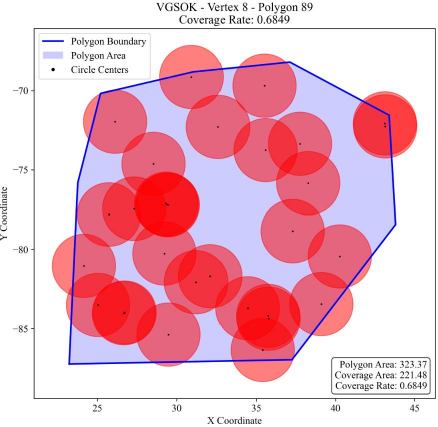}
	\end{subfigure}
	\begin{subfigure}[b]{0.15\textwidth}
		\centering
		\includegraphics[width=\textwidth]{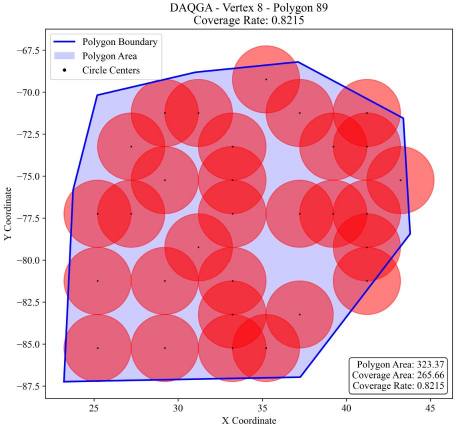}
	\end{subfigure}
	\begin{subfigure}[b]{0.15\textwidth}
		\centering
		\includegraphics[width=\textwidth]{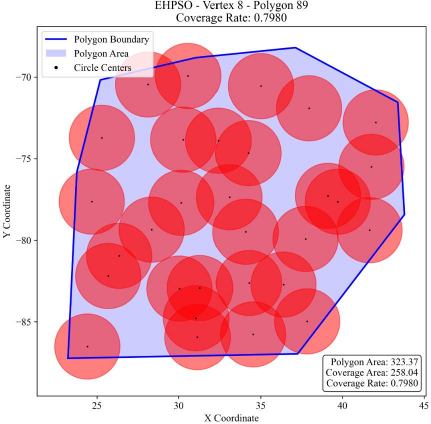}
	\end{subfigure}
	\begin{subfigure}[b]{0.15\textwidth}
		\centering
		\includegraphics[width=\textwidth]{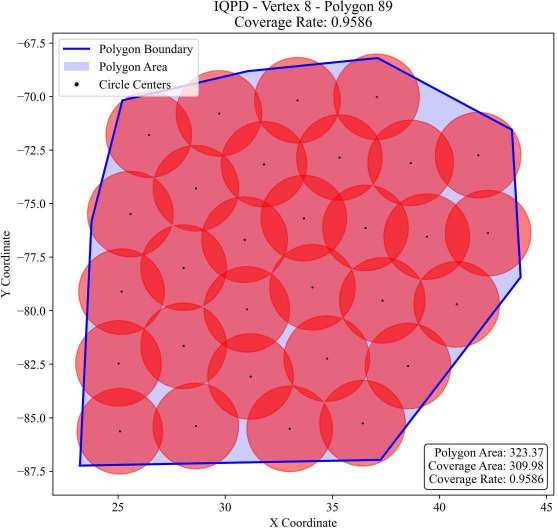}
	\end{subfigure}
	\begin{subfigure}[b]{0.2\textwidth}
		\centering
		\includegraphics[width=\textwidth]{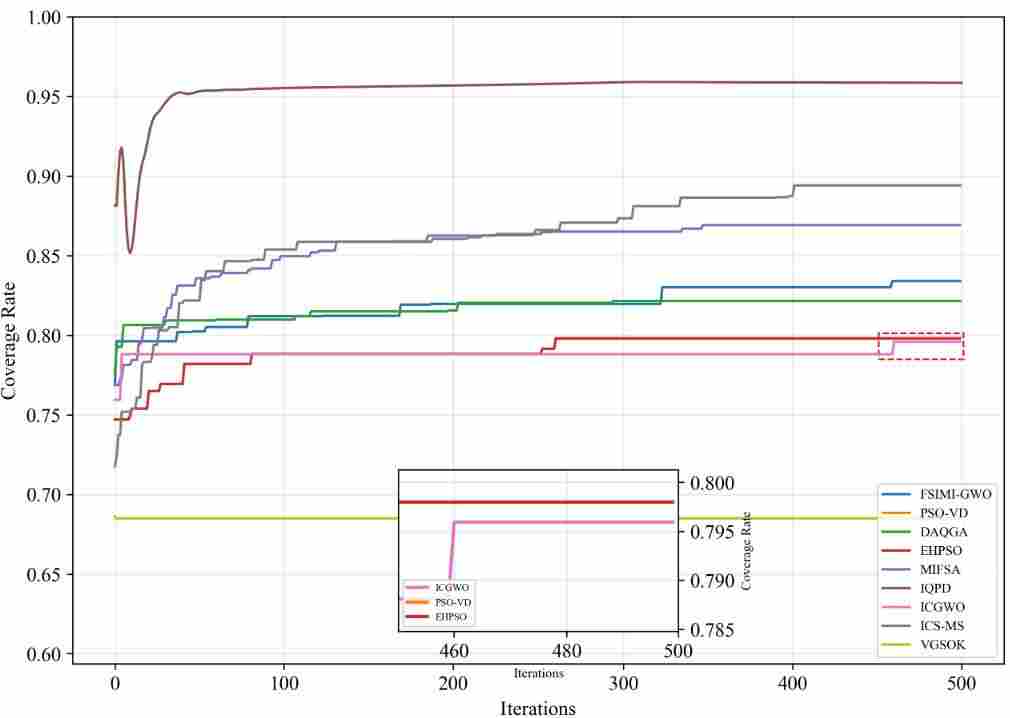}
	\end{subfigure}
	\caption{Optimized deployment layout and convergence curves for the case with 8 vertices and 30 circles.}
	\label{Fig:18}
\end{figure}

\subsection{Real-world Scenarios}
To further validate the efficiency of our proposed algorithm, we downloaded a real-world convex polygon dataset from  \href{https://zenodo.org/records/15746579}{\textit{EU Open Research Repository}}. This dataset, which consists of 100 convex polygons, is designed for testing real-world scenarios such as disaster response, UAS swarms, and wireless network coverage. We performed extensive experiments on this dataset with circle counts of {60, 80, 100} and a fixed circle radius of 222 m. 

The optimized coverage results for a part of countries and regions are presented in Figs.\ref{Fig:16}–-\ref{Fig:19} and Tables~\ref{tab:13}--\ref{tab:16}. Our algorithm clearly achieves the highest coverage rate among all metaheuristic algorithms, with runtime second only to VGSOK and far below the others. These findings confirm that our algorithm outperforms other methods in coverage efficiency on both synthetic and real-world data, offering an efficient solution to the circular coverage problem in convex polygonal domains.

\begin{figure}
	\centering
	\begin{subfigure}[b]{0.15\textwidth}
		\centering
		\includegraphics[width=\textwidth]{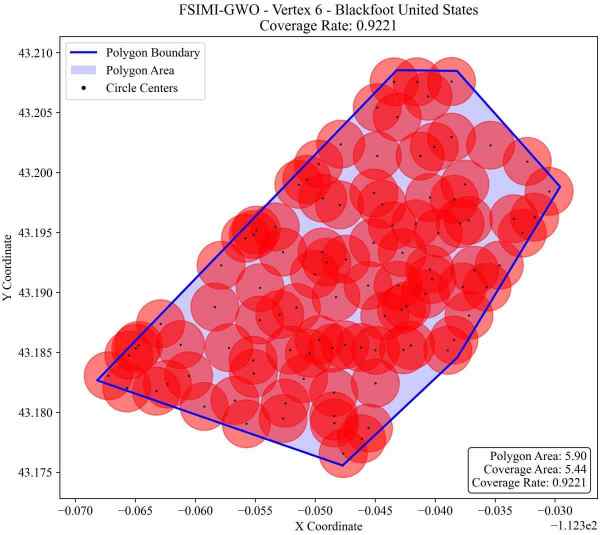}
	\end{subfigure}
	\begin{subfigure}[b]{0.15\textwidth}
		\centering
		\includegraphics[width=\textwidth]{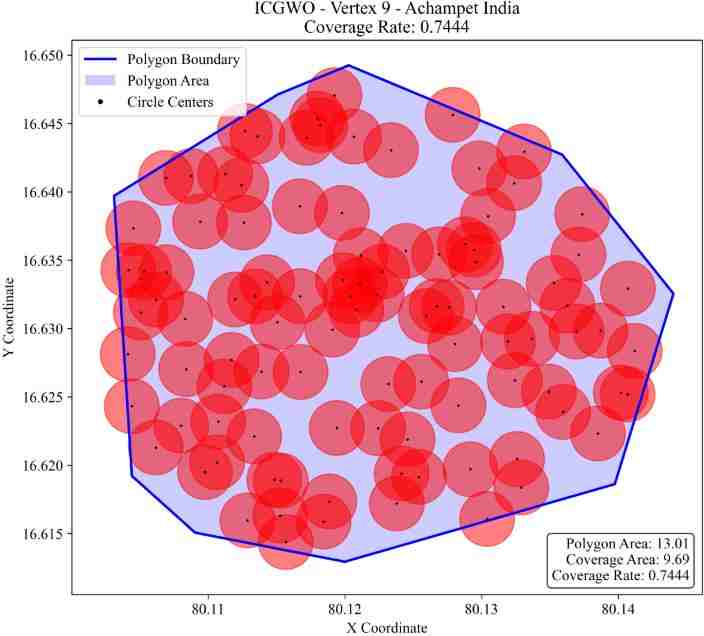}
	\end{subfigure}
	\begin{subfigure}[b]{0.15\textwidth}
		\centering
		\includegraphics[width=\textwidth]{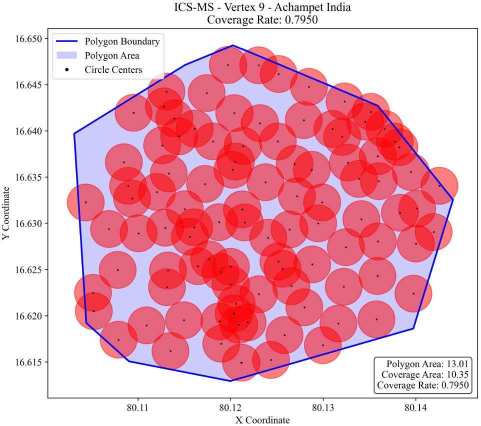}
	\end{subfigure}
	\begin{subfigure}[b]{0.15\textwidth}
		\centering
		\includegraphics[width=\textwidth]{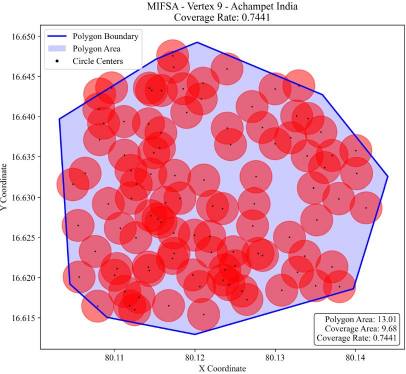}
	\end{subfigure}
	\begin{subfigure}[b]{0.15\textwidth}
		\centering
		\includegraphics[width=\textwidth]{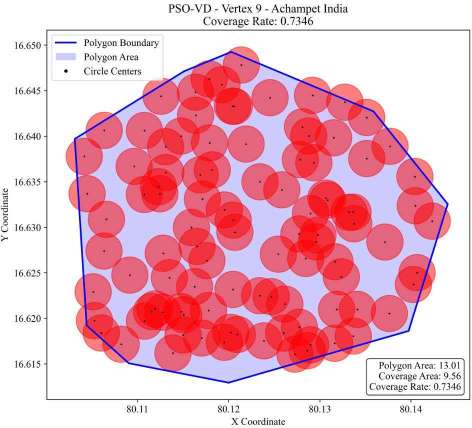}
	\end{subfigure}
	\begin{subfigure}[b]{0.15\textwidth}
		\centering
		\includegraphics[width=\textwidth]{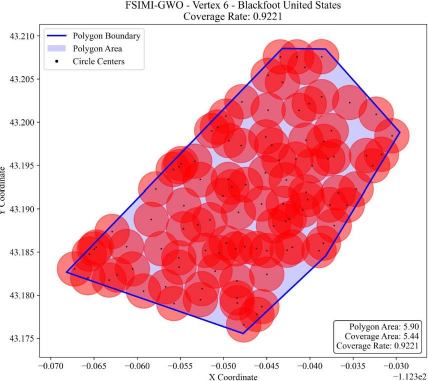}
	\end{subfigure}
	\begin{subfigure}[b]{0.15\textwidth}
		\centering
		\includegraphics[width=\textwidth]{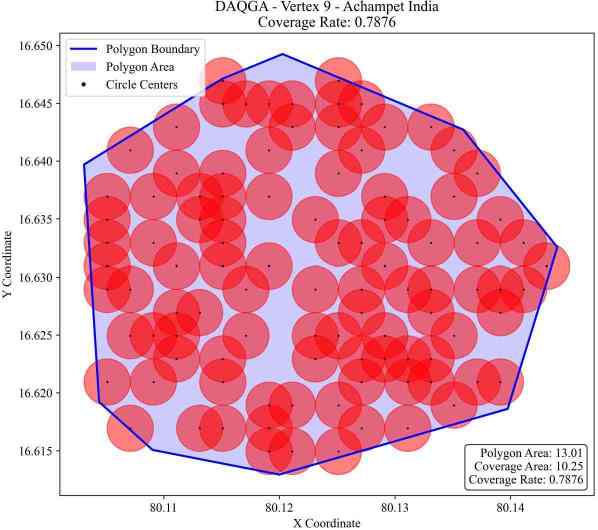}
	\end{subfigure}
	\begin{subfigure}[b]{0.15\textwidth}
		\centering
		\includegraphics[width=\textwidth]{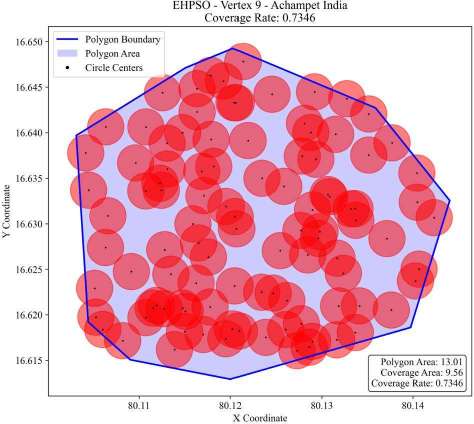}
	\end{subfigure}
	\begin{subfigure}[b]{0.15\textwidth}
		\centering
		\includegraphics[width=\textwidth]{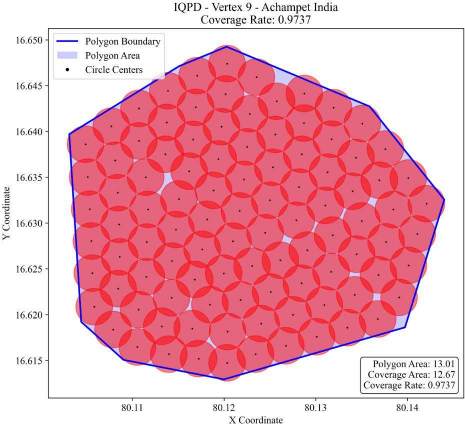}
	\end{subfigure}
	\caption{Optimized deployment for a 9-vertex polygon in Achampet, India, using different algorithms. The polygon area is 13.01 $km^2$. The IQPD algorithm covers 12.67 $km^2$, achieving the highest coverage rate at 0.9737.}
	\label{Fig:19}
\end{figure}

\begin{figure}
	\centering
	\begin{subfigure}[b]{0.15\textwidth}
		\centering
		\includegraphics[width=\textwidth]{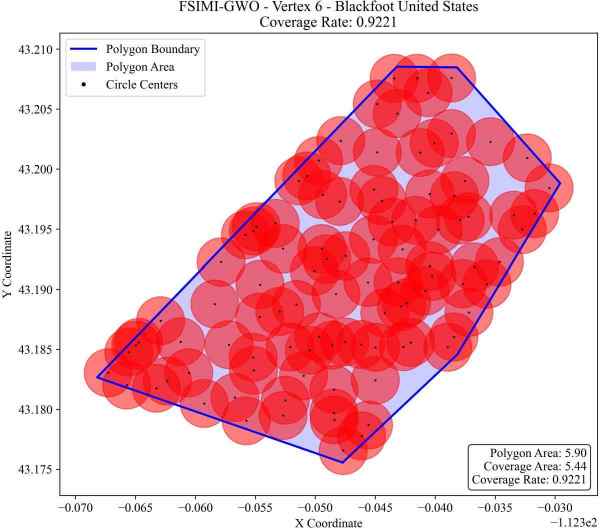}
	\end{subfigure}
	\begin{subfigure}[b]{0.15\textwidth}
		\centering
		\includegraphics[width=\textwidth]{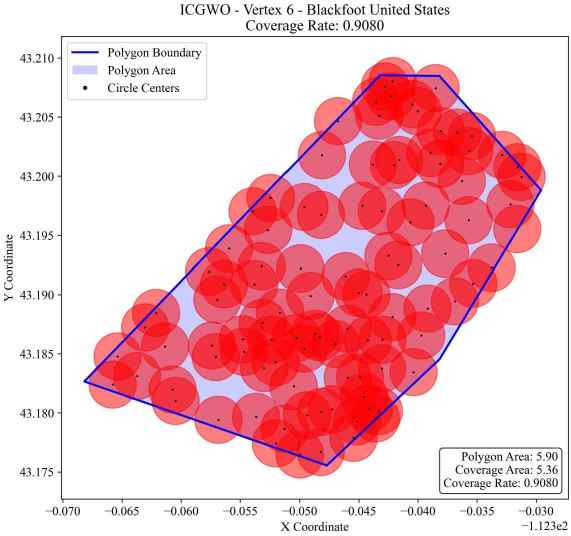}
	\end{subfigure}
	\begin{subfigure}[b]{0.15\textwidth}
		\centering
		\includegraphics[width=\textwidth]{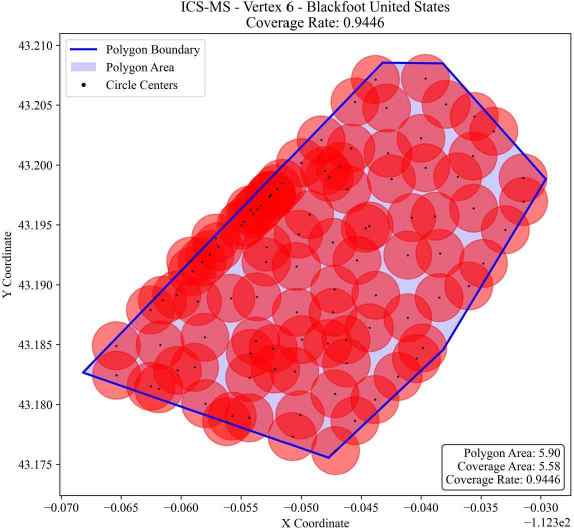}
	\end{subfigure}
	\begin{subfigure}[b]{0.15\textwidth}
		\centering
		\includegraphics[width=\textwidth]{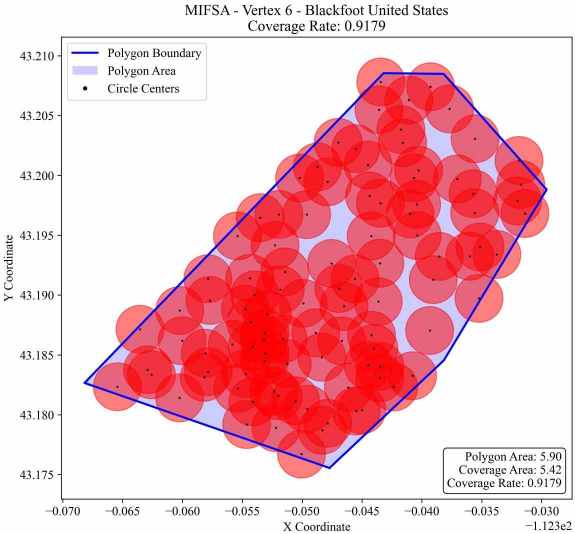}
	\end{subfigure}
	\begin{subfigure}[b]{0.15\textwidth}
		\centering
		\includegraphics[width=\textwidth]{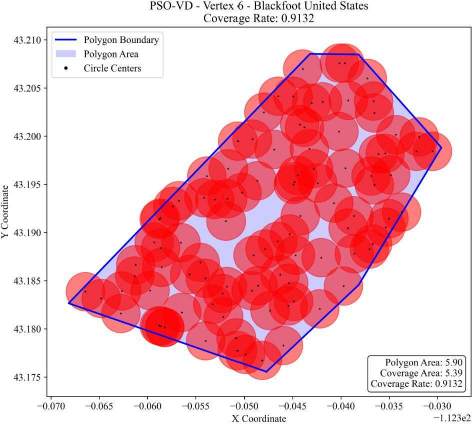}
	\end{subfigure}
	\begin{subfigure}[b]{0.15\textwidth}
		\centering
		\includegraphics[width=\textwidth]{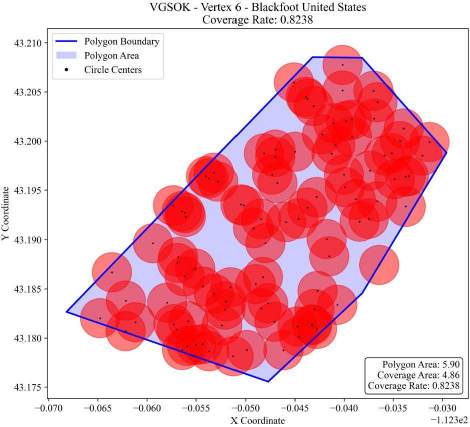}
	\end{subfigure}
	\begin{subfigure}[b]{0.15\textwidth}
		\centering
		\includegraphics[width=\textwidth]{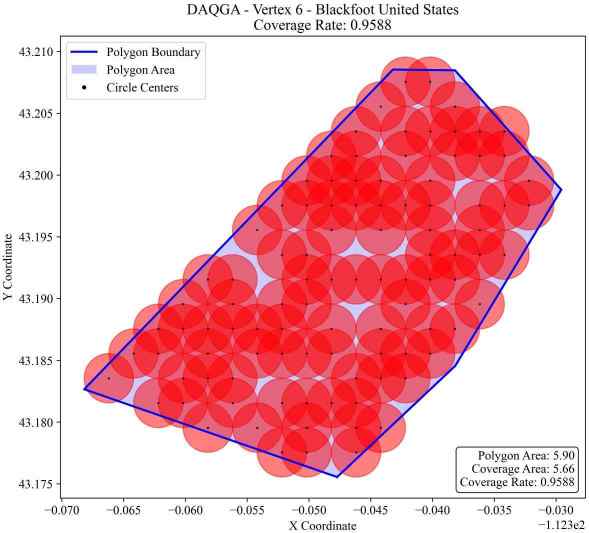}
	\end{subfigure}
	\begin{subfigure}[b]{0.15\textwidth}
		\centering
		\includegraphics[width=\textwidth]{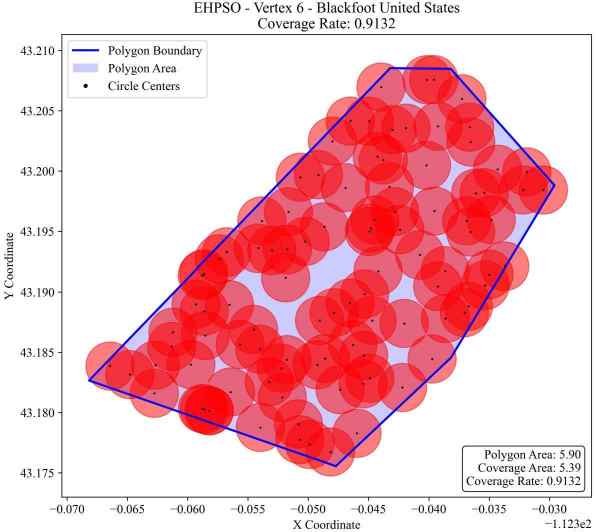}
	\end{subfigure}
	\begin{subfigure}[b]{0.15\textwidth}
		\centering
		\includegraphics[width=\textwidth]{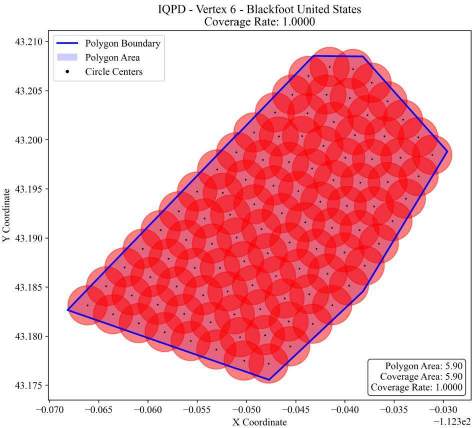}
	\end{subfigure}
	\caption{Optimized deployment for a 6-Vertex  polygon in Blackfoot, United States, using different algorithms. The polygon area is 5.90 $\mathrm{km}^2$. The IQPD algorithm covers 5.90 $\mathrm{km}^2$, achieving full coverage rate at 1.0000.}
	\label{Fig:20}
\end{figure}

\begin{figure}
	\centering
	\begin{subfigure}[b]{0.15\textwidth}
		\centering
		\includegraphics[width=\textwidth]{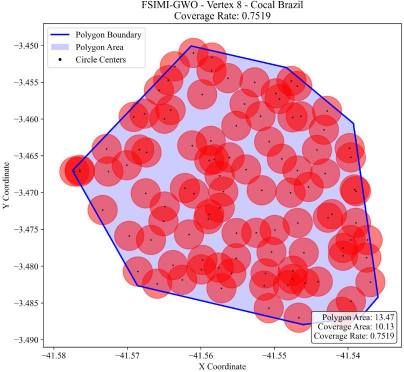}
	\end{subfigure}
	\begin{subfigure}[b]{0.15\textwidth}
		\centering
		\includegraphics[width=\textwidth]{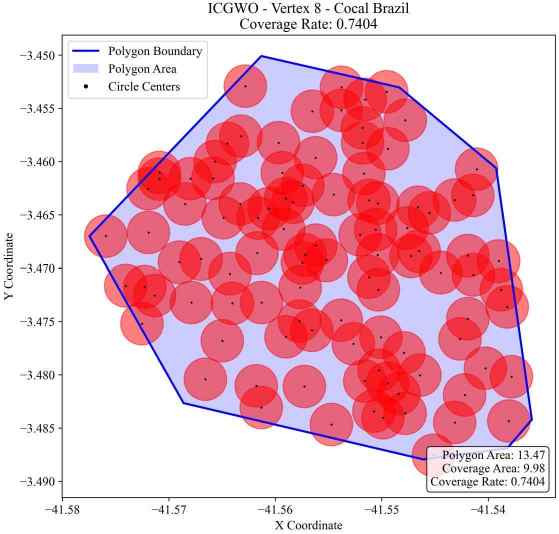}
	\end{subfigure}
	\begin{subfigure}[b]{0.15\textwidth}
		\centering
		\includegraphics[width=\textwidth]{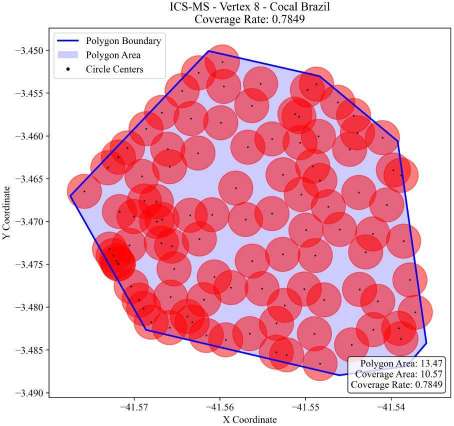}
	\end{subfigure}
	\begin{subfigure}[b]{0.15\textwidth}
		\centering
		\includegraphics[width=\textwidth]{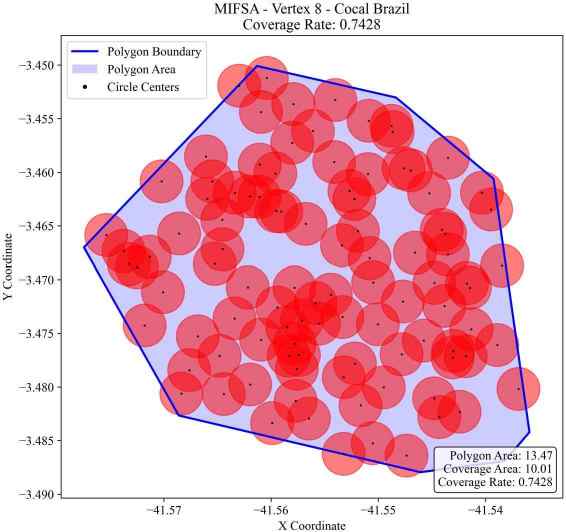}
	\end{subfigure}
	\begin{subfigure}[b]{0.15\textwidth}
		\centering
		\includegraphics[width=\textwidth]{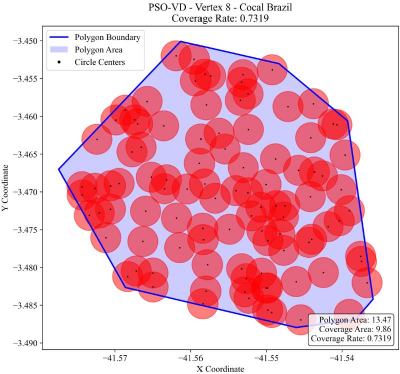}
	\end{subfigure}
	\begin{subfigure}[b]{0.15\textwidth}
		\centering
		\includegraphics[width=\textwidth]{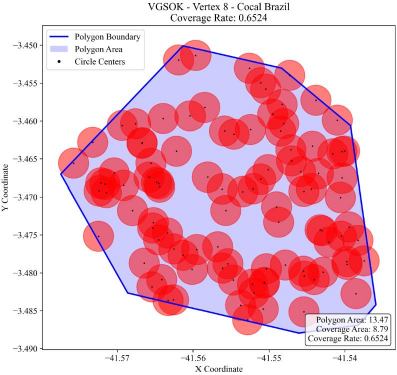}
	\end{subfigure}
	\begin{subfigure}[b]{0.15\textwidth}
		\centering
		\includegraphics[width=\textwidth]{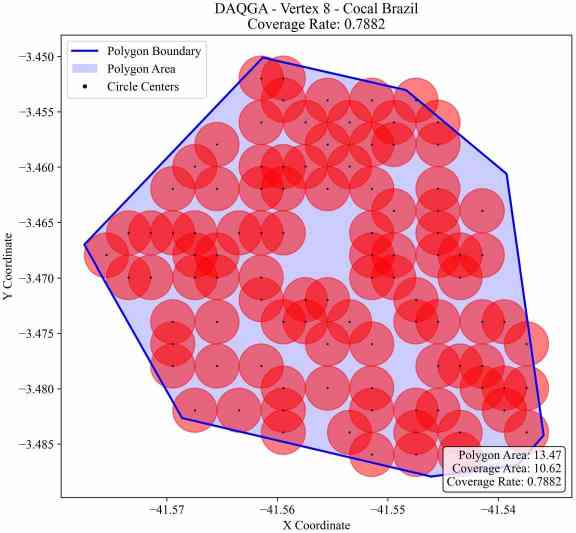}
	\end{subfigure}
	\begin{subfigure}[b]{0.15\textwidth}
		\centering
		\includegraphics[width=\textwidth]{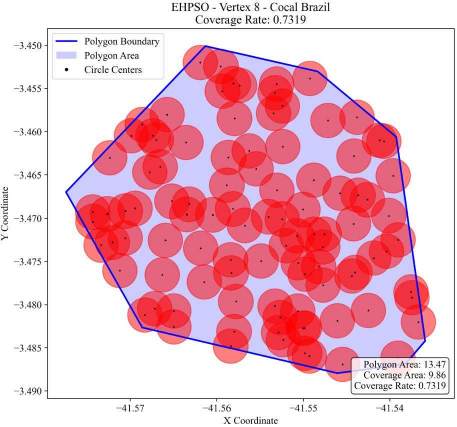}
	\end{subfigure}
	\begin{subfigure}[b]{0.15\textwidth}
		\centering
		\includegraphics[width=\textwidth]{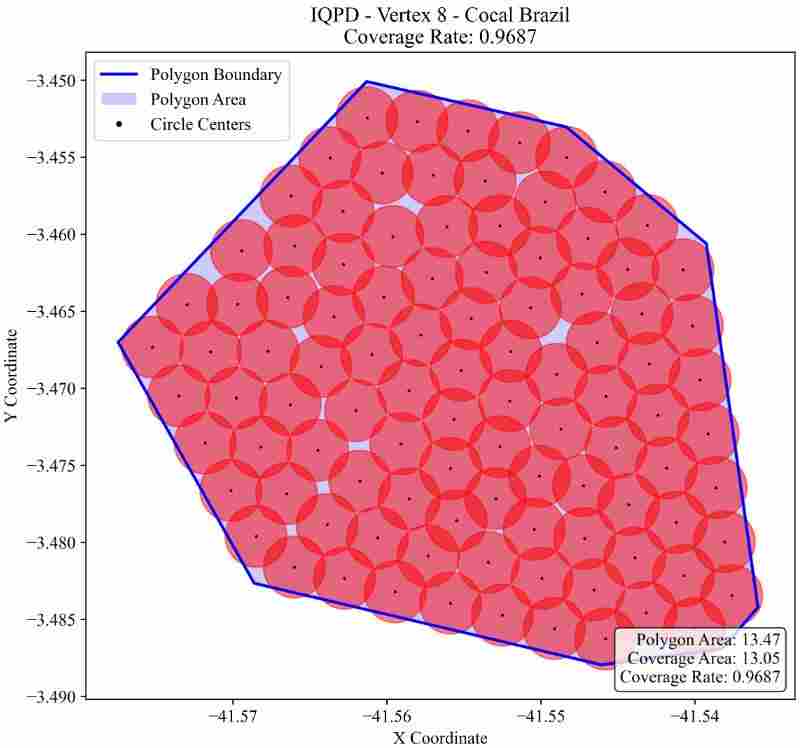}
	\end{subfigure}
	\caption{Optimized deployment for an 8-vertex polygon in Cocal, Brazil, using different algorithms. The polygon area is 13.47 $\mathrm{km}^2$. The IQPD algorithm covers 13.05 $\mathrm{km}^2$, achieving the highest coverage rate at 0.9687.}
	\label{Fig:21}
\end{figure}

\begin{figure}
	\centering
	\begin{subfigure}[b]{0.15\textwidth}
		\centering
		\includegraphics[width=\textwidth]{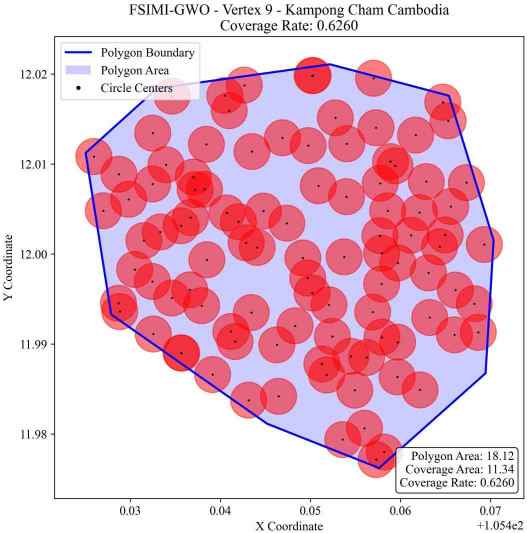}
	\end{subfigure}
	\begin{subfigure}[b]{0.15\textwidth}
		\centering
		\includegraphics[width=\textwidth]{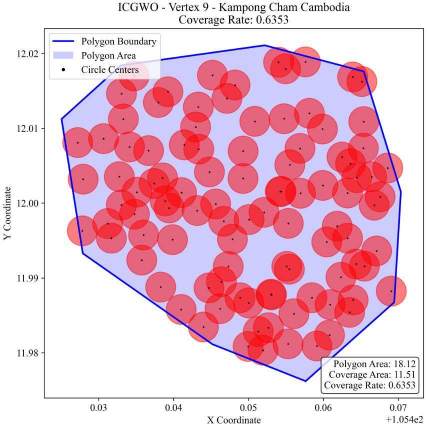}
	\end{subfigure}
	\begin{subfigure}[b]{0.15\textwidth}
		\centering
		\includegraphics[width=\textwidth]{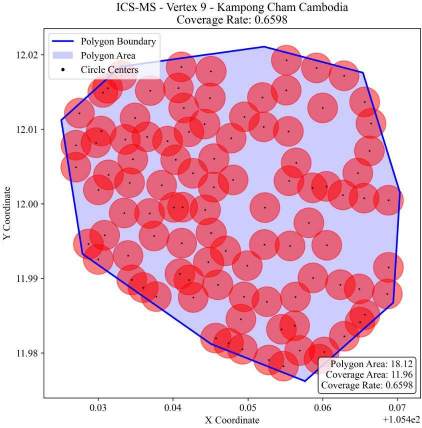}
	\end{subfigure}
	\begin{subfigure}[b]{0.15\textwidth}
		\centering
		\includegraphics[width=\textwidth]{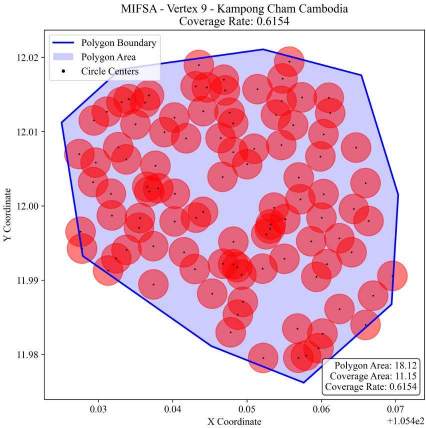}
	\end{subfigure}
	\begin{subfigure}[b]{0.15\textwidth}
		\centering
		\includegraphics[width=\textwidth]{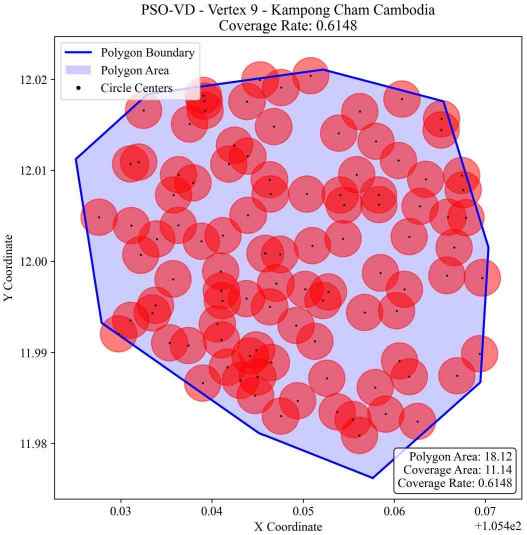}
	\end{subfigure}
	\begin{subfigure}[b]{0.15\textwidth}
		\centering
		\includegraphics[width=\textwidth]{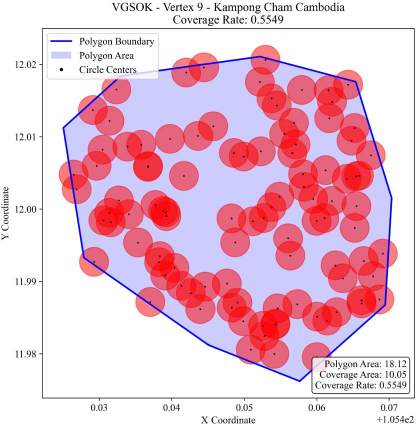}
	\end{subfigure}
	\begin{subfigure}[b]{0.15\textwidth}
		\centering
		\includegraphics[width=\textwidth]{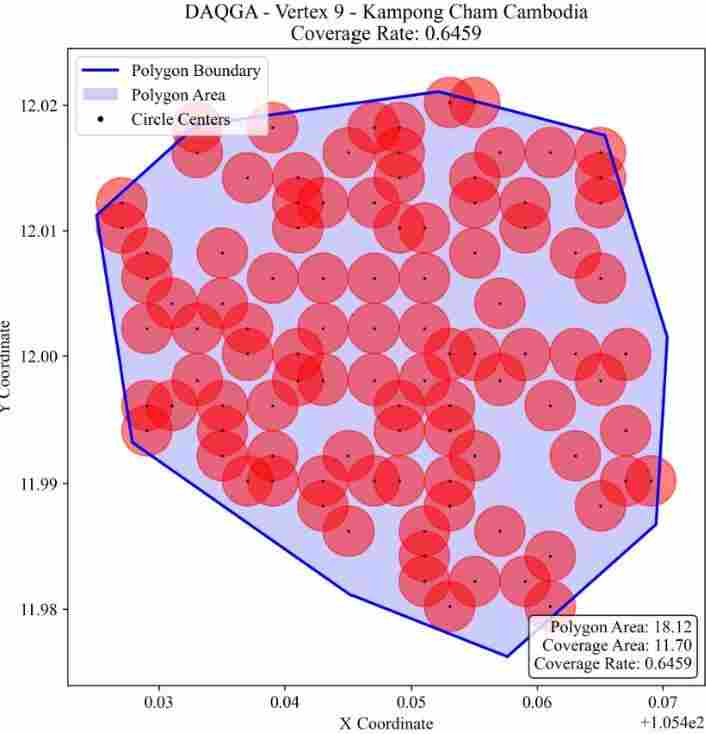}
	\end{subfigure}
	\begin{subfigure}[b]{0.15\textwidth}
		\centering
		\includegraphics[width=\textwidth]{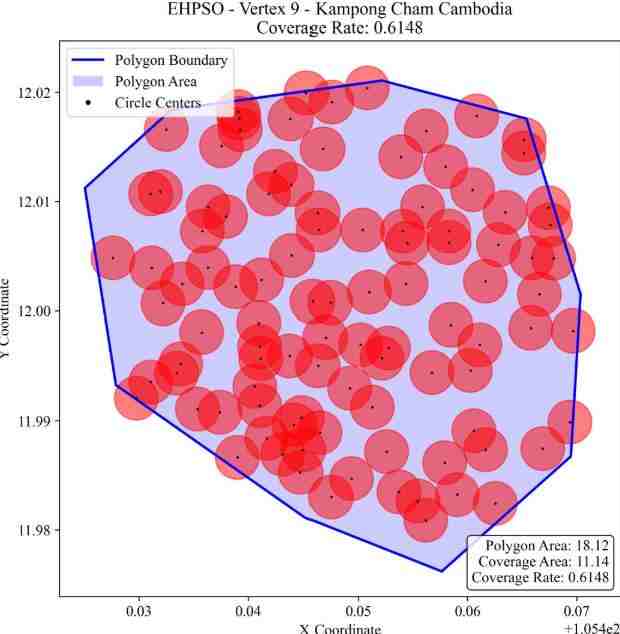}
	\end{subfigure}
	\begin{subfigure}[b]{0.15\textwidth}
		\centering
		\includegraphics[width=\textwidth]{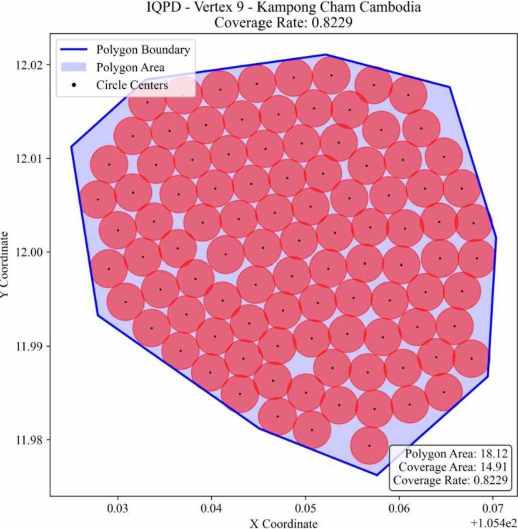}
	\end{subfigure}
	\caption{Optimized deployment for a 9-vertex polygon in Kampong Cham, Cambodia, using different algorithms. The polygon area is 18.12 $\mathrm{km}^2$. The IQPD algorithm covers 14.91 $\mathrm{km}^2$, achieving the highest coverage rate of 0.8229.}
	\label{Fig:22}
\end{figure}

\begin{table}[htbp]
	\centering
	\caption{Performance comparison of nine algorithms on the Achampet region of India.}
	\label{tab:13}
	\begin{tabular*}{\tblwidth}{@{\extracolsep{\fill}}lccccc@{}}
		\toprule
		\textbf{Algorithm} & \textbf{Coverage} & \textbf{Usage Rate} & \textbf{Uniformity Index} & \textbf{Min Gap} & \textbf{Time (s)} \\
		\midrule
		FSIMI-GWO & 0.760179 & 0.666769 & 0.655952 & -0.004000 & 15842.946 \\
		PSO-VD    & 0.734643 & 0.644371 & 0.600490 & -0.003803 & 3157.147 \\
		DAQGA     & 0.787589 & 0.690811 & 0.705230 & -0.002000 & 5754.953 \\
		EHPSO     & 0.734643 & 0.644371 & 0.600490 & -0.003803 & 11497.639 \\
		MIFSA     & 0.744149 & 0.652709 & 0.602012 & -0.003606 & 11233.510 \\
		ICGWO     & 0.744419 & 0.652354 & 0.577861 & -0.003525 & 2266.390 \\
		ICS-MS   & 0.795005 & 0.697238 & 0.598862 & -0.804282 & 3580.502 \\
		VGSOK     & 0.646428 & 0.008463 & 0.541228 & -0.032640 & 16.393 \\
		\textbf{IQPD} & \textbf{0.973662} & \textbf{0.853245} & \textbf{0.951155} & \textbf{-0.000882} & \textbf{693.121} \\
		\bottomrule
	\end{tabular*}
\end{table}

\begin{table}[htbp]
	\centering
	\caption{Performance comparison of nine algorithms on the Blackfoot region of United States.}
	\label{tab:14}
	\begin{tabular*}{\tblwidth}{@{\extracolsep{\fill}}lccccc@{}}
		\toprule
		\textbf{Algorithm} & \textbf{Coverage} & \textbf{Usage Rate} & \textbf{Uniformity Index} & \textbf{Min Gap} & \textbf{Time (s)} \\
		\midrule
		FSIMI-GWO & 0.922055 & 0.480096 & 0.652438 & -0.003672 & 8839.417 \\
		PSO-VD    & 0.913235 & 0.475504 & 0.662271 & -0.003980 & 12280.087 \\
		DAQGA     & 0.958791 & 0.499224 & 0.801136 & -0.002000 & 25096.199 \\
		EHPSO     & 0.913235 & 0.475504 & 0.662271 & -0.003980 & 4899.818 \\
		MIFSA     & 0.917863 & 0.477914 & 0.627287 & -0.003523 & 2188.611 \\
		ICGWO     & 0.908023 & 0.472068 & 0.574016 & -0.003519 & 231.851 \\
		ICS-MS   & 0.944565 & 0.491457 & 0.561035 & -0.976568 & 872.574 \\
		VGSOK     & 0.823838 & 0.007189 & 0.578550 & -0.030794 & 13.212 \\
		\textbf{IQPD} & \textbf{0.999989} & \textbf{0.519880} & \textbf{0.964816} & \textbf{-0.001704} & \textbf{140.975} \\
		\bottomrule
	\end{tabular*}
\end{table}

\begin{table}[htbp]
	\centering
	\caption{Performance comparison of nine algorithms on the Cocal region of Brazil.}
	\label{tab:15}
	\begin{tabular*}{\tblwidth}{@{\extracolsep{\fill}}lccccc@{}}
		\toprule
		\textbf{Algorithm} & \textbf{Coverage} & \textbf{Usage Rate} & \textbf{Uniformity Index} & \textbf{Min Gap} & \textbf{Time (s)} \\
		\midrule
		FSIMI\_GWO & 0.751861 & 0.656161 & 0.664567 & -0.003790 & 18367.040 \\
		PSO\_VD    & 0.731855 & 0.638701 & 0.598935 & -0.003821 & 15625.682 \\
		DAQGA     & 0.788192 & 0.687867 & 0.727899 & -0.002000 & 6552.750 \\
		EHPSO     & 0.731855 & 0.638701 & 0.598935 & -0.003821 & 3515.128 \\
		MIFSA     & 0.742780 & 0.648235 & 0.633732 & -0.003544 & 7084.266 \\
		ICGWO     & 0.740423 & 0.645590 & 0.603511 & -0.003379 & 293.290 \\
		ICS\_MS   & 0.784867 & 0.684886 & 0.604461 & -0.983687 & 856.008 \\
		VGSOK     & 0.652402 & 0.008241 & 0.537022 & -0.033144 & 16.106 \\
		\textbf{IQPD} & \textbf{0.968720} & \textbf{0.844646} & \textbf{0.935395} & \textbf{-0.001078} & \textbf{173.099} \\
		\bottomrule
	\end{tabular*}
\end{table}

\begin{table}[htbp]
	\centering
	\caption{Performance comparison of nine algorithms on the Kampong Cham region of Cambodia.}
	\label{tab:16}
	\begin{tabular*}{\tblwidth}{@{\extracolsep{\fill}}lccccc@{}}
		\toprule
		\textbf{Algorithm} & \textbf{Coverage} & \textbf{Usage Rate} & \textbf{Uniformity Index} & \textbf{Min Gap} & \textbf{Time (s)} \\
		\midrule
		FSIMI\_GWO & 0.625951 & 0.749097 & 0.643819 & -0.003912 & 9503.021 \\
		PSO\_VD    & 0.614798 & 0.735750 & 0.638610 & -0.003378 & 8459.312 \\
		DAQGA     & 0.645857 & 0.772920 & 0.666162 & -0.002000 & 6085.108 \\
		EHPSO     & 0.614798 & 0.735750 & 0.638610 & -0.003378 & 3896.826 \\
		MIFSA     & 0.615421 & 0.736495 & 0.596408 & -0.003430 & 2376.783 \\
		ICGWO     & 0.635290 & 0.759768 & 0.647040 & -0.003894 & 157.500 \\
		ICS\_MS   & 0.659809 & 0.789720 & 0.642867 & -0.784824 & 722.717 \\
		VGSOK     & 0.554859 & 0.008087 & 0.534962 & -0.036129 & 17.067 \\
		\textbf{IQPD} & \textbf{0.822864} & \textbf{0.984096} & \textbf{0.886717} & \textbf{-0.000328} & \textbf{569.155} \\
		\bottomrule
	\end{tabular*}
\end{table}

\subsection{Ablation Study}
To evaluate the contribution of each component, we conducted an ablation study on the proposed IQPD algorithm. The test scenario is still chosen from a real-world polygonal region in the \href{https://zenodo.org/records/15746579}{\textit{EU Open Research Repository}}, which corresponds to a lake shore area in the United States. This region is a convex polygon with 10 vertices. We set the circle radius to 200 $m$, which is determined based on the scale of the real-world polygon.

 Experimental results are shown in Table~\ref{tab:9} and Fig.~\ref{Fig:13}. The figure shows that removing the structure-preserving initialization strategy (i.e., adopting random initialization) increases the convergence time and slightly reduces the coverage rate. When the friction component is removed, the dynamic circles move outward indefinitely, placing multiple optimized nodes outside the target region. This significantly reduces the coverage and utilization rates and causes the algorithm to fail to converge. Removing the boundary encircling strategy prevents nodes that drift outside from being recovered into the target region, which also leads to a decrease in coverage and utilization rates. According to the ablation study, the friction component has the greatest impact on the overall coverage rate. Meanwhile, adopting random initialization increases the optimization time and reduces the algorithm's adaptability to real-world scenarios.

\begin{figure}
	\centering
	\begin{subfigure}[b]{0.15\textwidth}
		\centering
		\includegraphics[width=\textwidth]{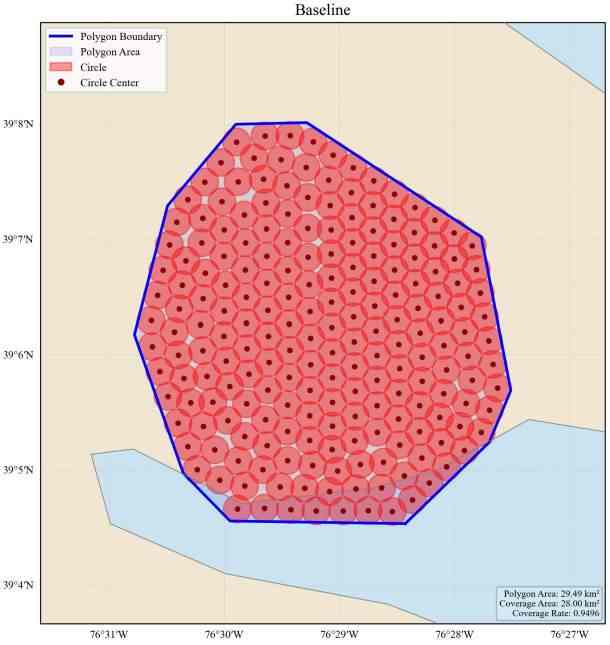}
		\caption{}
	\end{subfigure}
	\begin{subfigure}[b]{0.15\textwidth}
		\centering
		\includegraphics[width=\textwidth]{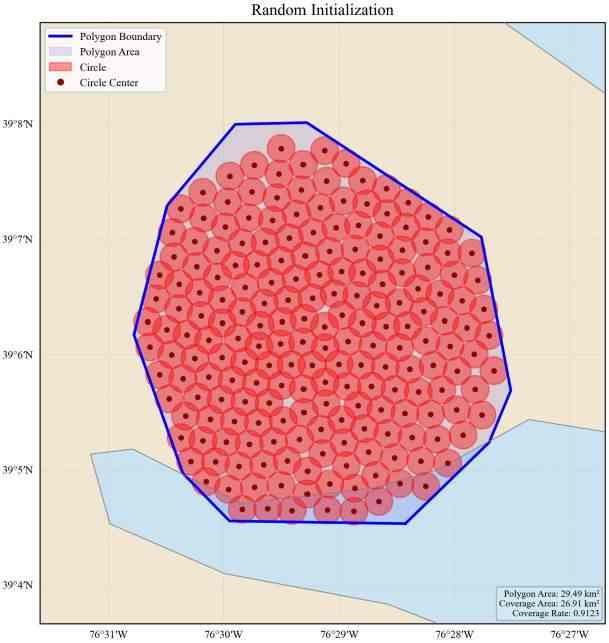}
		\caption{}
	\end{subfigure}
	\begin{subfigure}[b]{0.15\textwidth}
		\centering
		\includegraphics[width=\textwidth]{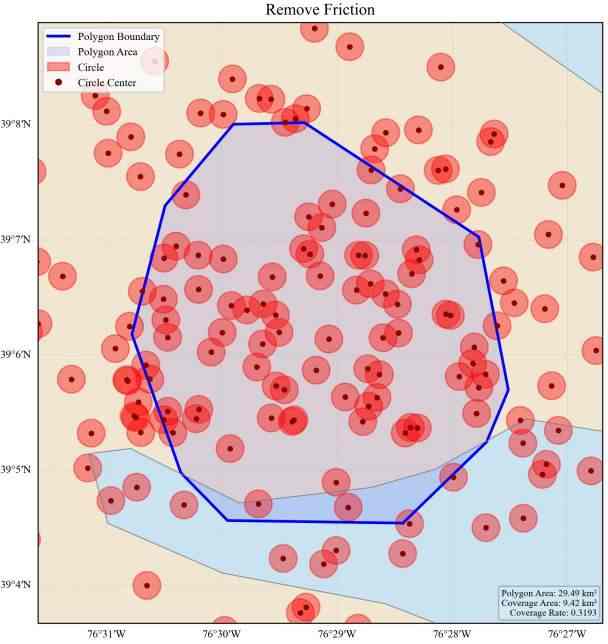}
		\caption{}
	\end{subfigure}
	\begin{subfigure}[b]{0.15\textwidth}
		\centering
		\includegraphics[width=\textwidth]{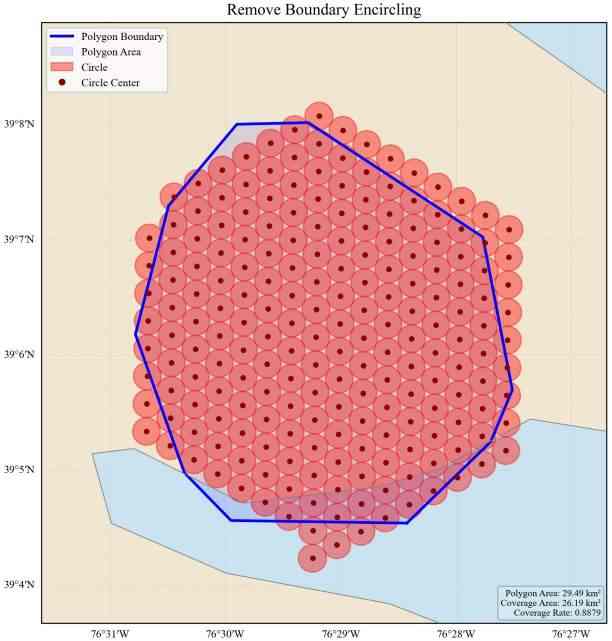}
		\caption{}
	\end{subfigure}
		\begin{subfigure}[b]{0.26\textwidth}
		\centering
		\includegraphics[width=\textwidth]{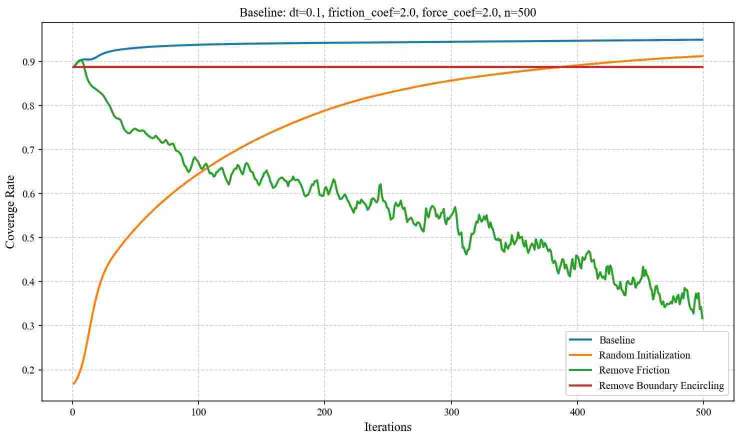}
		\caption{}
	\end{subfigure}
	\caption{Configuration diagram and convergence curves of the ablation study.}
	\label{Fig:13}
\end{figure}

\begin{table}
	\caption{Ablation study results of the proposed IQPD algorithm}
	\label{tab:9}
	\begin{tabular*}{\tblwidth}{@{\extracolsep{\fill}}lccccccc@{}}
		\toprule
		\textbf{Experiment} & \textbf{Polygon Area} & \textbf{Covered Area} & \textbf{CR} & \textbf{UR} & \textbf{UI} & \textbf{MG} & \textbf{Time (s)} \\
		\midrule
		Random Initialization & 29.49 & 26.91 & 0.9123 & 0.8686 & 0.8100 & -0.0010 & 142.74 \\
		Remove Friction & 29.49 & 9.42  & 0.3193 & 0.3039 & 0.4096 & -0.0038 & 140.71 \\
		Remove Boundary Encircling & 29.49 & 26.19 & 0.8879 & 0.8453 & 0.8157 & 0.1390 & 139.93 \\
		\textbf{Baseline} & 29.49 & 28.00 & \textbf{0.9496} & \textbf{0.9040} & \textbf{0.9268} & \textbf{-0.0008} & \textbf{163.52} \\
		\bottomrule
	\end{tabular*}
\end{table}

\subsection{Parameter analysis}

To further investigate the influence of the algorithm parameters, we conducted a sensitivity analysis. The proposed IQPD algorithm involves three key parameters: the friction coefficient \(\mu\), the attractive force coefficient \(\gamma\), and the time step \(\Delta t\). We tested \(\mu \in \{0.5, 1.0, 2.0, 5.0, 10.0\}\), \(\gamma \in \{0.01, 0.03, 0.06, 0.1, 0.2\}\), and \(\Delta t \in \{0.01, 0.03, 0.06, 0.1, 0.2\}\). Experiments were performed with three different numbers of circles \(n = 160, 180, 200\). The results are presented in Tables~\ref{tab:10}--\ref{tab:12} and Figs.~\ref{Fig:11}--\ref{Fig:12}.

The sensitivity analysis reveals that the proposed IQPD algorithm exhibits low sensitivity to the key parameters $\mu$, $\gamma$, $\Delta t$ with respect to coverage rate, utilization rate, and minimum gap, which demonstrates strong robustness and avoids complex parameter tuning. Moreover, the minimum gap stays near -0.002, where a small negative value implies minimal overlap. This indicates that the algorithm  prevents excessive coverage overlap successfully across all parameter configurations. The computational time and uniformity index are moderately affected. A larger $\gamma$ improves distribution uniformity, while a smaller $\mu$ and a moderate $\gamma$ help reduce computational time. The default parameter combination ($\mu=2.0$, $\gamma=0.03$, $\Delta t=0.1$) achieves a favorable balance between coverage effectiveness and computational efficiency across different node scales (160, 180, and 200).

\begin{table}
	\caption{Sensitivity analysis of key parameters (Friction coefficient, Attractive force coefficient, and Time step)  with 160 nodes.}
	\label{tab:10}
	\begin{tabular*}{\tblwidth}{@{\extracolsep{\fill}}lccccccccc@{}}
		\toprule
		\textbf{Parameter} & \textbf{Experiment} & \(\mu\) & \(\gamma \) & \(\Delta t\) & \textbf{CR} & \textbf{UR} & \textbf{UI} & \textbf{MG} & \textbf{Time (s)} \\
		\midrule
		\multirow{5}{*}{Friction} & friction\_coef=0.5 & 0.5 & 0.03 & 0.1 & 0.8273 & 0.9844 & 0.7826 & -0.0012 & 117.98 \\
		& friction\_coef=1.0 & 1.0 & 0.03 & 0.1 & 0.8263 & 0.9833 & 0.8005 & -0.0015 & 116.06 \\
		& baseline & 2.0 & 0.03 & 0.1 & 0.8254 & 0.9821 & 0.8026 & -0.0016 & \textbf{115.22} \\
		& friction\_coef=5.0 & 5.0 & 0.03 & 0.1 & 0.8240 & 0.9805 & 0.8059 & -0.0011 & 116.48 \\
		& friction\_coef=10.0 & 10.0 & 0.03 & 0.1 & 0.8230 & 0.9793 & 0.8143 & -0.0007 & 121.06 \\
		\midrule
		\multirow{5}{*}{Force} & force\_coef=0.01 & 2.0 & 0.01 & 0.1 & 0.8237 & 0.9802 & 0.8063 & -0.0013 & 123.95 \\
		& baseline & 2.0 & 0.03 & 0.1 & 0.8254 & 0.9821 & 0.8026 & -0.0016 & \textbf{115.22} \\
		& force\_coef=0.06 & 2.0 & 0.06 & 0.1 & 0.8240 & 0.9833 & 0.8005 & -0.0015 & 121.13 \\
		& force\_coef=0.1 & 2.0 & 0.10 & 0.1 & 0.8271 & 0.9842 & 0.7829 & -0.0013 & 124.60 \\
		& force\_coef=0.2 & 2.0 & 0.20 & 0.1 & 0.8281 & 0.9853 & 0.7830 & -0.0011 & 126.66 \\
		\midrule
		\multirow{4}{*}{Time step} & dt=0.05 & 2.0 & 0.03 & 0.05 & 0.8243 & 0.9809 & 0.8053 & -0.0013 & 115.34 \\
		& dt=0.2 & 2.0 & 0.03 & 0.20 & 0.8264 & 0.9833 & 0.8005 & -0.0015 & 115.38 \\
		& dt=0.5 & 2.0 & 0.03 & 0.50 & 0.8277 & 0.9849 & 0.7826 & -0.0011 & 118.95 \\
		& baseline & 2.0 & 0.03 & 0.10 & 0.8254 & 0.9821 & 0.8026 & -0.0016 & \textbf{115.22} \\
		\bottomrule
	\end{tabular*}
\end{table}

\begin{table}
	\caption{Sensitivity analysis of key parameters (Friction coefficient, Attractive force coefficient, and Time step) with 180 nodes.}
	\label{tab:11}
	\begin{tabular*}{\tblwidth}{@{\extracolsep{\fill}}lccccccccc@{}}
		\toprule
		\textbf{Parameter} & \textbf{Experiment} & \(\mu\) & \(\gamma\) & \(\Delta t\) & \textbf{CR} & \textbf{UR} & \textbf{UI} & \textbf{MG} & \textbf{Time (s)} \\
		\midrule
		\multirow{5}{*}{Friction} & friction\_coef=0.5 & 0.5 & 0.03 & 0.1 & 0.8907 & 0.9421 & 0.8390 & -0.0022 & 150.34 \\
		& friction\_coef=1.0 & 1.0 & 0.03 & 0.1 & 0.8863 & 0.9375 & 0.8352 & -0.0026 & 172.77 \\
		& baseline & 2.0 & 0.03 & 0.1 & 0.8822 & 0.9332 & 0.8245 & -0.0022 & 148.10 \\
		& friction\_coef=5.0 & 5.0 & 0.03 & 0.1 & 0.8778 & 0.9284 & 0.8220 & -0.0015 & 151.17 \\
		& friction\_coef=10.0 & 10.0 & 0.03 & 0.1 & 0.8749 & 0.9255 & 0.8101 & -0.0009 & 148.79 \\
		\midrule
		\multirow{5}{*}{Force} & force\_coef=0.01 & 2.0 & 0.01 & 0.1 & 0.8770 & 0.9276 & 0.8204 & -0.0013 & 146.25 \\
		& baseline & 2.0 & 0.03 & 0.1 & 0.8822 & 0.9332 & 0.8245 & -0.0022 & 148.10 \\
		& force\_coef=0.06 & 2.0 & 0.06 & 0.1 & 0.8864 & 0.9376 & 0.8355 & -0.0025 & 151.91 \\
		& force\_coef=0.1 & 2.0 & 0.10 & 0.1 & 0.8896 & 0.9410 & 0.8377 & -0.0023 & 148.51 \\
		& force\_coef=0.2 & 2.0 & 0.20 & 0.1 & 0.8935 & 0.9451 & 0.8483 & -0.0015 & 156.35 \\
		\midrule
		\multirow{4}{*}{Time step} & dt=0.05 & 2.0 & 0.03 & 0.05 & 0.8787 & 0.9294 & 0.8275 & -0.0017 & 138.16 \\
		& baseline & 2.0 & 0.03 & 0.10 & 0.8822 & 0.9332 & 0.8245 & -0.0022 & 148.10 \\
		& dt=0.2 & 2.0 & 0.03 & 0.20 & 0.8864 & 0.9376 & 0.8357 & -0.0025 & 141.50 \\
		& dt=0.5 & 2.0 & 0.03 & 0.50 & 0.8922 & 0.9437 & 0.8435 & -0.0021 & 141.39 \\
		\bottomrule
	\end{tabular*}
\end{table}

\begin{table}
	\caption{Sensitivity analysis of key parameters (Friction coefficient, Attractive force coefficient, and Time step) with 200 nodes.}
	\label{tab:12}
	\begin{tabular*}{\tblwidth}{@{\extracolsep{\fill}}lccccccccc@{}}
		\toprule
		\textbf{Parameter} & \textbf{Experiment} & \(\mu\) & \(\alpha\) & \(\Delta t\) & \textbf{CR} & \textbf{UR} & \textbf{UI} & \textbf{MG} & \textbf{Time (s)} \\
		\midrule
		\multirow{5}{*}{Friction} & friction\_coef=0.5 & 0.5 & 0.03 & 0.1 & 0.9278 & 0.8833 & 0.8133 & -0.0022 & 161.61 \\
		& friction\_coef=1.0 & 1.0 & 0.03 & 0.1 & 0.9195 & 0.8754 & 0.7917 & -0.0026 & 206.80 \\
		& baseline & 2.0 & 0.03 & 0.1 & 0.9122 & 0.8685 & 0.7769 & -0.0026 & 173.47 \\
		& friction\_coef=5.0 & 5.0 & 0.03 & 0.1 & 0.9044 & 0.8610 & 0.7690 & -0.0022 & 241.68 \\
		& friction\_coef=10.0 & 10.0 & 0.03 & 0.1 & 0.8996 & 0.8565 & 0.8047 & -0.0015 & 207.14 \\
		\midrule
		\multirow{5}{*}{Force} & force\_coef=0.01 & 2.0 & 0.01 & 0.1 & 0.9030 & 0.8597 & 0.7655 & -0.0022 & 164.64 \\
		& baseline & 2.0 & 0.03 & 0.1 & 0.9122 & 0.8685 & 0.7769 & -0.0026 & 173.47 \\
		& force\_coef=0.06 & 2.0 & 0.06 & 0.1 & 0.9196 & 0.8755 & 0.7922 & -0.0025 & 163.76 \\
		& force\_coef=0.1 & 2.0 & 0.10 & 0.1 & 0.9258 & 0.8814 & 0.8071 & -0.0023 & 162.85 \\
		& force\_coef=0.2 & 2.0 & 0.20 & 0.1 & 0.9342 & 0.8894 & 0.8395 & -0.0015 & 162.89 \\
		\midrule
		\multirow{4}{*}{Time step} & dt=0.05 & 2.0 & 0.03 & 0.05 & 0.9048 & 0.8614 & 0.7701 & -0.0023 & 189.68 \\
		& baseline & 2.0 & 0.03 & 0.10 & 0.9122 & 0.8685 & 0.7769 & -0.0026 & 173.47 \\
		& dt=0.2 & 2.0 & 0.03 & 0.20 & 0.8952 & 0.8522 & 0.7890 & -0.0008 & 162.87 \\
		& dt=0.5 & 2.0 & 0.03 & 0.50 & 0.9311 & 0.8864 & 0.8247 & -0.0021 & 188.83 \\
		\bottomrule
	\end{tabular*}
\end{table}

\begin{figure}
	\centering
	\begin{subfigure}[b]{0.26\textwidth}
		\centering
		\includegraphics[width=\textwidth]{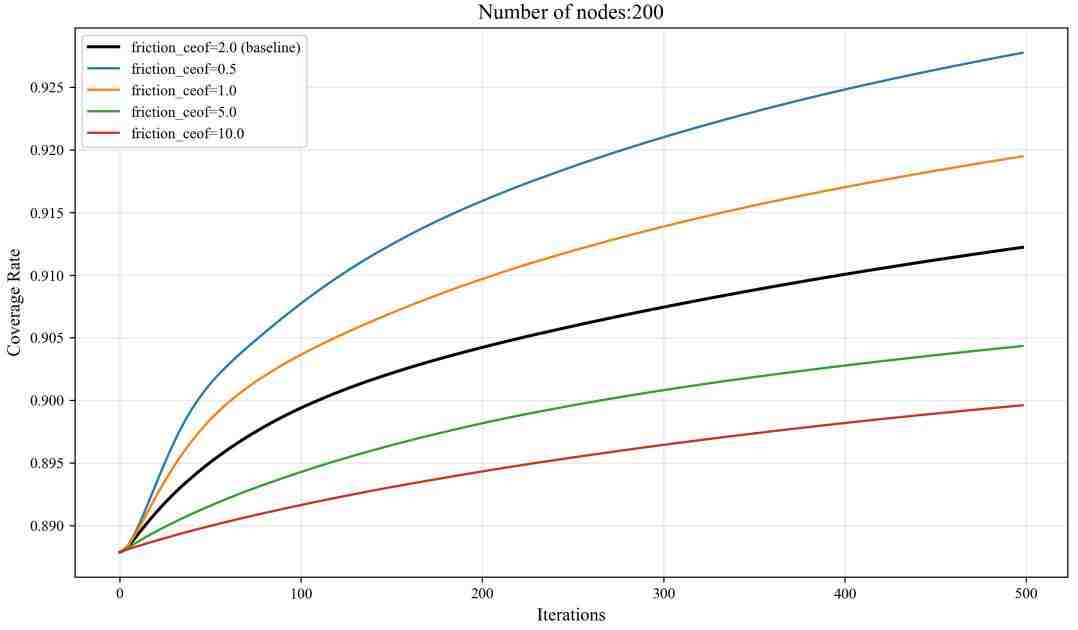}
		\caption{}
	\end{subfigure}
	\begin{subfigure}[b]{0.26\textwidth}
		\centering
		\includegraphics[width=\textwidth]{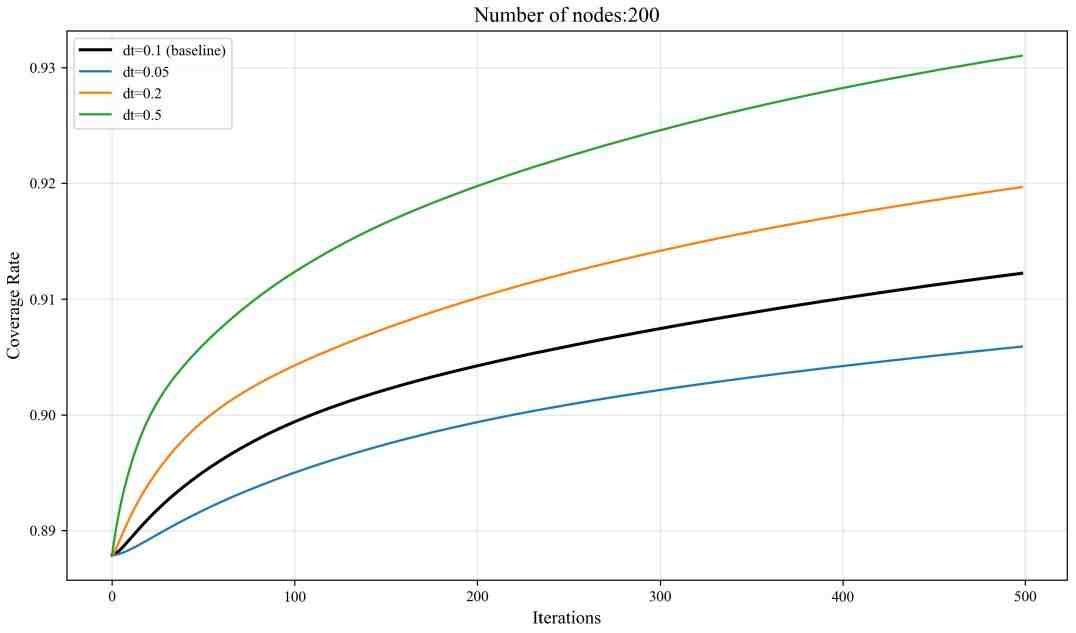}
		\caption{}
	\end{subfigure}
		\begin{subfigure}[b]{0.26\textwidth}
		\centering
		\includegraphics[width=\textwidth]{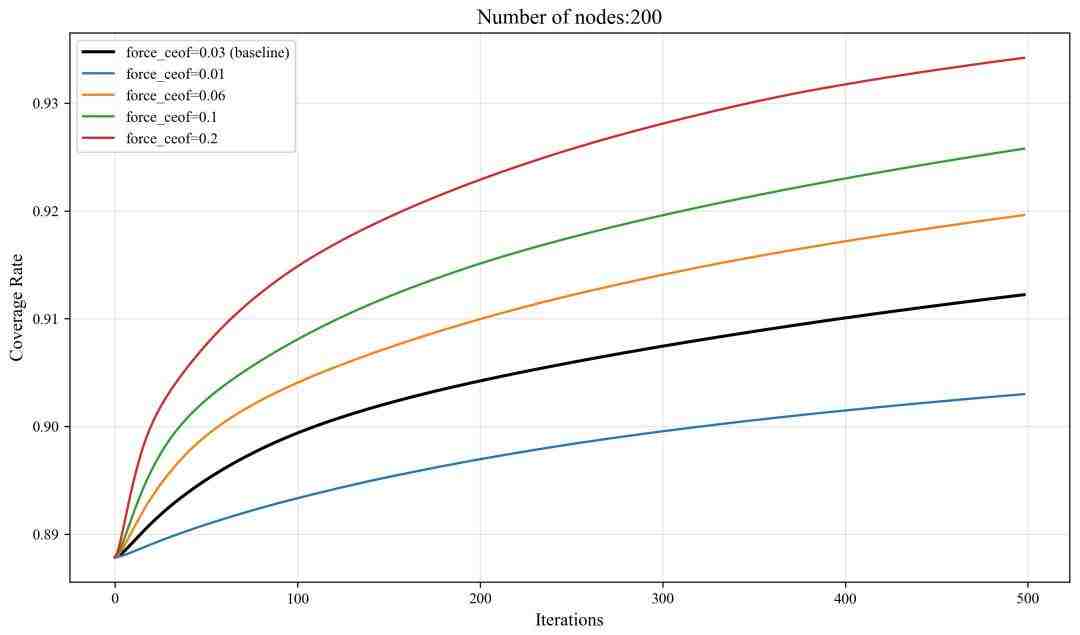}
		\caption{}
	\end{subfigure}
	\begin{subfigure}[b]{0.26\textwidth}
		\centering
		\includegraphics[width=\textwidth]{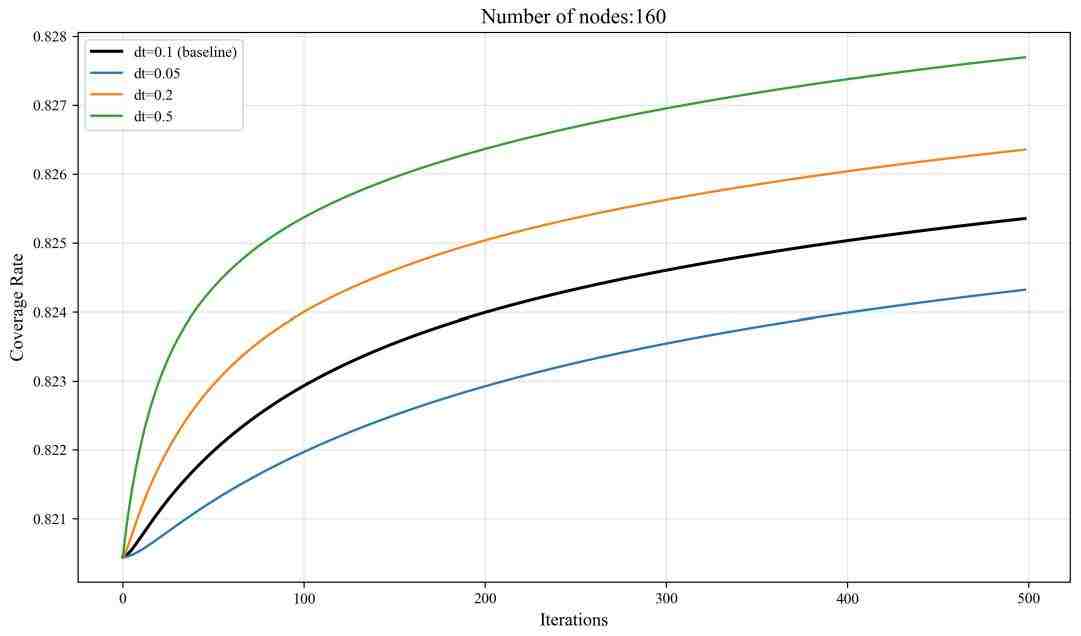}
		\caption{}
	\end{subfigure}
	\begin{subfigure}[b]{0.26\textwidth}
		\centering
		\includegraphics[width=\textwidth]{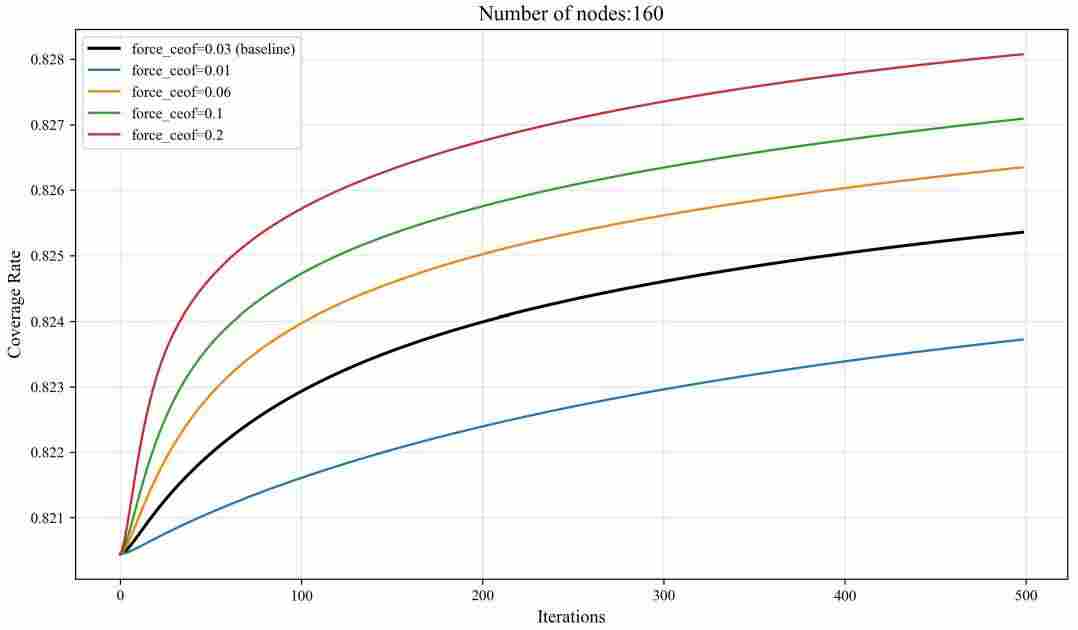}
		\caption{}
	\end{subfigure}
	\begin{subfigure}[b]{0.26\textwidth}
		\centering
		\includegraphics[width=\textwidth]{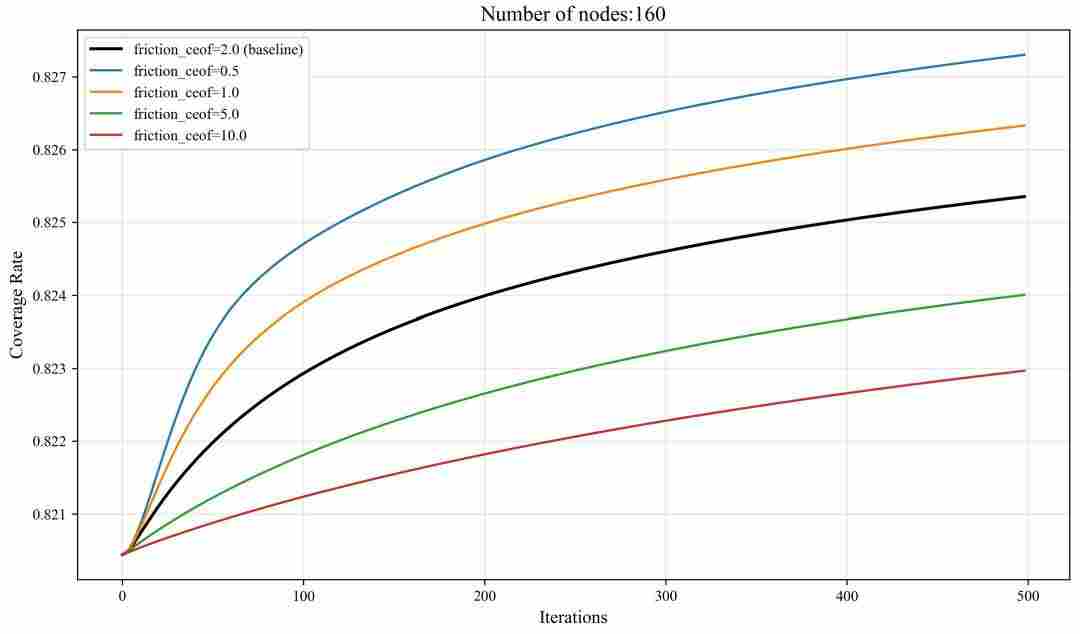}
		\caption{}
	\end{subfigure}
	\begin{subfigure}[b]{0.26\textwidth}
		\centering
		\includegraphics[width=\textwidth]{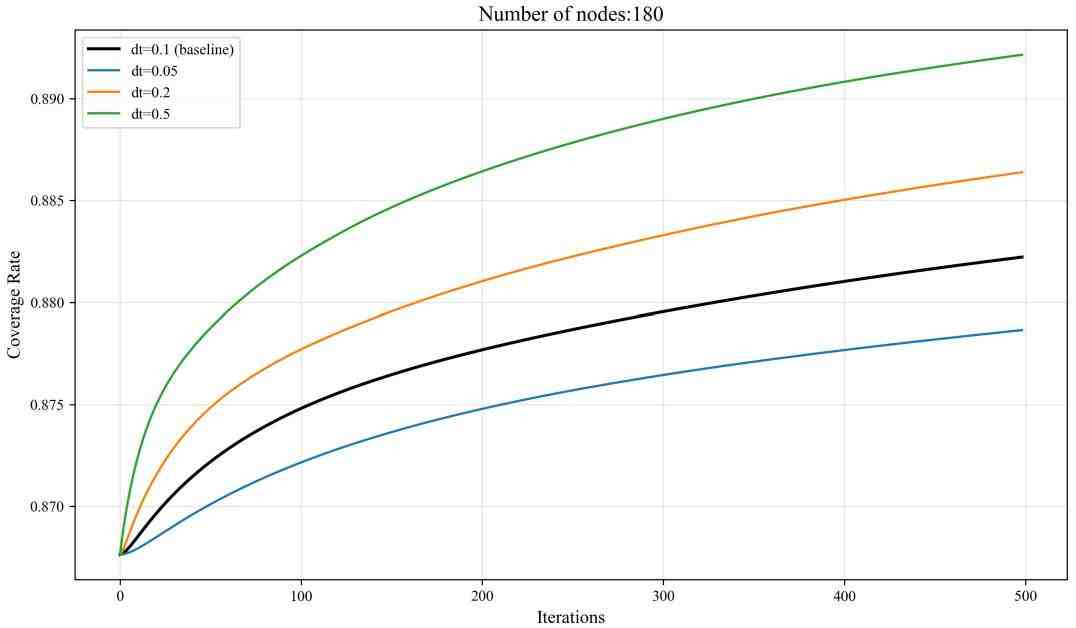}
		\caption{}
	\end{subfigure}
	\begin{subfigure}[b]{0.26\textwidth}
		\centering
		\includegraphics[width=\textwidth]{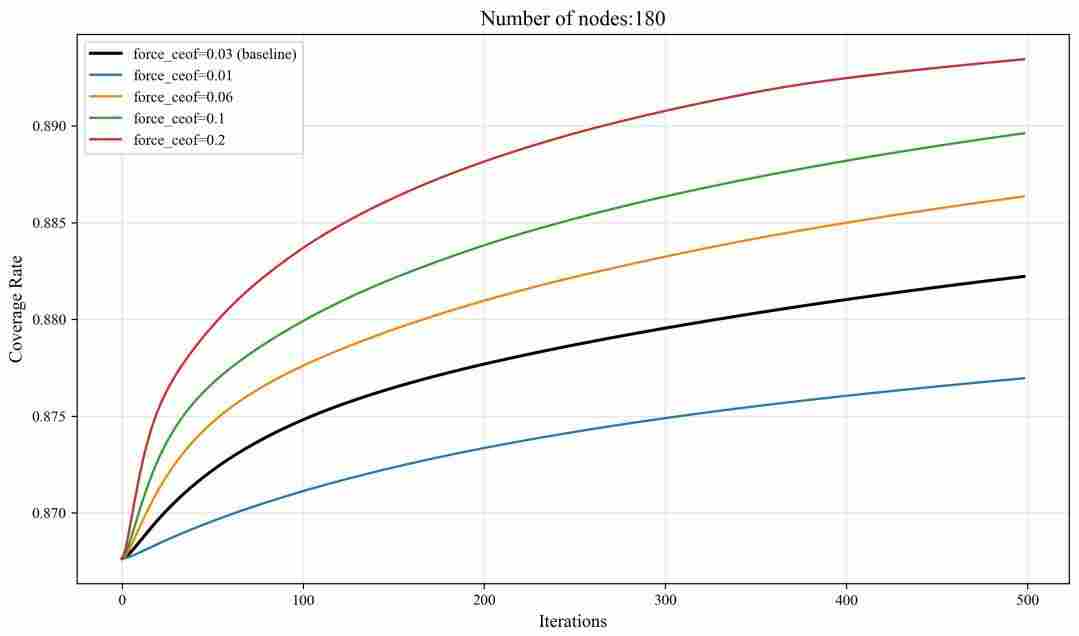}
		\caption{}
	\end{subfigure}
	\begin{subfigure}[b]{0.26\textwidth}
		\centering
		\includegraphics[width=\textwidth]{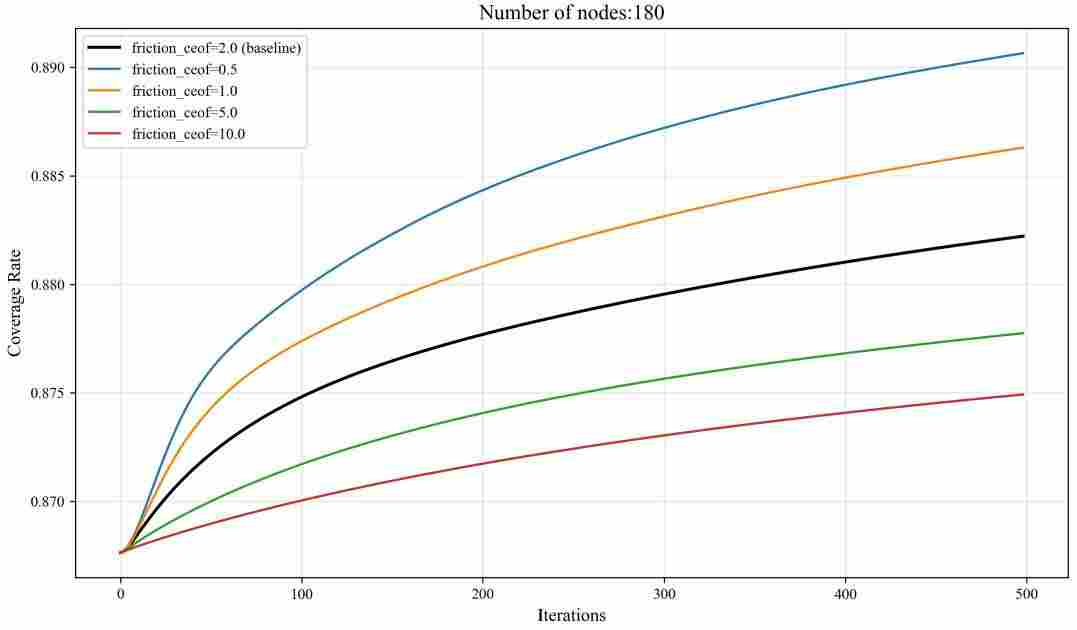}
		\caption{}
	\end{subfigure}
	\caption{Sensitivity analysis of key parameters (Friction coefficient, Attractive force coefficient, and Time step) with different nodes.}
	\label{Fig:11}
\end{figure}

\begin{figure}
	\centering
	\begin{subfigure}[b]{0.14\textwidth}
		\centering
		\includegraphics[width=\textwidth]{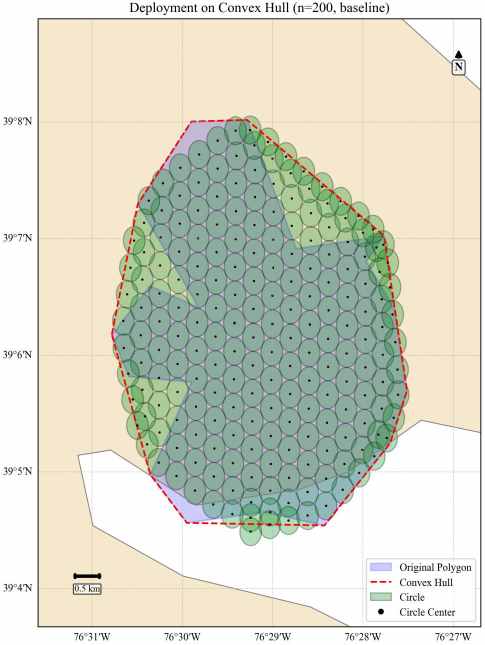}
		\caption{}
	\end{subfigure}
	\begin{subfigure}[b]{0.14\textwidth}
		\centering
		\includegraphics[width=\textwidth]{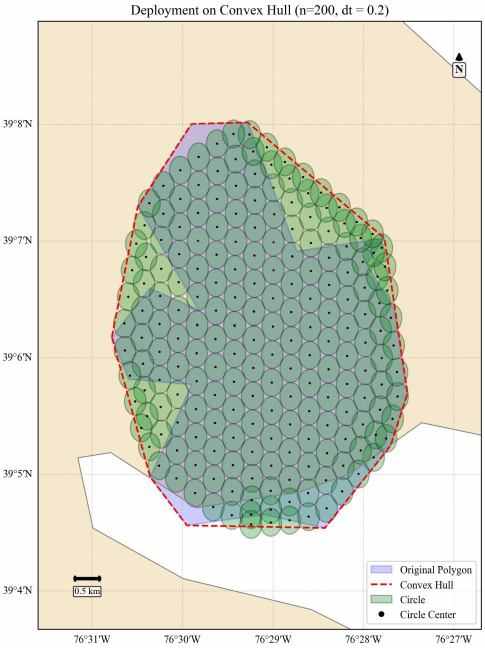}
		\caption{}
	\end{subfigure}
	\begin{subfigure}[b]{0.14\textwidth}
		\centering
		\includegraphics[width=\textwidth]{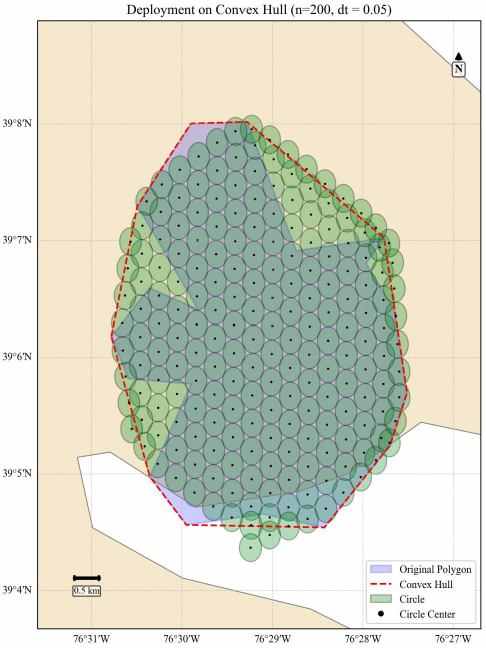}
		\caption{}
	\end{subfigure}
	\begin{subfigure}[b]{0.14\textwidth}
		\centering
		\includegraphics[width=\textwidth]{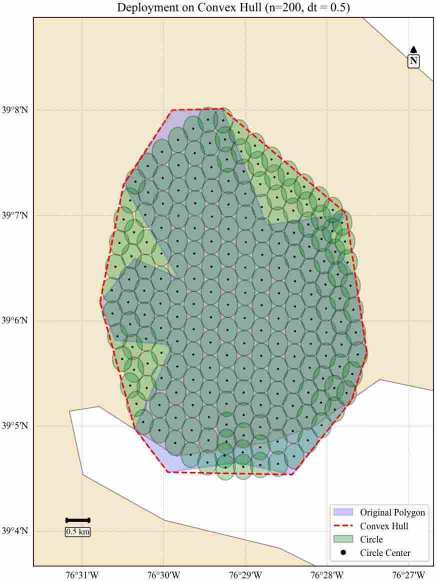}
		\caption{}
	\end{subfigure}
	\begin{subfigure}[b]{0.14\textwidth}
		\centering
		\includegraphics[width=\textwidth]{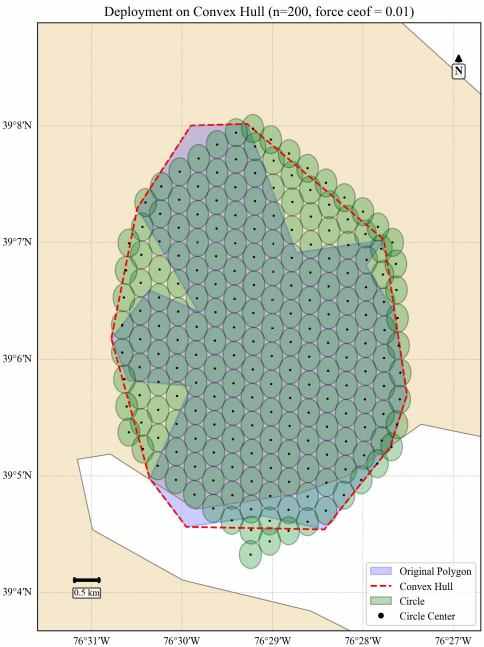}
		\caption{}
	\end{subfigure}
	\begin{subfigure}[b]{0.14\textwidth}
		\centering
		\includegraphics[width=\textwidth]{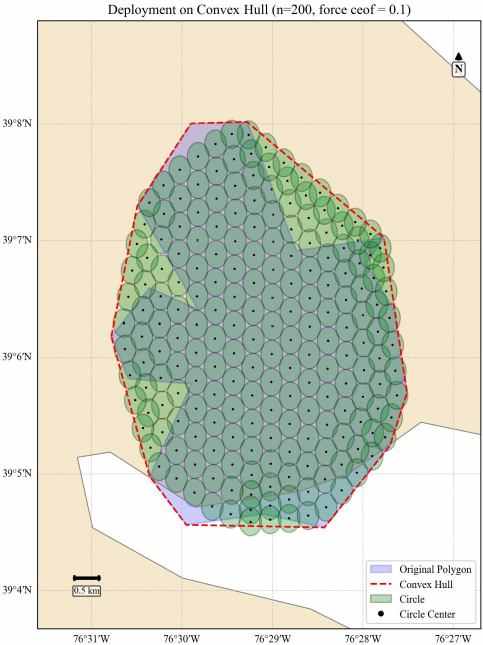}
		\caption{}
	\end{subfigure}
	\begin{subfigure}[b]{0.14\textwidth}
		\centering
		\includegraphics[width=\textwidth]{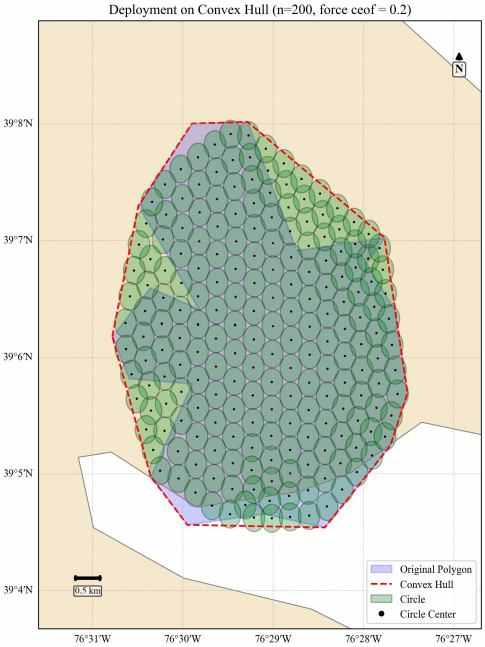}
		\caption{}
	\end{subfigure}
	\begin{subfigure}[b]{0.14\textwidth}
		\centering
		\includegraphics[width=\textwidth]{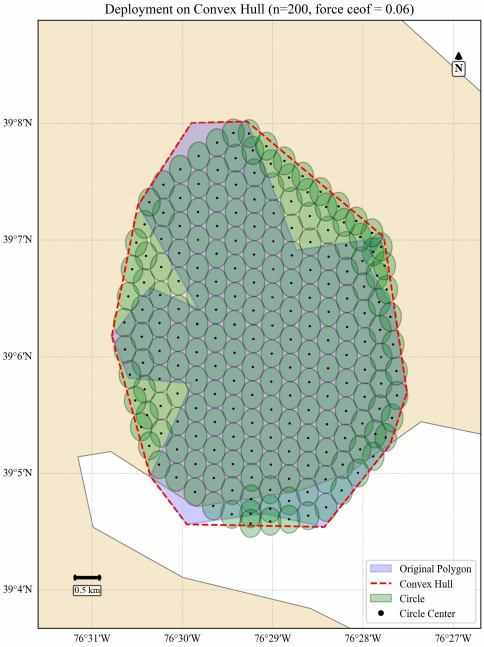}
		\caption{}
	\end{subfigure}
	\begin{subfigure}[b]{0.14\textwidth}
		\centering
		\includegraphics[width=\textwidth]{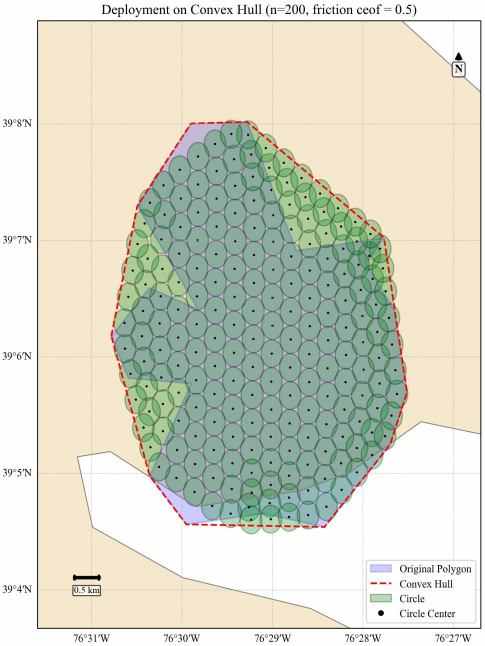}
		\caption{}
	\end{subfigure}
		\begin{subfigure}[b]{0.14\textwidth}
		\centering
		\includegraphics[width=\textwidth]{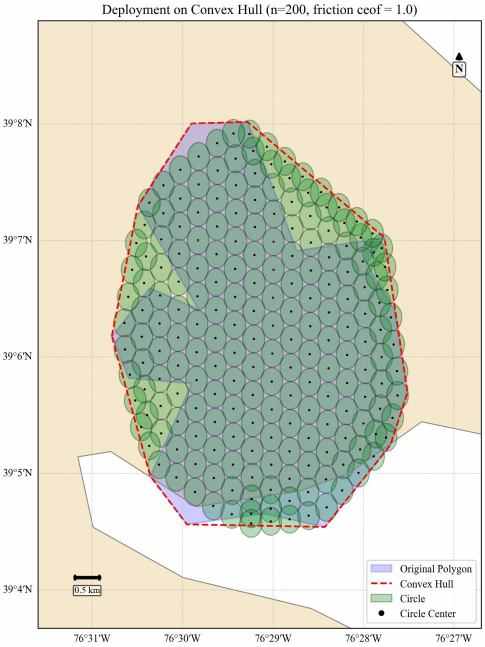}
		\caption{}
	\end{subfigure}
		\begin{subfigure}[b]{0.14\textwidth}
		\centering
		\includegraphics[width=\textwidth]{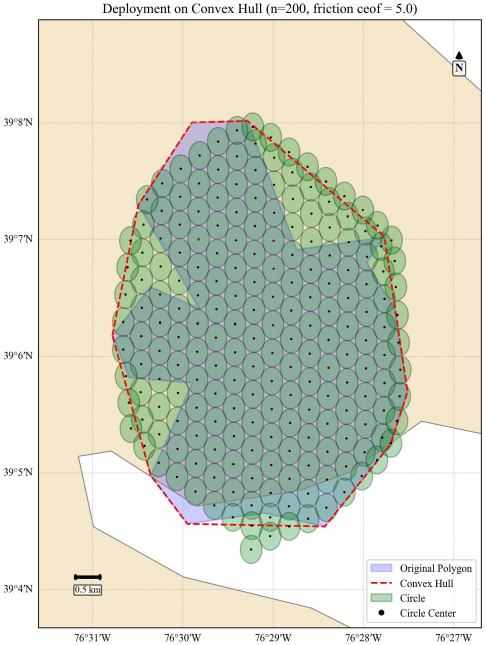}
		\caption{}
	\end{subfigure}
		\begin{subfigure}[b]{0.14\textwidth}
		\centering
		\includegraphics[width=\textwidth]{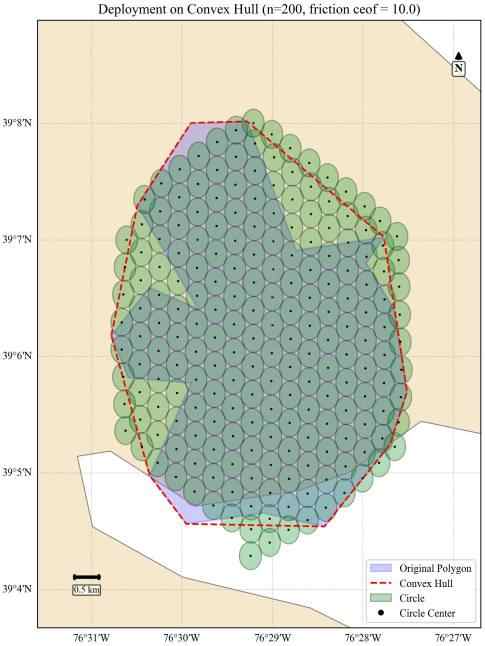}
		\caption{}
	\end{subfigure}
	\caption{Optimized configurations under different parameter combinations for 200 nodes.}
	\label{Fig:12}
\end{figure}

\section{Conclusion and future work}
This paper proposes a fast improved quasi-physical algorithm for solving the optimal wireless coverage in arbitrary convex polygons. The improvements are mainly summarized as follows:

(1) A structure-preserving initialization strategy based on regular hexagonal tiling via scaling and affine transformations is designed, which effectively enhances the quality of the initial solution and reduces the time required for subsequent iterative optimization;  

(2) A virtual force field incorporating friction and a radius-expansion optimization iteration model are constructed, expanding the global search space, reducing the degree of overlap among circles, and improving the quality of the final solution ultimately;  

(3) A boundary-surrounding strategy based on normal and tangential gradients is proposed. This strategy is used to retrieve circles that overflow the polygon boundary, allowing originally overflowing circles to re-enter the interior of the polygon for effective coverage, thereby significantly increasing the usage rate of circles and the coverage rate.  

We guarantee the convergence of the proposed algorithm through rigorous mathematical analysis and provide a thorough computational complexity study. Extensive numerical experiments, conducted on both synthetic data (simulating real-world scenarios) and real-world datasets, demonstrate that our algorithm  outperforms other metaheuristic algorithms consistently in terms of coverage rate and utilization rate that are two critical performance indicators. Notably, on extremely narrow convex domains, IQPD significantly outperforms its competitors while ensuring that all coverage circles remain within the target region. Furthermore, our algorithm substantially reduces optimization time, with runtime second only to VGSOK among all compared methods.

Ablation studies confirm that the integration of the three key components effectively improves coverage. Among these, the friction component exerts the most significant influence on overall coverage; its removal causes the algorithm to lose convergence entirely. The structure-preserving initialization effectively reduces convergence time, while the boundary encircling strategy optimizes the distribution of circles near edges, thereby enhancing overall coverage rate.

Sensitivity analysis reveals that IQPD exhibits low sensitivity to the key parameters \(\mu\), \(\gamma\), and \(\Delta t\) with respect to coverage rate, utilization rate, and minimum gap, demonstrating strong robustness and eliminating the need for complex parameter tuning. The minimum gap consistently remains around \(-0.002\), indicating that the algorithm successfully avoids excessive coverage overlap under all parameter configurations. Computational time and uniformity index are moderately affected. In particular, a larger \(\gamma\) improves distribution uniformity, while a smaller \(\mu\) combined with a moderate \(\gamma\) reduces computational time. The default parameter combination (\(\mu = 2.0\), \(\gamma = 0.03\), \(\Delta t = 0.1\)) achieves a favorable balance between coverage effectiveness and computational efficiency across different node scales (\(n = 160, 180, 200\)).

Although the proposed IQPD algorithm surpasses existing methods in three important metrics (CR,UR,CT), the research in this paper is limited to convex polygons. Real-world scenarios are also involving non-convex regions with holes or concave corners. Therefore, constructing a general circle coverage mechanism applicable to arbitrary two-dimensional planar regions is an important direction for future work. If these issues are successfully addressed, further research could be conducted on coverage optimization problems with dynamic adjustment of circle radii and numbers, optimization of unequal circle coverage, and the densest packing problem in three-dimensional space. It is anticipated that the IQPD algorithm can be applied to practical problems like wireless network coverage. Further research in this direction would contribute valuable solutions to operational optimization and resource allocation.

\newpage

\section*{Acknowledgments}
The work was supported by National Natural Science Foundation of China (Grant No. 12371016, 11871083) and National Key R \& D Program of China (Grant No. 2020YFE0204200).

\bibliographystyle{elsarticle-num-names}
\bibliography{cas-refs}
\end{document}